\documentclass[11pt]{article}

\usepackage{dsfont}
\usepackage{fullpage}
\usepackage{graphicx}%
\usepackage{amsmath,amssymb,amsfonts}%
\usepackage{amsthm}%
\usepackage{mathtools} 
\usepackage{physics}
\usepackage{hyperref}
\usepackage{mathrsfs}%
\usepackage[title]{appendix}%
\usepackage{xcolor}%
\usepackage{textcomp}%
\usepackage{booktabs}%
\usepackage{algorithmicx}%
\usepackage{algpseudocode}%
\usepackage{listings}%
\usepackage[numbers]{natbib}
\usepackage{url}
\usepackage{soul}

\newtheorem{theorem}{Theorem}%
\newtheorem{proposition}[theorem]{Proposition}%
\newtheorem{lemma}[theorem]{Lemma}%

\newtheorem{remark}{Remark}%
\newtheorem{cor}{Corollary}
\newcommand{\DQAOA}[0]{\(\mathcal{D}\)-QAOA}

\newtheorem{definition}{Definition}%

\newcommand{\BBI}[1]{{\color{black} #1}}
\title{Learning Cut Distributions with Quantum Optimization}

\author{
Bao Bach$^{1, 2*}$, Cameron Ibrahim$^{2*}$, Reuben Tate$^{5}$, Jad Salem$^{4}$,  Stephan Eidenbenz$^{5}$, Ilya Safro$^{2,3}$\thanks{Corresponding authors. Emails: \texttt{(baobach,cibrahim,isafro)@udel.edu}}\\[1ex]
$^{1}$Quantum Science and Engineering, University of Delaware, Newark, DE, USA\\
$^{2}$Department of Computer and Information Sciences, University of Delaware, Newark, DE, USA\\
$^{3}$Department of Physics and Astronomy, University of Delaware, Newark, DE, USA\\
$^{4}$Mathematics Department, United States Naval Academy, Annapolis, MD, USA\\
$^{5}$CAI-3: Information Sciences, Los Alamos National Laboratory, Los Alamos, NM, USA
}

\begin{document}
\newpage

\maketitle
\begin{abstract}
The probabilistic nature of quantum computation can be an asset rather than an obstacle for combinatorial optimization. In this work, we tackle the problem in which the sampled distribution itself is the object being optimized. This perspective is especially natural for fairness, coverage, and robustness objectives, where success is determined by the least-protected components rather than by average or best performance. We develop this idea by learning graph cut distributions. The resulting framework turns Fair Cut Cover into a native quantum distribution-learning problem and uses quantum circuits as trainable generators of cut distributions. \BBI{We prove that standard QAOA at depth one achieves higher Fair Cut Cover values than optimal Semidefinite programming plus Gaussian hyperplane rounding on complete graphs with at least four vertices. We further establish a separation on a family of high-girth graphs using published QAOA evaluations, compare against a classical column-generation heuristic in numerical experiments, and examine finite-sampling and noise effects in quantum hardware implementations.} 
These results suggest a broader role for this type of quantum algorithms: not merely finding good bitstrings, but learning useful combinatorial distributions for worst- and average-case-aware optimization.

\end{abstract}

\section{Introduction}

Classical optimization methods are typically designed to find a single feasible solution that minimizes or maximizes a scalar objective function\cite{boyd2004convex, williamson2011design}. However, in some applications, this is not the only quantity of interest. When fairness, robustness, or coverage matter, the relevant question is often not which one solution is best, but whether the worst-served part of the application domain can be protected reliably \cite{daskin1983maximum, bertsimas2006robust}. This naturally shifts attention from individual solutions to distributions over solutions, which can express variability, balance, and protection against worst-case outcomes in a way that any single solution cannot.

This distributional viewpoint is especially relevant in quantum computation, where the output of an algorithm is inherently probabilistic. \emph{Rather than treating this probabilistic output as an obstacle that must be post-processed into one good solution, we treat it as the optimization target itself.} This perspective is particularly appropriate for objectives in which performance is determined by the least-favored component and is therefore naturally defined over distributions rather than individual configurations\cite{garcia2021maxmin}.
It also raises a basic question: \emph{can quantum optimization directly learn useful structured distributions, rather than merely sample candidate solutions from which one keeps the best observed outcome?}

\BBI{
Prior work in quantum generative modeling has explored whether parameterized quantum circuits can learn or engineer useful probability distributions \cite{liu2018differentiable, gili2023quantum}. Here, we consider this question in a different setting by studying combinatorial optimization problems in which the goal is to learn a distribution that maximizes a worst-case criterion over the underlying discrete structure. Our setting is specifically inspired by the maximin fairness paradigm  \cite{salem2024expected}. Such problems provide a natural setting for comparing classical methods, which typically return one solution together with an approximation guarantee, with quantum methods, which manipulate and measure distributions natively. Such distributions can be advantageous over integral solutions in iterative decision-making scenarios, as one can sample a solution at each iteration. In this sense, a maximin-optimal distribution can ensure a notion of {\it fairness over time.} For instance, if one is tasked with placing a set of voter registration booths (or some other pop-up facilities) each day, then a maximin-optimal distribution will ensure that no neighborhood is routinely distant from a booth.
}

\BBI{
We make this concrete in the setting of graph cuts, where distributions over cuts offer a simple but expressive model for studying fairness-oriented optimization objectives. Each cut distribution naturally defines, for each edge, a probability that that edge is cut.
These probabilities encode rich structural information, which we can use to explore how robustly different parts of the graph are separated under a given distribution. Therefore, optimizing over cut distributions opens the door to fairness, coverage, and robustness objectives that are uniquely suited to the probabilistic nature of quantum computing.
}

\begin{figure*}
    \centering
    \includegraphics[width=0.9\linewidth]{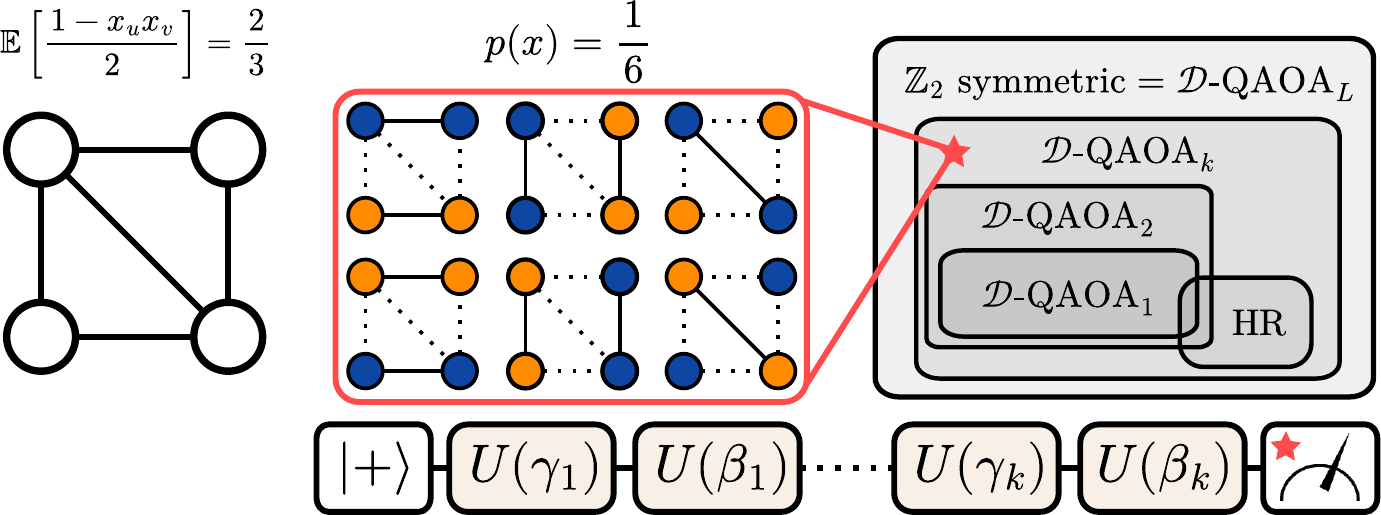}
    \caption{\BBI{The Hierarchy of $\mathbb{Z}_2$-symmetric Cut Distributions supported by \(\mathcal{D}\)-QAOA and Hyperplane Rounding.}
    Given a graph \(G\) (left), one can construct a \(\mathcal{D}\)-QAOA circuit from which to sample cuts on \(G\) (bottom-right). Each layer added to the circuit expands the space of representable distributions (top-right). The optimal cut distribution for fair cut cover (red) is representable after \(k\) layers, and consists of 6 possible cuts with uniform probability (equal normalized weight among cuts), yielding an edge cut probability of \(2/3\) for each edge. Every \(\mathbb{Z}_2\)-symmetric 
    cut distribution on \(n\) vertices is representable in \(L\) layers, and the space of hyperplane rounding distributions (HR) will also eventually be entirely contained in this space.
    } 
    \label{fig:distributionhierarchy}
\end{figure*}

\paragraph{Graph problems over cut distributions}
We explore this idea through the lens of the Fractional Cut Cover problem \cite{vsamal2015cubical, neto2019fractional}. 
Given all cuts, the  Cut Cover problem asks for the smallest set of cuts (a cut is a subset of edges that partitions graph vertices into two disjoint sets 
) such that every edge $e \in E$ is cut at least once. Its linear relaxation is the  Fractional Cut Cover problem, which assigns non-negative weights to all cuts so that the total weight is minimized and each edge is covered by cuts with a total weight of at least one \cite{neto2019fractional}. An example of this is given in Figure \ref{fig:distributionhierarchy}. Fractional Cut Cover is closely related to the geometry of the cut cone and has recently been shown to be APX-Complete by \cite{benedetto2024generalized}. The authors provide an algorithm based on semidefinite programming (SDP) and hyperplane rounding, which achieves the best possible polynomial-time approximation ratio under the Unique Games Conjecture \cite{benedetto2025primal}.

A solution to Fractional Cut Cover can be interpreted probabilistically: after normalizing the cut weights, it becomes a distribution over cuts. Under this interpretation, the value of the solution is determined by the edge that is least likely to be cut, so the problem becomes one of maximizing the minimum edge-cut probability. This recasts Fractional Cut Cover as a maximin optimization problem over distributions rather than individual cuts, making it a natural model for fairness-oriented objectives in which worst-case performance is the quantity of interest. It also makes the problem an attractive testbed for our setting, because it combines a clear probabilistic interpretation with strong geometric and algorithmic structure, as it is both structurally rich (via its connections to Grothendieck symmetric inequalities \cite{friedland2020symmetric} and theta functions \cite{karger1998approximate}) and algorithmically sharp (via SDP with randomized rounding, and Unique Game Conjecture tight hardness results \cite{benedetto2024generalized}). \BBI{The focus of this manuscript is on cut distributions, but our approach is still relevant to general distributional or fractional graph problems \cite{scheinerman2011fractional}. As a concrete example, we show how to use the QAOA ansatz to solve Fractional Coloring in Appendix~\ref{appendix:fractional_coloring}.}

\paragraph{Summary of the results}

\BBI{We study the Fractional Cut Cover problem through the lens of cut distributions, in which the objective is to maximize the minimum edge-cut probability over a distribution of cuts. We establish a sharp contrast between the best known polynomial-time worst-case classical approach and a quantum variational algorithm. On the classical side, we analyze the semidefinite programming (SDP) relaxation together with hyperplane rounding, clarifying the structural limitations of this pipeline for optimizing the fair-cut objective. On the quantum side, we formulate the problem as a distribution-learning task with a QAOA ansatz, where the resulting quantum state $\ket{\psi(\beta,\gamma)}$ induces a cut distribution chosen to maximize the minimum edge cut probability.

Our main findings show that the quantum approach is not just an alternative view, but provably results in better performance on important graph families. For every complete graph \(K_n\) with \(n\ge4\), we prove that standard QAOA at depth one achieves a strictly larger minimum edge-cut probability than Gaussian hyperplane rounding can attain, even after optimizing over all unit-vector embeddings. We also derive a spectral upper bound on hyperplane-rounding performance and combine it with QAOA locality and published numerical evaluations to obtain a separation on an explicit infinite family of high-girth cubic graphs. Furthermore, given enough layers, $\mathcal{D}$-QAOA (introduced later) can realize $\mathbb{Z}_2$-symmetric distribution exactly. These results identify limitations of the hyperplane-rounding distribution family and demonstrate that QAOA can generate cut distributions with better Fair Cut Cover value.

We complement these theoretical results with numerical comparisons against a distribution-aware column-generation heuristic and with hardware experiments examining finite-sampling and implementation effects. These studies investigate how circuit expressibility, initialization, and noise influence the quality of learned cut distributions, and motivate further exploration of quantum optimization for distributional combinatorial problems.

More broadly, our results position near-term quantum optimization as a natural framework for learning combinatorial distributions that are difficult to access through classical relaxations and rounding alone. Rather than serving only as approximate solvers for scalar objectives, variational quantum circuits can act as expressive generators of optimal discrete distributions, potentially opening a new route to provable quantum advantage in optimization.}

\section{Result}

Let $G=(V, E)$ be an undirected graph 
. A cut over graph $G$ can be represented by a spin assignment on nodes $x \in \{\pm1\}^{\abs{V}}$. We define the cuts indicator matrix $Y \in \{0, 1\}^{\abs{E} \times 2^{\abs{V}}}$ where, for each edge $e=(u,v)$ and cut $x$,
\begin{align}
    Y_{ex} = \mathds{1}[x_{u} \neq x_{v}] =\frac{1 - x_{u}x_{v}}{2}.
\end{align}
Here $\mathds{1}[x_{u} \neq x_{v}]$ denotes the indicator if edge $(u, v)$ is cut by cut $x$. 
The \emph{Cut Cover Problem} is to find the minimum number of cuts in a graph such that each edge spans at least one of the cuts \cite{cutcover, loulou1992minimal}. 
Its linear relaxation, \emph{Fractional Cut Cover} \cite{neto2019fractional}, is the problem on graph $G$ of assigning a non-negative weight %
$w_x \geq 0$ to each cut $x$ %
so that every edge is covered by cuts with a total weight of at least one, while minimizing the total weight of all cuts:
\begin{align}
    \label{eq:fcc}
    \eta(G) = \min_{w \geq 0} \langle w, \mathbf{1} \rangle\quad \text{s.t.}\quad Yw \geq \mathbf{1}_{E}.
\end{align}
Normalizing $w$ produces a probability distribution $p_x = w_x/\langle w, \mathbf{1} \rangle$ on all cuts, 
which we refer to as a \emph{cut distribution}. For any cut distribution $p$ and edge $e$, the \emph{edge cut probability} is $\mathrm{Pr}_{e}(p) = \langle Y_{e},p\rangle$.
This leads to the maximin \emph{fair-cut cover} objective: learn a cut distribution that maximizes the
smallest edge cut probability,
\begin{align}
    \label{eq:fcc_fair}
    \bar{\eta}(G) = \max_{p}\min_{e \in E} \langle Y_{e},p\rangle = \frac{1}{\eta(G)}.
\end{align}

As such, Fractional Cut Cover and Fair Cut Cover are equivalent, as their objective are reciprocal (see Appendix~\ref{appendix:equiv}). %
From Equation (\ref{eq:fcc_fair}), it can be observed that the Fair Cut Cover problem is monotonic, where  $\bar{\eta}(G) \leq \bar{\eta}(H)$ for every subgraph $H \subseteq G$.

\BBI{For comparison, MaxCut seeks a single cut maximizing the total number of edges cut 
\begin{align}
    \operatorname{MaxCut}(G)
=
\max_{x\in\{\pm1\}^{|V|}}
\sum_{e\in E}Y_{ex}~.
\end{align}
Fair Cut Cover instead concerns the coverage of individual edges under repeated sampling. Consider a triangle: every maximum cut cuts two edges and leaves the third uncut. Always selecting the same maximum cut therefore gives one edge zero coverage probability. Sampling uniformly among the three nontrivial cut partitions gives every edge cut probability two-thirds. Both strategies achieve the same expected cut size, but their minimum edge-cut probabilities differ. Thus, maximizing expected cut size alone does not determine how a distribution balances coverage across edges.} 

\BBI{The SDP relaxation provides an upper bound on the optimal Fair Cut Cover value, while hyperplane rounding produces a cut distribution whose achieved fairness serves as a classical benchmark. For background on semidefinite relaxations and hyperplane rounding, we refer readers to Goemans and Williamson~\cite{goemans1995improved}. Their connection to fractional cut covering, including the role of gauge and antiblocker duality, is developed in ~\cite{benedetto2024generalized, benedetto2025primal}; Appendices~\ref{appendix:equiv},~\ref{appendix:SDP} and ~\ref{appendix:sdp_rounding} summarize the formulations and properties used here.}

The max-min viewpoint is a special case of the distributional fairness objectives framework studied in~\cite{salem2024expected, hojny2025}, which is introduced in the Method Section.
It has recently been shown in \cite{benedetto2024generalized, benedetto2025primal} that Fractional Cut Cover is APX-Complete. Furthermore, under the Unique Games Conjecture, the SDP plus random-hyperplane rounding used by Goemans-Williamson \cite{goemans1995improved, khot2007optimal} for MaxCut, also yields the best possible polynomial-time approximation algorithm for Fractional Cut Cover. Consequently, it is shown in  \cite{benedetto2025primal} how SDP with randomized rounding gives the Goemans-Williamson coefficient approximation ratio, $\alpha_{\mathrm{GW}} \approx 0.878$, for Fractional Cut Cover, \BBI{which transfers directly to Fair Cut Cover as follows}:
\BBI{\begin{align}
    \alpha_{\mathrm{GW}}\mathrm{SDP}(G) \leq \mathrm{SDP}_{\mathrm{HR}}(G) \leq \bar{\eta}(G) \leq \mathrm{SDP}(G), 
\end{align}}
where $\mathrm{SDP}_{\mathrm{HR}}(G)$ denotes the optimal value yielded by SDP + Hyperplane Rounding and $\mathrm{SDP}(G)$ denotes the optimal value of SDP.

Details on the formulation and characteristics of the SDP are mentioned in Appendix \ref{appendix:SDP}. Nevertheless, we show that the space of cut distributions derivable from the randomized hyperplane rounding approach fails to capture all possible cut distributions, for example, the optimal fair-cut distribution of the complete graph $K_{n}$ where $n \geq 4$, as discussed in Appendix~\ref{appendix:sdp_rounding}%
. This drawback is stated in Theorem \ref{thm:dist_subset}.

\paragraph{Quantum optimization for Fractional Cut Cover}
The Fractional 
Cut Cover can be represented as the Fair Cut Cover problem on a graph $G(V, E)$, where it seeks a distribution over cuts, $p$, that maximizes the \emph{smallest} edge cut probability as shown in Equation (\ref{eq:fcc_fair}). Any quantum circuits, representing quantum state $\ket{\psi}$, natively induce such distributions. Concretely, measuring each qubit of state $\ket{\psi}$ in the computational ($Z$) basis produces a random spin assignment $x\in\{\pm1\}^{|V|}$ with probability $\mathrm{Pr}_{x}(\ket{\psi}) = \abs{\braket{x}{\psi}}^{2}$.
For any edge $e = (u, v) \in E$, the probability that $e$ is cut under this induced distribution is 
\begin{equation}
    \begin{aligned}
        \mathrm{Pr}_{e}(\ket{\psi}) &= \sum_{x\in\{\pm1\}^{|V|}} \mathrm{Pr}_{x}(\ket{\psi})\frac{1-x_u x_v}{2}\\
        &=\frac{1 - \bra{\psi}Z_{u}Z_{v}\ket{\psi}}{2}
    \end{aligned}
\end{equation}
Therefore, for a variational family of quantum states  $\ket{\psi(\theta)}$ generated from an ansatz, the Fair Cut Cover problem can be framed as a variational quantum optimization problem,
\begin{align}
    \label{eq:quantum_fair_objective}
    \mathrm{Q}(G) = \max_{\theta}\min_{e \in E} \frac{1 - \bra{\psi(\theta)}Z_{u}Z_{v}\ket{\psi(\theta)}}{2}.
\end{align}
\BBI{Fair Cut Cover retains the binary assignments and two-qubit correlations familiar from MaxCut while making the output distribution itself the optimization target. Each measured bitstring provides a sample of every edge-cut indicator, so the objective can be estimated without reconstructing the full distribution. Together with its established classical formulations, this makes Fair Cut Cover a useful test problem for studying how different families of sampling distributions affect fairness.}
This formulation is closely related to the Quantum Approximate Optimization Algorithm (QAOA) \cite{farhi2014quantum} for the MaxCut problem. In standard QAOA for MaxCut, the objective is to find the distribution over cuts that maximizes the \emph{average} edge cut probability, $\max_{\theta} \sum_{(u, v) \in E} [1-\bra{\psi(\theta)}Z_{u}Z_{v}\ket{\psi(\theta)}]/2$, where $\ket{\psi(\theta)}$ is the quantum state produced by QAOA.
Regardless of the similarities, the QAOA scheme is not directly applicable to the Fractional Cut Cover due to the non-linearity of the min function. However, the QAOA ansatz remains a natural variational family for this problem (edge-cut
probability), with the objective changed from an average to a max-min criterion. The change therefore concerns the classical
objective evaluated from the measurement outcomes.\\

\emph{Standard QAOA ansatz:} For edge-transitive graphs, \cite{shaydulin2021exploiting} shows that for every edge $(u,v)$, the expectation value $\langle Z_{u}Z_{v} \rangle$ are equal, regardless of number of layers. 
This implies that MaxCut and Fair Cut Cover naturally share the same objective $\max_{\ket{\psi}} (1-\bra{\psi}Z_{u}Z_{v}\ket{\psi})/2$ for these edge-transitive graphs. 
This relationship effectively reflects the relationship between MaxCut and Fractional Cut Cover on edge-transitive graphs through both the SDP solver and the quantum solver, as discussed in Appendix \ref{appendix:SDP}.\\

\emph{Multi-angle QAOA:} For a general case, it is more effective to consider the multi-angle QAOA (ma-QAOA) ansatz \cite{herrman2022multi} due to its expressibility. To show the connection between reachable states (circuit expressibility) and the obtained distribution, we note that the Trace distance between states is the generalization of the total variation distance between distributions, denoted as \(\norm{p - q}_{\mathrm{TV}}\), through contractivity \cite{nielsen2010quantum}.

The expressibility of ma-QAOA is such that with sufficient layers, every state of the \(+1\) eigenspace of \(X^{\otimes \lvert V \rvert}\) is reachable by an ma-QAOA circuit initialized on a majority of graph structures \cite{kazi2025analyzing}. 
This has significant implications for the space of \(\mathbb{Z}_2\)-symmetric cut distributions, that is, the set of cut distributions \(p\) such that \(\forall x \in \{\pm 1 \}^n, p(x) = p(-x).\) In  Appendix \ref{appendix:DQAOA_epxressibility}, we outline how any \(\mathbb{Z}_2\)-symmetric cut distribution can be embedded in this eigenspace and show some properties of the given ansatz when considering FCC problem in Appendix~\ref{appendix:ma_QAOA_prop}. \BBI{We also note that a similar result can be found for distributions that do not preserve \(\mathbb{Z}_2\)-symmetry, as a concrete example in Appendix ~\ref{appendix:fractional_coloring} and \cite{tsvelikhovskiy2026expressive}, we show that QAOA with a constraint-preserving mixer for Fractional Coloring problem in which we can capture any distribution over independent sets.}

It is desirable for our QAOA ansatz to achieve full expressibility for an arbitrary input graph, so we introduce the following minor modification to ma-QAOA, which we call distributional QAOA or \(D\)-QAOA.
\begin{definition}
    \label{def:d_QAOA}
    A \(\mathcal{D}\)-QAOA circuit for a graph \(G = (V, E)\) is defined as a \textbf{multi-angle QAOA} circuit on G with the following modification:
    \begin{enumerate}
        \item If \(G\) is disconnected or if \(G\) is a path/cycle with \(\lvert V \rvert\geq 4,\) add an ancilla qubit \(A\) and add \(X_{A}\) to the mixer unitary. Choose one vertex \(v\) with maximal degree for each connected component and add \(Z_vZ_{A}\) to the phase separator unitary.
        \item If \(G\) is disconnected and \(\lvert V \rvert \leq 4,\) do (1), then add an additional ancilla \(B\) and add \(Z_{A}Z_{B}\) to the phase separator unitary and add \(X_{B}\) to the mixer unitary.
        \item Let $H$ denote the graph obtained after applying steps (1) and (2). If $H$ is bipartite with bipartition $V(H)=L \sqcup R$, where both $\abs{L}$ and $\abs{R}$ are even, add an ancilla qubit $C$ and connect it to any vertex $v\in V(H)$. Add Pauli $X_C$ to the mixer generators and $Z_vZ_C$ to the phase-separator generators, each with an independent variational parameter in every layer.
    \end{enumerate}
    
\end{definition}
Notably, for a majority of graph structures, \(\mathcal{D}\)-QAOA is equivalent to ma-QAOA. Furthermore, the fair cut solution of the original graph can be obtained by tracing out the information of the ancilla added by \DQAOA{}; this is analogous to discarding the bit corresponding to the ancilla whenever a bitstring is sampled.
We show in Appendix \ref{appendix:DQAOA_epxressibility} that \(D\)-QAOA produces an ansatz with a sufficient amount of expressability to introduce the following corollary based on the results of \cite{kazi2025analyzing}.

\begin{cor}
\label{thm:sufficient_p_covering_distributions}
    For every \(G=(V,E)\) and every $\epsilon \geq 0$, there exists a depth $k^\prime$ such that for every \(k \geq k^\prime\) and every $\mathbb Z_2$-symmetric cut distribution $p$ on $\{\pm1\}^{\abs{V}}$, there exists a $k$-layer \(\mathcal{D}\)-QAOA state $\tilde{\rho}$ such that for any \(x \in \{\pm1\}^{\abs{V}}\),
\begin{equation}
        \tilde p(x) \coloneq \Tr{\ketbra{x}{x}\tilde{\rho}}, \text{ and } \norm{\tilde p-p}_{\mathrm{TV}} \leq \epsilon.
\end{equation}    
\end{cor}

We denote the Fair Cut Cover value of the \(k\)-layer \DQAOA{} circuit with parameters \(\gamma, \beta\) as \(Q_k(G, \gamma, \beta)\) and \(Q_k(G) = \max_{\gamma, \beta} Q_k(G, \gamma, \beta)\). We then define \(Q_k^{\mathrm{std}}(G)\) on the original graph $G$, but allow for only two shared angles per layer, one cost and one mixer angle, as with standard QAOA. We can define a conditional monotonicity property for \DQAOA{} as follows.

\begin{cor} 
    \label{cor:sufficient_p_covering_distributions}
    For any graph \(G\), there exists a \(k^\prime\) such that for every \(k \geq k^\prime\) and every subgraph \(H \subseteq G\) we have
    \begin{align}
        \mathrm{Q}_{k}(G) \leq \mathrm{Q}_{k}(H).
    \end{align}
\end{cor}

Moreover, for problems that respect $\mathbb{Z}_{2}$ symmetry, such as Fair Cut Cover or MaxCut, any cut distribution can be represented by a $\mathbb{Z}_{2}$ symmetric cut distribution; this is discussed in Appendix \ref{appendix:equiv}.

\paragraph{Separation gap between quantum and hyperplane rounding}
As the reciprocal of Fractional Cut Cover, we still use SDP plus randomized hyperplane rounding for Fair Cut Cover, where details of the concrete SDP formulation are discussed in Appendix \ref{appendix:SDP}.\\
We first show that Gaussian hyperplane rounding from optimal SDP embedding cannot represent every symmetric cut distribution and can therefore fail to attain the optimal Fair Cut Cover value.
\begin{theorem}
    \label{thm:dist_subset}
    Let \(\mathcal{D}\) be the space of $\mathbb{Z}_{2}$-symmetric cut distributions of a graph with node size at least $4$ and \(\mathcal{D}_{\mathrm{HR}}\) be the space of cut distributions corresponding to hyperplane rounding. Then \(\mathcal{D}_{\mathrm{HR}} \subsetneq \mathcal{D}.\)
\end{theorem}
The proof of Theorem \ref{thm:dist_subset} is mentioned in Appendix \ref{appendix:sdp_rounding}. Regardless of the drawback, SDP embeddings followed by hyperplane rounding still provide a principled baseline for our quantum results. Importantly, it does \emph{not} preclude instance-wise improvements: a method can outperform $\mathrm{SDP}_{\mathrm{HR}}(G)$ on structured graph families without contradicting UGC-tightness, which concerns the best possible approximation ratio over \emph{all} graphs. As a concrete example, we compare our QAOA-based solver against the classical approach, yielding the following theorem. 
\BBI{
\begin{theorem}
    \label{thm:Kn_separation}
    For any complete graph \(K_n,\) with $n \geq 2$
    \[
    \mathrm{Q}_1^{\mathrm{std}}(K_n) \geq \mathrm{SDP}_{\mathrm{HR}}(K_n) = \frac{1}{\pi}\arccos\!\Big(\frac{1}{1-n}\Big).
    \]
    where strict inequality occurs when $n \geq 4$.
\end{theorem}
    }
Here the $\mathrm{SDP}(G)$ for Fair Cut Cover with the positive semi-definite matrix $C$, is defined as follow:
\begin{equation}
    \label{eq:fair_cut_cover_SDP}
    \begin{aligned}
        \mathrm{SDP}(G) \coloneq& \max_{C} (t) \\ 
    \text{s.t: }&  \frac{w_{ij}}{2}(1 - C_{ij}) \geq t\mathbf{1}_{E}, \forall ij \in E \\
                & C \succeq 0, \qquad \text{diag}(C) = \mathbf{1}_{V} 
    \end{aligned}
\end{equation}
where $w_{ij}$ denotes weight of edge $(ij)$. The proof of Theorem \ref{thm:Kn_separation} and details regarding hyperplane rounding are described in Appendix \ref{appendix:sdp_rounding}. Here, both separations in Theorem \ref{thm:dist_subset} and \ref{thm:Kn_separation} arises because SDP + hyperplane rounding restricts edge cut probabilities to the elliptope image $\arccos{\langle v_{u}, v_{v}\rangle}/\pi$ as stated in the proof of Theorem \ref{thm:dist_subset} in Appendix \ref{appendix:sdp_rounding}
, whereas QAOA-induced distributions are not confined to this map (but effectively reach all possible $\mathbb{Z}_2$-symmetric cut distributions as stated in Corollary \ref{thm:sufficient_p_covering_distributions} with enough layers). This result applies to any Gaussian hyperplane rounding of one fixed unit-vector embedding, as this SDP formulation provides the optimal embedding for hyperplane rounding for this problem, stated in Proposition \ref{prop:sdp_optimal_HR}, Appendix \ref{appendix:sdp_rounding}.\\ 
\color{black}
The second characteristic of the SDP solver is monotonicity, where $\mathrm{SDP}(G) \leq \mathrm{SDP}(H)$ for any 
subgraph $H \subseteq G$, discussed in Appendix \ref{appendix:SDP}. As the optimal value for the fair-cut cover SDP can be found for many classes of edge-transitive graphs \cite{karger1998approximate}, especially the complete graph family, we derive Corollary \ref{cor:non_trivial_upperbound}, which yields a non-trivial upper bound for any graph.
\begin{cor}
    \label{cor:non_trivial_upperbound}
    For any connected graph \(G(V, E)\)
    \[
             \frac{1}{\pi}\arccos\left(\frac{1}{1 - \abs{V}}\right)  \leq \mathrm{SDP}_{\mathrm{HR}}(G) \leq \frac{1}{\pi}\arccos\left(\frac{1}{1 - \omega(G)}\right)
    \]
    where \(\omega(G)\) is the size of the maximum clique in \(G.\)
\end{cor}

Here, both inequalities are based on the subgraph monotonicity of the SDP solver \BBI{and the relationship between the SDP solver and Gaussian hyperplane rounding}, which is discussed in detail in Appendix \ref{appendix:sdp_rounding}. Based on this monotonicity characteristic, both the SDP and the Fair Cut Cover optimal values are upper bounded by the largest clique with size $\omega(G)$ (subgraph) and lower bounded by the clique with size $\abs{V}$ (supergraph).
The family of complete graphs $K_{n}$ with $n \geq 4$ gives a good example of gap, but it is not the only case. 
Here we report several numerical results for highly symmetric graphs with maximum clique size at most $3$, beyond the complete-graph case, where we observe a separation among the optimal Fair Cut Cover $\bar{\eta}(G)$, the best value obtained by SDP with hyperplane rounding distribution $\mathrm{SDP}_{\mathrm{HR}}(G)$, and the value achieved from $k-$layers standard QAOA-ansatz,  $\mathrm{Q}^{\mathrm{std}}_{k}(G)$, as summarized in Table \ref{tab:highly_symmetric_graphs_qaoa}. 

\begin{table*}
\begin{tabular}{|l|l|l|l|l|l|l|l|l|}
\hline
\textbf{Graph} $G(V, E)$            & $\bar{\eta}(G)$ & $\mathrm{SDP}_{\mathrm{HR}}(G)$  & $\mathrm{Q}^{\mathrm{std}}_{1}(G)$ & $\mathrm{Q}^{\mathrm{std}}_{2}(G)$ & $\mathrm{Q}^{\mathrm{std}}_{3}(G)$ & $\mathrm{Q}^{\mathrm{std}}_{4}(G)$ & $\mathrm{Q}^{\mathrm{std}}_{5}(G)$ \\ \hline
Petersen Graph    &  4/5  &  0.7323  &   0.6925      &  $\geq 0.7404$   & $\geq 0.7816$ & $\geq 0.7958$ & $\geq 0.7980$ \\ \hline
Clebsch Graph       &   4/5   &  0.7048 & 0.6431  & $\geq 0.6985$  & $\geq 0.7275$  & $\geq 0.7618$ & $\geq 0.7803$  \\ \hline
Paley Graph $q=13$ &    2/3  & 0.6254 & 0.5957  & $\geq 0.6280$ & $\geq 0.6443$ & $\geq 0.6542$ & $\geq 0.6616$ \\ \hline
Paley Graph $q=17$ &   11/17   & 0.6037 & 0.5765 & $\geq 0.6070$ & $\geq 0.6184$ &  $\geq 0.6240$ & $\geq 0.6284$  \\ \hline
Shrikhande Graph &    2/3   &  0.6082 & 0.5957  & $\geq 0.6206$  & $\geq 0.6302$  & $\geq 0.6338$  & $\geq 0.6382$  \\ \hline
\end{tabular}

\caption{Results of classical and quantum methods on Fair Cut Cover for highly symmetric graphs with max cliques size at most $3$. Here $\bar{\eta}(G)$ denotes the optimal Fair Cut Cover value, $\mathrm{SDP}_{\mathrm{HR}}$ denotes best value obtained by SDP with hyperplane rounding, and $Q^{\mathrm{std}}_{k}$ denotes values achieved by $k$-layers \textbf{standard QAOA} ansatz. For every cases, standard QAOA ansatz with $2-3$ layers can beat the SDP with hyperplane rounding results. \BBI{The first two graphs have max-clique size of $\omega(G) = 2$ with trivial upper-bound $\mathrm{SDP}_{\mathrm{HR}}\leq \frac{1}{\pi}\arccos\left(\frac{1}{1-2}\right) = 1$, whereas the remaining graphs have $\omega(G)=3$ with upper-bound $\mathrm{SDP}_{\mathrm{HR}}\leq \frac{1}{\pi}\arccos\left(\frac{1}{1-3}\right) = 2/3$.}
}
\label{tab:highly_symmetric_graphs_qaoa}
\end{table*}

\BBI{Another interesting separation comes from \cite{basso2021quantum}, when considering high-girth $(D+1)$-regular graphs where the girth is greater than $2p + 1$. In this case, each expectation value $\langle Z_{u}Z_{v} \rangle$ of the QAOA at depth $p$ only sees the same subgraph of the $D$-ary tree regardless of which edge $uv$ is considered. Therefore, this equality holds regardless of the choice of angles within standard QAOA. Consequently, in this set of graphs, the result of QAOA on MaxCut directly transfers to the Fair Cut Cover objective. Based on this observation and the numerical results presented by \cite{farhi2025lower}, we give a separation argument between QAOA and SDP for high-girth graphs.

\begin{theorem}[High-girth separation from hyperplane rounding]
    \label{thm:high-girth-fcc-separation}
For every $d$-regular graph $G$ with $d\ge1$ and adjacency matrix $A_G$,
\begin{equation}
     \mathrm{SDP}_\mathrm{HR}(G)
     \le \frac{1}{\pi}
     \arccos\left(\frac{\lambda_{\min}(A_G)}{d}\right)
\end{equation}
There is an explicit infinite family of connected,
non-bipartite cubic Ramanujan graphs $G_r$, indexed by even
integers $r\ge16$, with $g(G_r)\ge2r$, such that we obtain the upper bound of the hyperplane rounding approach on the family of graphs $G_r$ as follows:
\begin{equation}
 \mathrm{SDP}_\mathrm{HR}(G_r) \leq \frac{1}{\pi}\arccos\left(-\frac{2\sqrt{2}}{3}\right).
\label{eq:fcc-hg-constant}
\end{equation}
Furthermore, using the depth-$15$ numerical evaluation
reported in \cite[Table~1]{farhi2025lower}, this family
satisfies
\begin{equation}
    Q^{\mathrm{std}}_{15}(G_r)
    > \frac{1}{\pi}\arccos\left(-\frac{2\sqrt{2}}{3}\right)
    \ge \mathrm{SDP}_{\mathrm{HR}}(G_r)
\end{equation}
for every even $r\ge 16$.
\end{theorem}
The proof of Theorem \ref{thm:high-girth-fcc-separation} can be found in Appendix~\ref{appendix:sdp_rounding}. Interestingly, Figure 2 in \cite{farhi2025lower} which is a comparison between standard QAOA and the Thompson-Parekh-Marwaha vector algorithm \cite{thompson2021explicit}, can also be directly used as an example of how QAOA performs on Fair Cut Cover on high-girth regular graphs over its classical counterpart. Here, the TPM algorithm explicitly optimizes over its particular vector family in which the inner product between the vectors of each edge is the same. Therefore, optimizing over this vector assignment consequently optimizes both the MaxCut and Fair Cut Cover objective. This gives the results of \cite{farhi2025lower} as another example of potential advantage of QAOA for this problem. }

\paragraph{Sampling complexity to approximate edge cut probability}
Given an unweighted graph $G=(V, E)$ and $p \in \mathcal{D}$ be any distribution over cuts, based on the celebrated Hoeffding's inequality \cite{hoeffding1963probability}, for any edge $e\in E$, we have the sample complexity to approximate the edge cut probability $\mathrm{Pr}_{e}(p)$.
\begin{theorem}
    \label{prop:sampling_hoeffding}
    Let \(\hat p\) be the empirical distribution constructed from \(T\) i.i.d. samples drawn from \(p.\)
    For any $\epsilon \in (0, 1)$, achieving a maximum \textbf{absolute error} of 
    \begin{align}
        \lVert  Y(\hat p - p )\rVert_\infty \leq \epsilon
    \end{align}
    with a probability of at least \(1 - \delta\) suffice with \(\mathcal{O}(\epsilon^{-2}\log(2\lvert E\rvert/\delta))\) samples.
    
    Achieving a maximum \textbf{relative error} of 
    \begin{align*}
        \max_e \left\lvert\frac{\langle Y_e, \hat p - p\rangle}{\langle Y_e, p\rangle}\right\rvert \leq \epsilon
    \end{align*}
    with a probability of at least \(1 - \delta\) suffice with \(\mathcal{O}((\epsilon^2\min_e\langle Y_e, p\rangle)^{-1}\log(2\lvert E\rvert/\delta))\) samples, assuming $\min_e\langle Y_e, p\rangle > 0$.

\end{theorem}

Without any extra assumption on the graph structure, these bounds give us an upper bound on the sample complexity to approximate the edge cut probability for any type of distribution. Therefore, it serves as an upper bound for both the SDP solver and the quantum solver. This gives us the bound for the number of shots needed to run our quantum circuit, given the obtained quantum states from the variational framework. This is discussed in more detail in Appendix \ref{appendix:sampling}. \BBI{We want to note that the bound in Theorem~\ref{prop:sampling_hoeffding} is not tight and in practice, we need much fewer number of shots to obtain a good approximation. This is demonstrated both numerically and experimentally in \textbf{Numerical Results} and \textbf{Hardware Experiments}}

\paragraph{Training \DQAOA{} for Fair Cut Cover}

The minimum \(f(x) \coloneqq \min_ig_i(x)\) of a collection of unique smooth functions \(g_i(x)\) is potentially non-smooth wherever there exists \(x, i,j\) such that \(f(x) = g_i(x) = g_j(x)\). As a result, the loss surface for a QAOA-ansatz circuit for the Fair Cut Cover problem on a graph with many edges can be non-smooth, meaning that small-angle initialization~\cite{gaidai2024performance} can struggle greatly to find good parameters for the entire graph. Especially since these non-smooth points can also be the optimal points; for example, the optimal distribution for complete graphs requires all edges to have the same edge cut probability. 

One approach we have taken to address this issue is to smooth the landscape using a smooth approximation such as LogSumExp (SoftMin), which is commonly employed in both optimization  \cite{boyd2004convex,blanchard2021accurately} and machine learning \cite{pillutla2018smoother, blondel2020learning}. This approach both improves the small-angle initialization technique and enables the calculation of gradients using a quantum device via the parameter-shift rule \cite{wierichs2022general}. 
Considering our current objective for Fair Cut Cover using variational quantum algorithm, with $\mathrm{Pr}_{e}(\ket{\psi(\theta)}) \coloneq (1 - \langle Z_{u}Z_{v}\rangle_{\theta})/2$, we have the smooth \BBI{LogSumExp (LSE)} objective as
\begin{align}
    \min_{e \in E} \mathrm{Pr}_{e}(\ket{\psi(\theta)}) \approx f_{\tau}(\theta) \coloneq \nonumber 
      -\tau\ln{\Bigg(\sum_{e \in E} \exp{\frac{-\mathrm{Pr}_{e}(\ket{\psi(\theta)})}{\tau}}\Bigg)} 
\end{align}
We demonstrate improved solution quality of this approach in Figure \ref{fig:training}. 

\begin{figure}[!ht]
    \centering
    \includegraphics[width=0.6\linewidth]{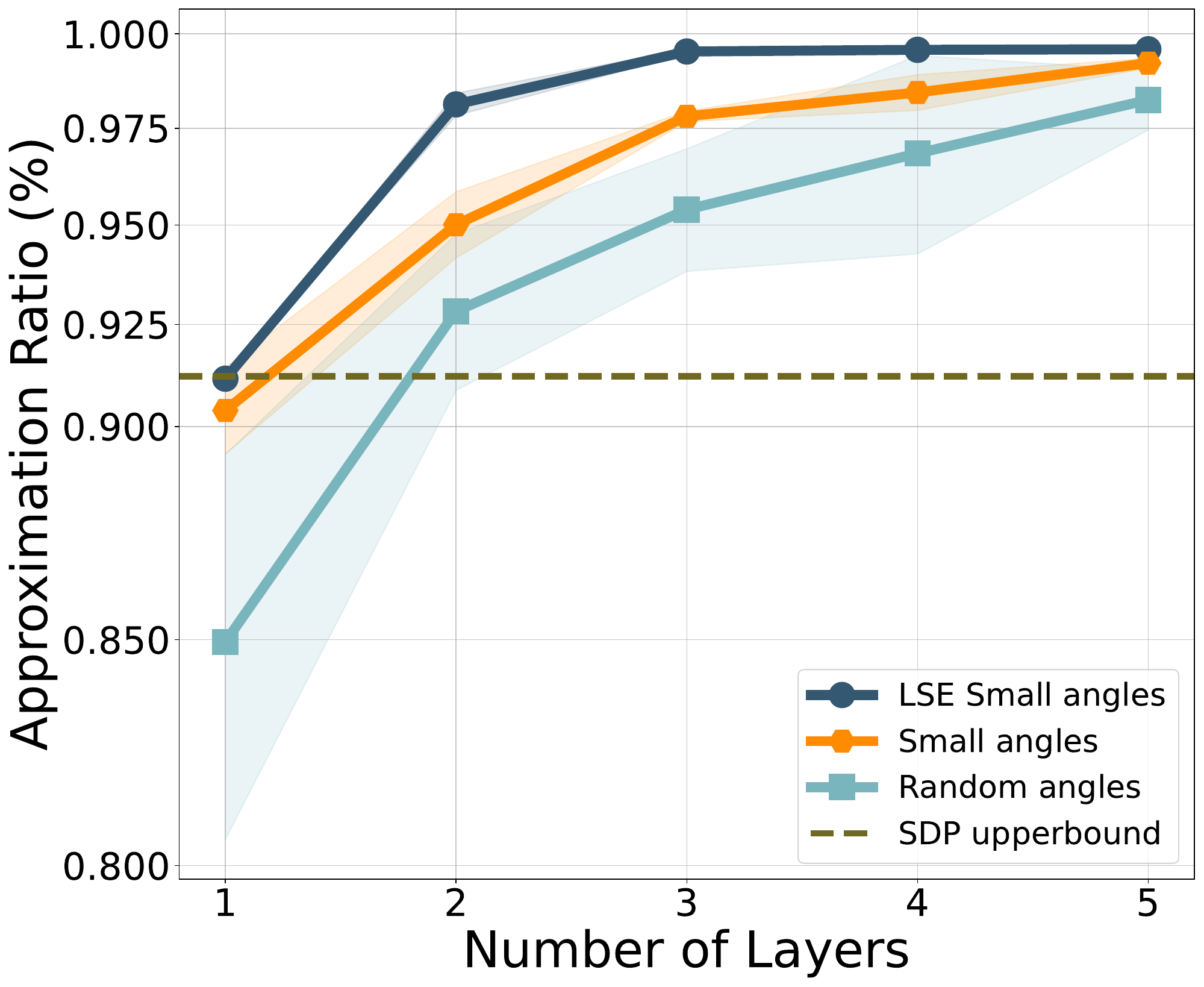}
    \caption{Performance of $k$-layer \DQAOA{} ansatz solver when using different initialization strategies on a single Erdős- Rényi $16$-node graph with max clique of size $4$. Every initialization strategy is repeated $10$ times. The line shows the mean over those $10$ initial seeds at a given number of layers, and the shaded region denotes the standard deviation over those $10$ seeds. Here, the SDP is solved optimally. In the experiments, we construct the random-angle initialization by uniformly and independently sampling each angle from $(0, 2\pi)$. Similarly, the initialization from small angles approach is constructed through uniformly and independently sampling each angle from $(0, 0.05)$.}
    \label{fig:training}
\end{figure}

\paragraph{Variance of gradient for objective}
Here, we look at the variance of the gradient for our min and LogSumExp objective. Following the analysis of \cite{kazi2025analyzing}, by assuming that we have enough layers for our \DQAOA{} ansatz to form unitary 2-design \cite{larocca2022diagnosing}, we give the upper bound for our objectives as follows
\begin{align}
    \mathrm{Var}[\partial_{\vartheta}f_{\tau}(\theta)] \leq \frac{d^{2}\abs{E}}{(d^{2} - 4)(d + 2)},
\end{align}
where $d = 2^{\abs{V}}$. Furthermore, for the LogSumExp objective, the expectation value of the gradient $\mathbb{E}[\partial_{\vartheta}f_{\tau}(\theta)] = 0$.  Moreover, for our min objective $f(\theta) \coloneq \min_{e\in E}p_{e}(\theta)$, as $f(\theta)$ is locally Lipschitz and based on the Clarke subgradient \cite{clarke1990optimization}, we have the Variance of gradient for $f(\theta)$ where $\theta$ are not non-smooth points as
\begin{align}
    \mathrm{Var}[\partial_{\vartheta}f(\theta)] \leq \frac{d^{2}\abs{E}}{(d^{2} - 4)(d + 2)},
\end{align}

More details on these calculations are shown in Appendix \ref{appendix:var_grad_obj}. In Figure \ref{fig:variance_grad}, we show the numerics results for the variance of the gradient of both the LogSumExp objective, $\mathrm{Var}[\partial_{\vartheta}f_{\tau}(\theta)]$ and the min objective, $\mathrm{Var}[\partial_{\vartheta}f(\theta)]$, for graphs with different numbers of nodes. Here, for each graph size, we use $20$ random instances of Erdős- Rényi up to $100$ ansatz layers, where for each instance, $100$ parameter points with a range from $(0, 2\pi)$ are considered. As a result, we can see that the analytical upper bound gives a tight bound.  
\begin{figure*}
    \centering
    \includegraphics[width=1.0\linewidth]{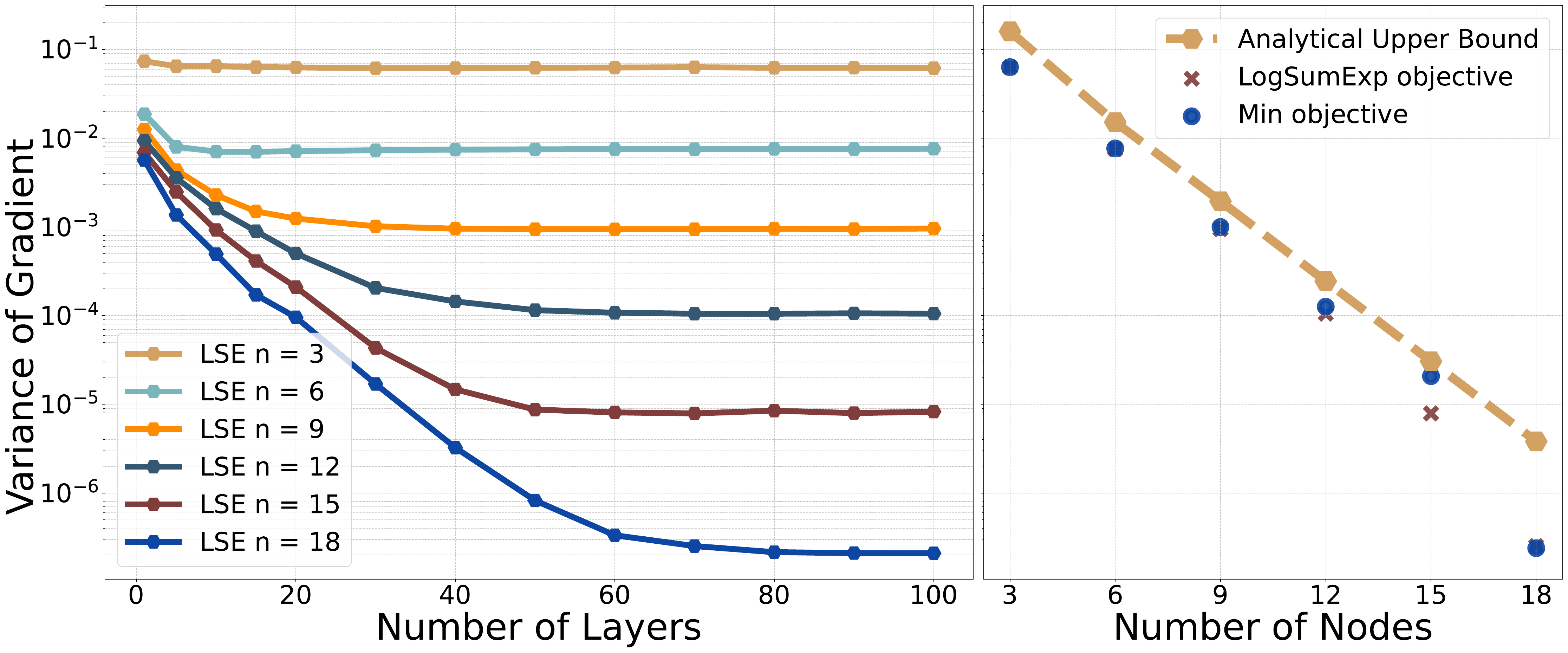}
    \caption{Variance of the gradient for the LogSumExp objective ($\mathrm{Var}[\partial_{\vartheta}f_{\tau}(\theta)]$) and the min objective for smooth points ($\mathrm{Var}[\partial_{\vartheta}f(\theta)]$) is calculated with respect to different graph sizes as in the left figure. For the right figure, the analytical upper bound for the LSE smooth objective is plotted, which bounds both the LogSumExp objective and the Min objective.}
    \label{fig:variance_grad}
\end{figure*}

\paragraph{Numerical Results}
The details of the numerical results for the simulation can be found in Section \ref{sec:numerical_details}. \BBI{Here, we create two different graph datasets for different benchmarks: A set of random graphs from $10$ to $20$ nodes with a max clique of size $4$ to $5$ (motivated by Corollary~\ref{cor:non_trivial_upperbound}), and a set of random triangle-free 3-regular graphs from $50$ to $500$ nodes}. For the small graph cases (from $10$ to $20$ nodes), we use \DQAOA{} ansatz and perform training using the Adam optimizer with small-angle initialization on the LogSumExp objective. \BBI{ For the larger instances of triangle-free 3-regular graphs from $50$ to $500$ nodes, we use the standard QAOA approach with the light-cone based optimization strategy \cite{pelofske2024short} combined with initialization methods from MaxCut such as the fixed-angle conjecture \cite{wurtz2021fixed}. For these large-scale graphs, the refinement budget is capped at $10$ iterations.}

For numerical simulation, we present our results when running on two different graph datasets. For the small graph dataset, we compare our benchmark directly with SDP as shown in Figure \ref{fig:maQAOA_performance}. Using a $2$-layer \DQAOA{} ansatz on $220$ randomly generated graphs with a max clique of size $4$ to $5$ (each node size contains  20 graphs), we beat the SDP hyperplane-rounding upper bound described in Corollary \ref{cor:non_trivial_upperbound}. Here, we simulate our circuits with both a noiseless simulator and a real device emulator. We can observe from the numerical simulation that, given a small number of layers, the \DQAOA{} ansatz can give a very good approximation ratio. Furthermore, it beats the SDP upper bound as stated in Corollary \ref{cor:non_trivial_upperbound} for all $220$ instances by a fair amount.

\begin{figure*}
    \centering
\includegraphics[width=1\linewidth]{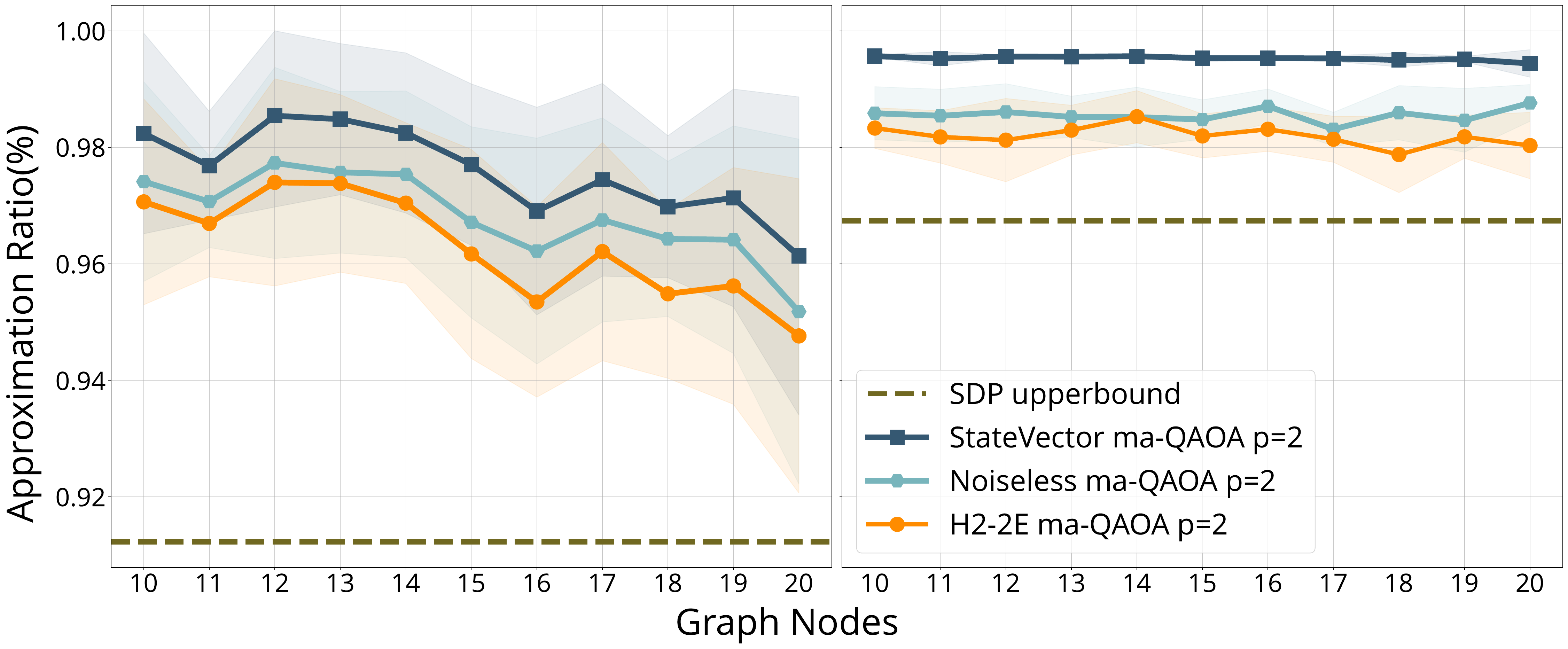}
    \caption{Numerical simulation results on $220$ graphs ranging from $10$ to $20$ nodes with either max clique of size $4$ (on the left) or $5$ (on the right). Here, the SDP upper bound is calculated based on Corollary \ref{cor:non_trivial_upperbound}, whereas the $2$-layer \DQAOA{} ansatz solver is plotted with exact expectation value (state vector simulation); noiseless simulator, and Quantinuum's H2-2 emulator over $10000$ shots.}
    \label{fig:maQAOA_performance}
\end{figure*}

\BBI{It can also be observed from Figure~\ref{fig:maQAOA_performance} that the \DQAOA{} ansatz seems to perform worse on the set of graphs with max clique of size $4$ than the one with max clique of size $5$. Here we note that although the clique number is closely related to the Fair Cut Cover value of a graph, graph structure does have an impact. In Figure~\ref{fig:test_graph_structure}, we show the minimum cut frequency for edges within a maximum clique of the host graph, outside of the maximum clique of the host graph, and edges incident to the maximum clique of the host graph. As the clique number of the graph instances goes from 4 to 5, we can see that the edges inside the maximum clique begin to dominate the overall objective value. This may be due to the manner in which this dataset was constructed: a large collection of Erdős- Rényi graphs were sampled and then filtered by their clique number. It may be that because larger cliques are rarer in the Erdős- Rényi model, they have a larger structural impact on the objective of Fair Cut Cover.}

\begin{figure}[!ht]
    \centering
    \includegraphics[width=1.0\linewidth]{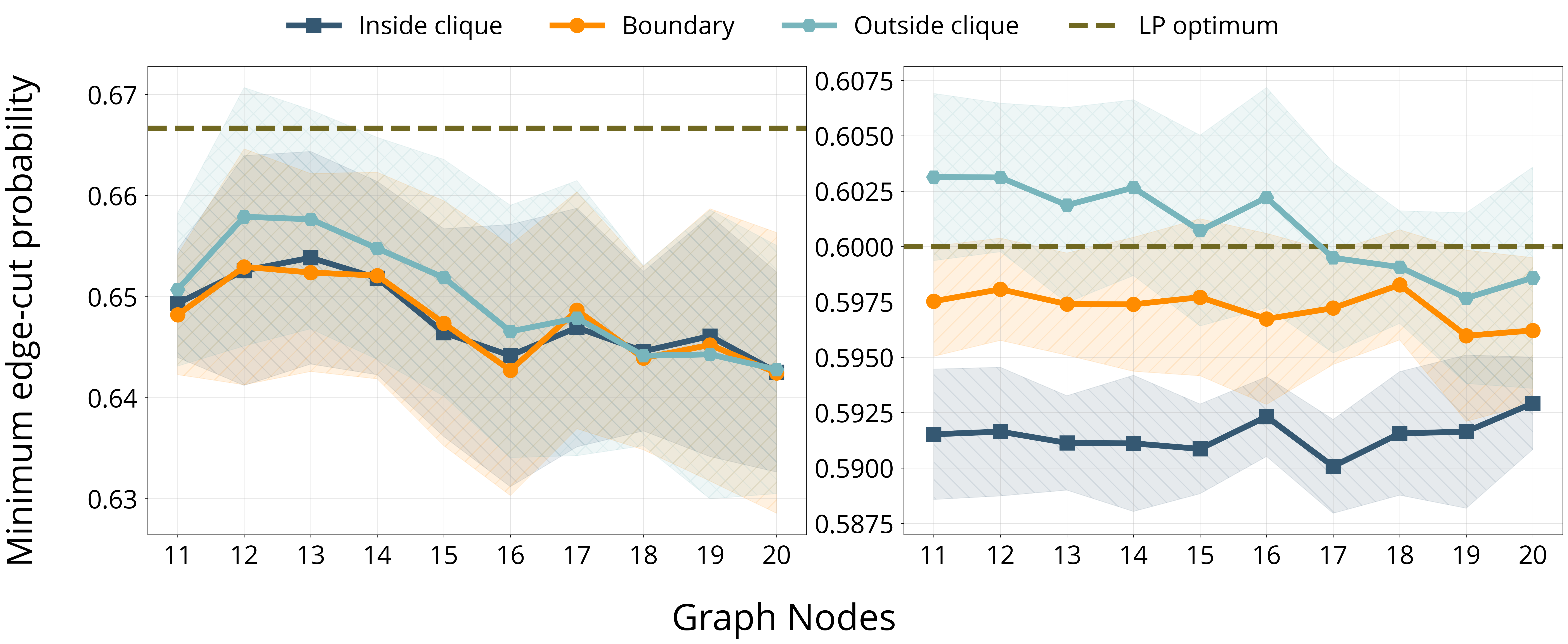}
    \caption{Minimum edge-cut probabilities by edge category for depth-2 $\mathcal{D}$-QAOA on graphs with maximum clique size of either $4$ (on the left) or $5$ (on the right). Edges are classified relative to a reference maximum clique as inside (both endpoints in the clique), boundary (exactly one endpoint in the clique), or outside (neither endpoint in the clique). For each graph, the minimum cut probability within each category is estimated from the noiseless-simulation samples at the optimized parameters. Solid curves show the mean of these per-graph minima, and shaded bands indicate the standard deviation across graph instances. Dashed horizontal lines indicate the exact LP optima, $2/3$ and $3/5$, respectively. The panels use different vertical scales.}
    \label{fig:test_graph_structure}
\end{figure}

\BBI{For the instances of triangle-free 3-regular graphs, we compare our benchmark with a proposed heuristic based on Column Generation, as explained in detail, and the hyperparameter settings in Appendix \ref{appendix:heuristic_column_generation}. This is a well-established technique to solve large-scale linear programming problems \cite{desaulniers2006column}. \BBI{At each iteration, the restricted master problem assigns probabilities to the current set of cuts (the \emph{cut pool}) to maximize the minimum edge-cut probability; we then use the resulting dual edge weights to search for a new cut by applying simulated annealing to a weighted MaxCut problem (the \emph{pricing subproblem})}. The results are shown in Figure \ref{fig:large_scale_tfcg}. Using up to a $3$-layer standard QAOA ansatz with the fixed-angle conjecture on regular graphs from \cite{wurtz2021fixed} on $50$ randomly generated triangle-free 3-regular graphs (each node size contains 5 graphs), we beat the Column Generation heuristic with a limited budget at larger scales.

We note that the non-monotonic dependence of the $p=3$ standard QAOA results on graph size is consistent with the local nature of shallow QAOA. For standard QAOA, the cut probability of an edge depends only on its rooted depth-three causal neighborhood, containing at most $30$ vertices in a cubic graph. Consequently, at fixed parameters $\theta$, the FCC objective can be written as
\[
Q^{\mathrm{std}}_3(G;\theta)=\min_{\tau\in\mathcal T_3(G)}q_\tau(\theta),
\]
where $\mathcal T_3(G)$ is the set of rooted neighborhood types present in \(G\). The triangle-free property does not make these neighborhoods identical: squares, pentagons, and overlapping cycles can produce different local responses. Because FCC measures the least-covered edge, even a single unfavorable neighborhood can determine the value, irrespective of its frequency. Different sampled graphs can therefore exhibit different bottlenecks despite having the same degree and girth. Averaging over five instances per size reduces sampling variability but does not imply a monotonic trend with graph size. 
This dependence on local neighborhoods is closely related to studies of QAOA parameter transferability for MaxCut, which exploit similarities between local subgraphs \cite{galda2021transferability,galda2023similarity}. An important distinction is that, at fixed parameters, the expected MaxCut objective aggregates these local contributions according to their multiplicities, whereas Fair Cut Cover is controlled by the least favorable neighborhood type, even if it occurs only once.

The decrease observed for column generation should be interpreted in the context of its heuristic pricing budget. Although the permitted cut-pool size increases with graph size, simulated annealing remains capped at $100,000$ steps for each pricing call. Increasing the number of admissible columns does not ensure that the pricing procedure identifies sufficiently useful cuts for the larger weighted MaxCut subproblems. Limited pricing effort can leave improving columns undiscovered and restrict the quality of the resulting cut distribution. The observed decrease is therefore consistent with a limitation of the specified pricing budget; it does not establish a corresponding decrease in the optimal FCC value or an intrinsic limitation of column generation.
}

\begin{figure}[!ht]
    \centering
    \includegraphics[width=0.6\linewidth]{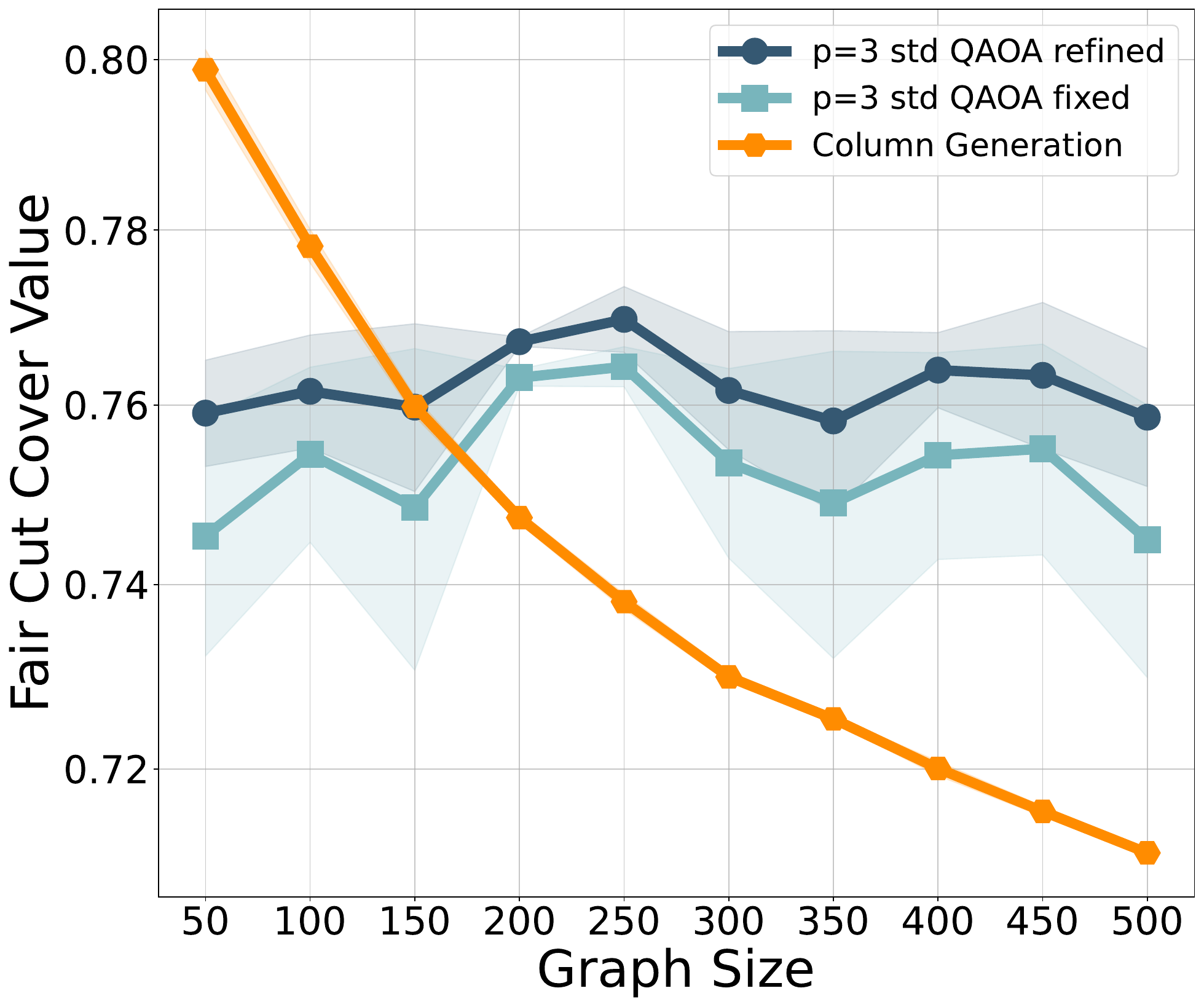}
    \caption{Numerical simulation results on $50$ graphs ranging from $50$ to $500$ nodes (each node size containing $5$ graphs) for random triangle-free $3$-regular graphs. Here we show the average and standard deviation over $5$ graphs at each node size; we compare the results from $p=3$ standard QAOA against the Column Generation approach. Furthermore, we compare the two lines where ``$p=3$ QAOA fixed" denotes the expectation value from using the fixed-angle conjecture and ``$p=3$ QAOA refined" denotes the expectation value when using the fixed-angle conjecture as initialization with $10$ refined iterations.}
    \label{fig:large_scale_tfcg}
\end{figure}

\BBI{
We also want to point out that using the fixed-angle conjecture initialization strategy for MaxCut gives a very good head start for our Fair Cut Cover objective, as shown in Figure~\ref{fig:large_scale_tfcg}. This is not very surprising, as the same ansatz is used for both problems, as well as their deep relationship (the objectives are the same for both problems when considering edge-transitive graphs). 
This reuse of MaxCut parameters for Fair Cut Cover is related to
cross-problem parameter transfer, where parameters optimized for one
combinatorial objective initialize a different objective. Such transfer
has previously been investigated from MaxCut to Maximum Independent
Set~\cite{nguyen2025cross}. Here, we investigate transfer from
an expected-cut objective to a worst-edge coverage objective. 
To give more context about the trainability, we look at an interesting class of graphs called ``Snark" \cite{brinkmann2013generation} which is a class of non-trivial 3-regular triangle-free graphs. These graphs have been shown to be \texttt{NP-hard} for some problems. More specifically, we look at the Double-star snark graph with $30$ nodes and $45$ edges through the fixed-angle conjecture~\cite{wurtz2021fixed}, the balanced method \cite{vcepaite2025quantum} and the naive small-angle initialization as initialization for our QAOA. This graph is a concrete example of a graph where a finite layer of QAOA beats both Column Generation with fixed budget and the SDP + Hyperplane rounding approach. The details are shown in Figure~\ref{fig:trainability_snark}. On the left panel, we look at $p=11$ standard QAOA using both initialization approaches with $16$ gradient-based iterations of refinement. On the right panel, we show the value obtained for $\mathcal{D}-$QAOA and standard QAOA for a finite number of layers with the initialization strategy and also with $16$ gradient-based iterations of refinement. Here we use $5$-layer $\mathcal{D}-$QAOA as this is where we beat the Column Generation methods with fixed resources.}
\begin{figure}[!ht]
    \centering
    \includegraphics[width=\linewidth]{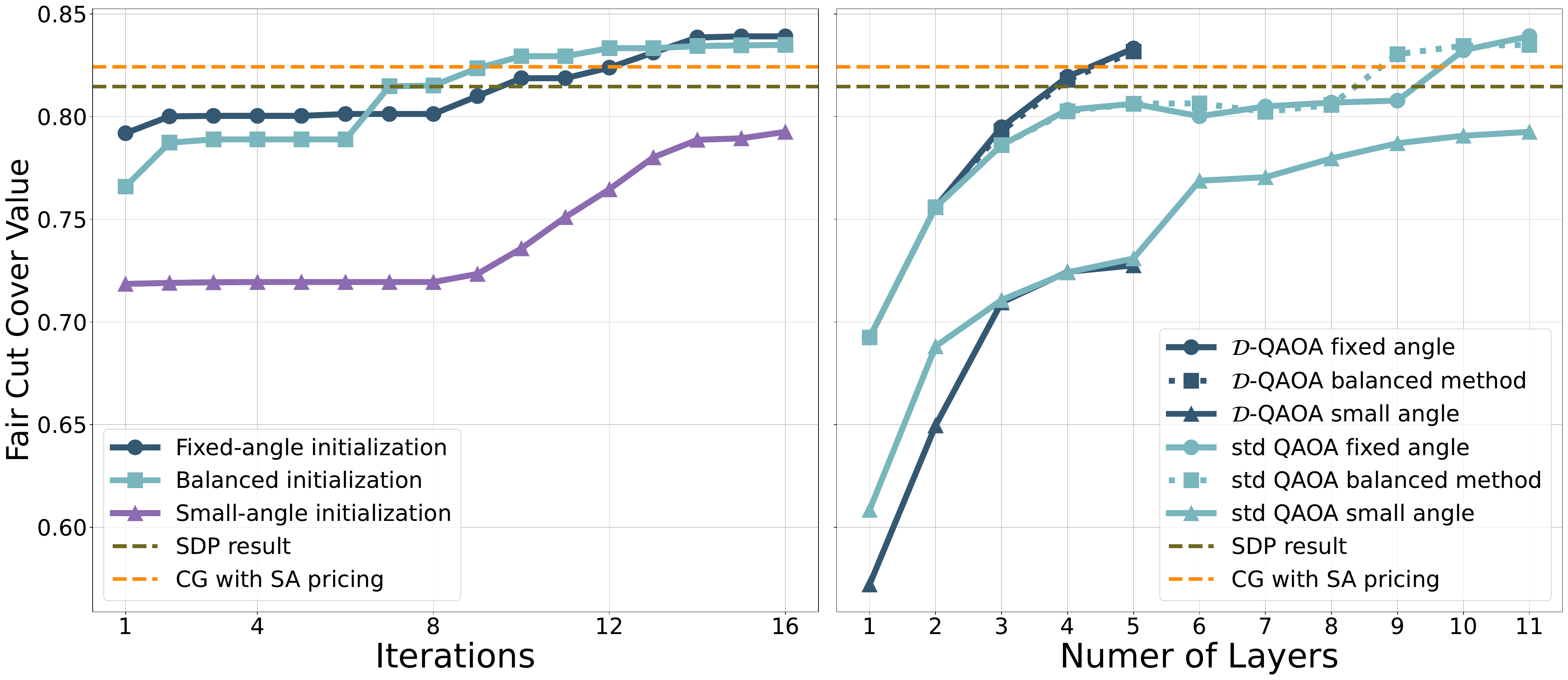}
    \caption{Effect of initialization and parameter refinement on Fair Cut Cover performance on a $30$-vertex, $45$-edge Double-star snark graph. The vertical axis shows the minimum edge-cut probability, $\min_{e\in E}P_e(\boldsymbol{\theta})$. Left: Refinement trajectories of standard QAOA at depth $p=11$, starting from the fixed-angle, balanced method and the naive small-angle initialization, over $16$ gradient-based optimization iterations. Right: FCC values attained by standard QAOA and $\mathcal{D}$-QAOA as a function of circuit depth, using each initialization strategy followed by $16$ gradient-based refinement iterations. Horizontal reference lines indicate the values obtained by SDP with hyperplane rounding and column generation with simulated-annealing pricing.}
    \label{fig:trainability_snark}
\end{figure}
\BBI{
\paragraph{Hardware Experiments}
The details of the experiments and compilation can be found in Appendix \ref{sec:exp_details}}. For experiments, we perform $60$ experiments on the H2-2 device with graph size ranging from $10$ to $15$ nodes, with a max clique size of $5$, and one experiment on a $60$-node complete graph on the Helios-1 device. All small instances were optimized for up to 1000 iterations. \BBI{For the Helios-1 experiment, we use a single-layer standard-QAOA circuit with optimal parameters through analytical optimization. The executed parameters and their selection procedure are specified in Appendix~\ref{hw:symmetry}. The analysis distinguishes the ideal expectation value,  the performance of the noiseless circuits with finite sampling, and deviations observed in the experimental distribution. Appendix~\ref{appendix:hardware-methods} describes the hardware data, estimators, and reference distributions.

We first show the results of small instances in Figure~\ref{fig:hardware_smallscale}. Across the six graph sizes, the noiseless setting (sampling from statevector), noisy simulation (using H2-2 emulator), and hardware curves track each other closely and approach the exact reference as the shot count increases (each instance has the budget of $1500$ shots). This agreement indicates that finite-shot sampling effects account for most of the reduction visible in these small-instance experiments at the considered sampling budgets. It does not establish that hardware noise is absent; rather, any additional noise-induced reduction is small relative to the sampling effects resolved by this comparison. %
\begin{figure}[!ht]
    \centering
    \includegraphics[width=\linewidth]{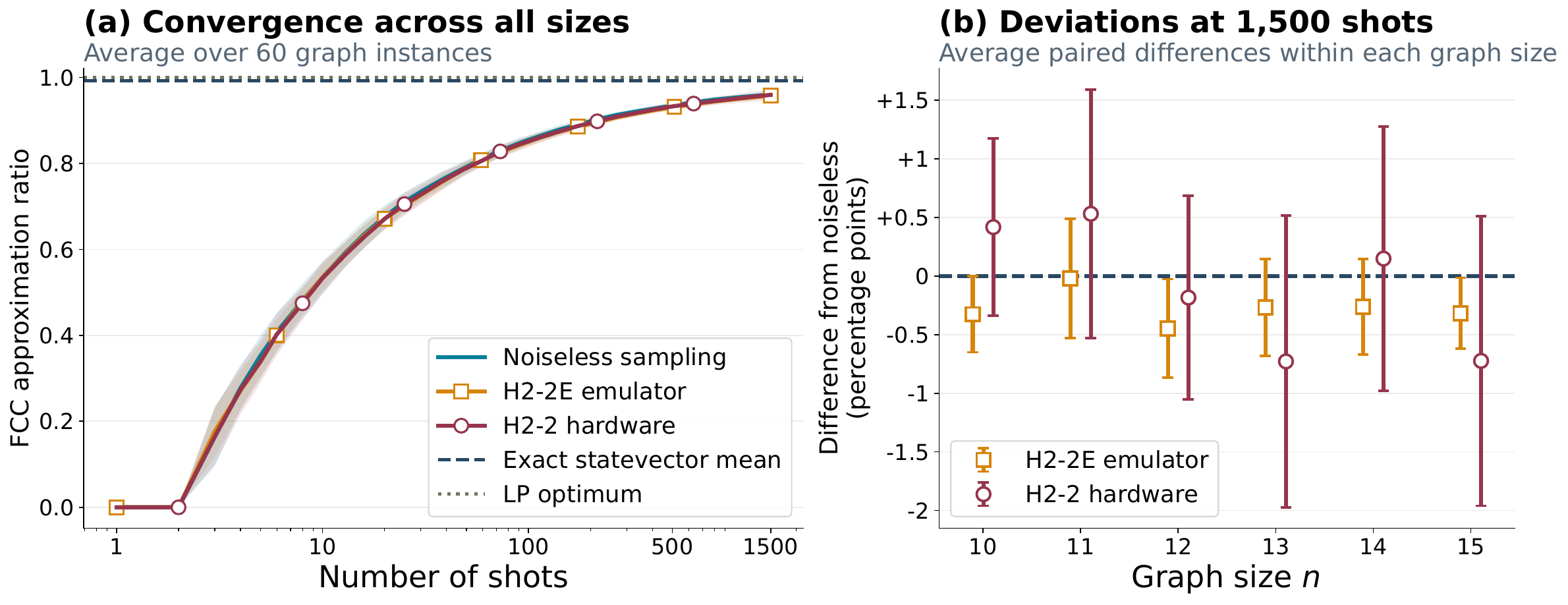}
    \caption{Finite-shot convergence and deviations from noiseless QAOA sampling on H2-2. (a) Average empirical FCC approximation ratio versus number of shots for noiseless sampling, the H2-2E noisy emulator, and H2-2 hardware. Each curve averages over 60 graphs: 10 instances of maximum clique size 5 at each size $n=10,\ldots,15$. For each graph and shot subset, the approximation ratio is the minimum empirical edge-cut probability divided by the FCC optimum. The gap between the noiseless curve and the statevector reference illustrates the finite-shot effect of taking the minimum over empirical edge-cut probabilities. (b) Mean paired difference from noiseless sampling at 1,500 shots, in percentage points. For each graph, we subtract the noiseless ratio from the emulator or hardware ratio before averaging within each graph size. Negative values indicate lower FCC performance than noiseless sampling. Shaded bands in (a) show one standard deviation above and below the mean approximation ratio across all 60 graphs at each shot count. Error bars in (b) show one standard deviation above and below the mean paired difference from noiseless sampling, computed across the 10 graphs at each size.}
    \label{fig:hardware_smallscale}
\end{figure}

The $K_{60}$ experiment exhibits a different behavior. First, we note that as established in Lemma~\ref{prop:best_Kn_value}, we know the optimal FCC value of $\bar{\eta}({K_{60}})= 30/59$. Figure~\ref{fig:k60-distributions} compares $5000$ hardware samples with $5000$ noiseless samples for $p=1$ QAOA and with $5000$ samples from uniformly random bitstrings. Furthermore, it also shows how the samples change when applying a noisy channel to the ideal QAOA sampled distribution. 

To describe the sampled cuts, let $W=\sum_{i=1}^{60}b_i$ be the Hamming weight of a measured bitstring where $b_i = (1 - x_{i})/2$. The two sides of its cut contain $W$ and $60-W$ vertices. We define the cut imbalance as $\Delta=|W-30|$, half the absolute difference between their sizes. Thus, $\Delta=0$ denotes a balanced cut, and larger $\Delta$ denotes a less balanced partition. On $K_{60}$, such a cut crosses $W(60-W)=900-\Delta^2$ edges, so the imbalance distribution directly describes how the samples depart from maximum cuts.

Uniform sampling assigns each vertex independently to either side with probability $1/2$ and can be thought of as the distribution produced after measuring the initial state $\ket{+}^{n}$ multiple times. It provides an inexpensive classical baseline that requires no circuit or parameter optimization. Each edge is cut with probability $1/2$, giving the population FCC ratio
$(1/2)/(30/59)=59/60 \approx 0.9833$. This high ratio arises because the optimal FCC on $K_{60}$ is itself only slightly above $1/2$; uniform sampling therefore leaves little room for improvement in this normalized objective. Its samples nevertheless have $W\sim\operatorname{Binomial}(60,1/2)$ and $\mathbb{E}[\Delta^2]=15$, whereas the noiseless QAOA distribution is more concentrated near balanced cuts. Including uniform samples at the same shot counts lets us assess whether the implemented circuit improves on random cuts while exposing the finite-sample effect of estimating a minimum over edges. It also provides the limiting measurement distribution of the global depolarizing model below.

Permutation symmetry allows the ideal QAOA sampled distribution to be computed in a $61-$dimensional basis indexed by Hamming weight~\cite{shaydulin2021classical} (more information is in Appendix~\ref{hw:symmetry},~\ref{hw:reduced-dynamics}). Its population FCC value is approximately $0.507909$, close to the optimum $30/59$. The optimized reference serves as a benchmark. Nevertheless, the hardware sampled distribution places substantially more weight on imbalanced cuts and exhibits a much broader distribution of empirical edge-cut probabilities, with substantial probability
at large $\Delta$. Since a cut loses $\Delta^2$ edges relative to
a balanced cut, this large-imbalance tail reduces the average number of covered edges. More precisely, writing $\overline p=1770^{-1}\sum_e p_e$ for the average edge-cut probability, any distribution $\mu$ on cuts of $K_{60}$ satisfies
\begin{equation}
    \label{eq:average_cut_prob_K60}
  \overline p
  =\frac{900-\mathbb{E}_\mu[\Delta^2]}{1770}
  =\frac{30}{59}-\frac{\mathbb{E}_\mu[\Delta^2]}{1770},
  \qquad
  \mathrm{FCC}_{K_{60}}(\mu)=\min_e p_e \leq \overline p.
\end{equation}
Here $\mathrm{FCC}_{K_{60}}(\mu)$ is the Fair Cut Cover value of distribution $\mu$ over graph $K_{60}$. Equation~\eqref{eq:average_cut_prob_K60} also gives the following decomposition of the gap to the optimal FCC value:
\begin{align}
    \bar{\eta}(K_{60})-\mathrm{FCC}_{K_{60}}(\mu)
=
\frac{\mathbb{E}_{\mu}[\Delta^{2}]}{1770}
+
\left(
\bar{p}-\min_{e\in E}p_e
\right).
\end{align}
The first term measures the loss in average edge coverage caused by cut imbalance, while the second measures the gap between average and worst-edge coverage. Both terms are non-negative, and the identity holds without any symmetry assumption. This connects the cut-imbalance distribution in Figure~\ref{fig:k60-distributions}(a) with the edge-cut probabilities in Figure~\ref{fig:k60-distributions}(b). When all quantities are computed from the same finite sample set, the decomposition also holds for the empirical distribution. In that case, the second term includes finite-sampling variation, even when all population edge-cut probabilities are equal, and therefore cannot be attributed entirely to hardware-induced differences in population edge coverage.

The hardware implementation thus falls below the simple random-cut baseline in this experiment. These results indicate that hardware noise plays a major role in the observed performance degradation, as the hardware FCC remains substantially below the noiseless QAOA reference at the same shot count.
\begin{figure}[!ht]
    \centering
    \includegraphics[width=\linewidth]{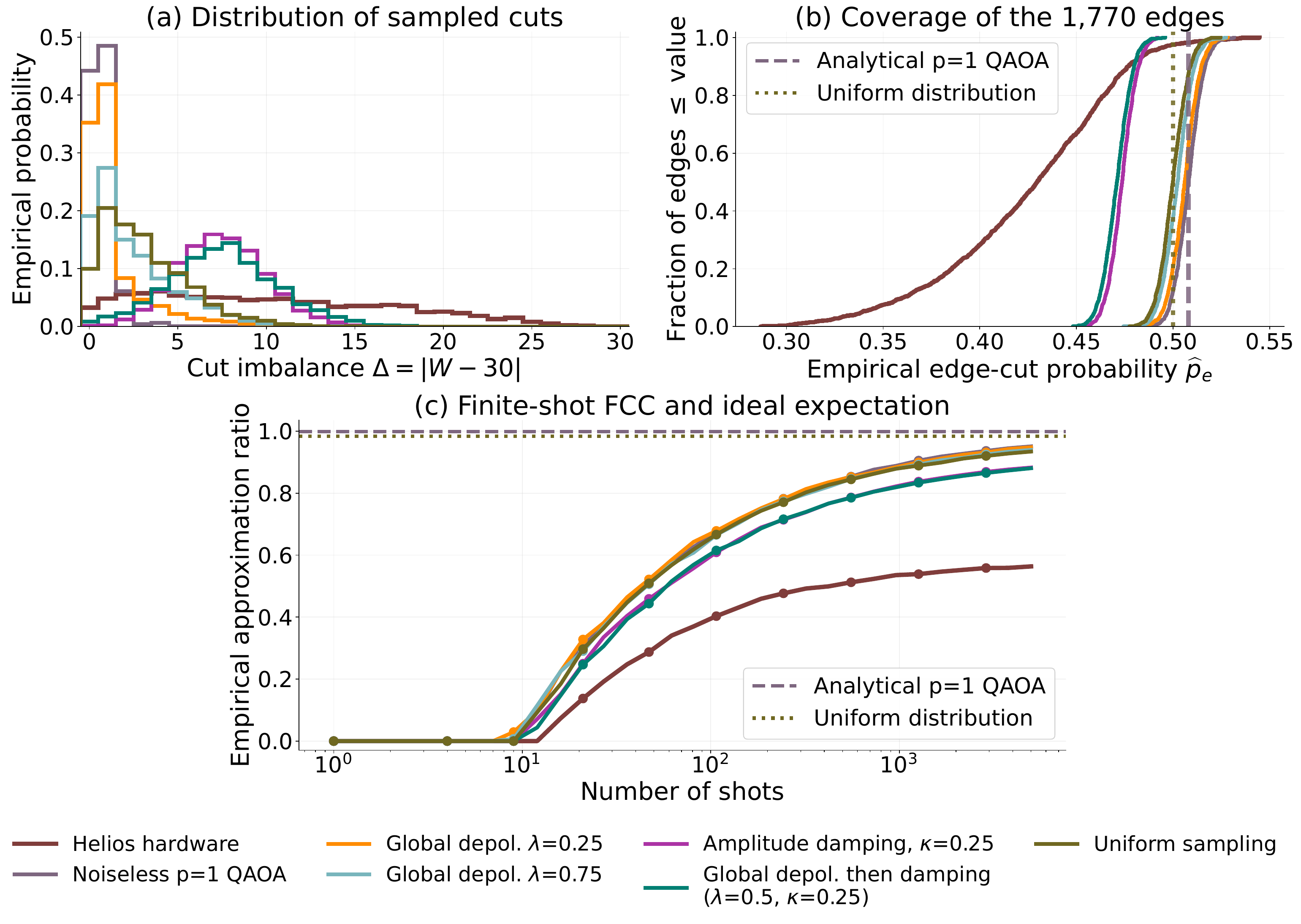}
    \caption{How hardware and noisy channels modify the $p=1$ QAOA sampled distribution on $K_{60}$. Hardware samples from Quantinuum Helios-1 are compared with samples from the noiseless $p=1$ QAOA circuit at the optimal angles, uniform random bitstrings, and QAOA under noisy channels. (a) Empirical distribution of the cut imbalance $\Delta=|W-30|$, where $W$ is the bitstring Hamming weight and the two sides of the cut have $W$ and $60-W$ vertices; $\Delta=0$ denotes a balanced cut. (b) Empirical cumulative distribution across the $1770$ edges of their estimated cut probabilities. (c) Empirical FCC approximation ratio over $5000$ samples. Every panel uses $5000$ samples per source. The dashed lines show the exact $p=1$ QAOA value, and the dotted line shows the uniform reference at ($1/2$) in (b) and their normalized FCC values in (c). The noise models are global depolarization with $\lambda=0.25$ or $0.75$, independent amplitude damping with $\kappa=0.25$, and global depolarization with $\lambda=0.5$ followed by amplitude damping with $\kappa=0.25$.}
    \label{fig:k60-distributions}
\end{figure}

To illustrate how the distribution transforms under the noisy channel, we also apply global depolarization (mixing the state with the maximally mixed state), amplitude damping, and their sequential combination. Depolarization mixes the QAOA distribution with uniform sampling, whereas amplitude damping biases the measured bitstrings toward zero and changes the distribution of cut imbalance. These models do not provide a calibrated explanation of the device noise but rather give some insight into how a theoretical noise model affects the output distribution. In particular, their identical action on all qubits preserves equality of the population edge-cut probabilities. At matched shot counts, the sampled model references show a narrower spread of empirical edge-cut probabilities than the hardware data. Also, because the maximally mixed state produces a good approximation ratio (as stated in the uniform sampling), the effect of the global depolarization does not worsen the FCC value by much.}

\section{Discussion}
\paragraph{Insufficient number of layers to capture $\mathbb{Z}_{2}$-symmetric cut distribution}
Besides the fact that QAOA with enough layers can capture all $\mathbb{Z}_{2}$-symmetric cut distributions, while SDP can not. An interesting part that highly separates the QAOA solver from the SDP solver is its characteristic when the number of layers is not sufficient. As discussed in more detail in Appendix \ref{appendix:ma_QAOA_prop}, when the number of layers is too small, the subgraph monotonicity of the QAOA solver is not guaranteed due to the difference in the circuit expressibility between the original graph and its subgraph. A concrete example is with the cycle and path graphs, where, from the analysis of \cite{kazi2025analyzing}, regardless of the number of layers, the vanilla multi-angle QAOA ansatz based on cycle and path can not capture all $\mathbb{Z}_{2}$ cut distributions, and we need modifications such as our introduced \DQAOA{}.  As mentioned in Appendix \ref{appendix:ma_QAOA_prop}, it is also hard to characterize the behaviour of the standard QAOA ansatz on the subgraph when compared to the original graph. The only exception is when the subgraph monotonicity is held for $k=1$ standard QAOA on triangle-free regular graphs. This motivates an interesting aspect of circuit design where, instead of adding more layers with the same graph structure to improve the expressibility, clever edge additions without affecting the objective can give good-enough results with fewer gate counts.

\paragraph{Universality Result}
We argue that Corollary \ref{thm:sufficient_p_covering_distributions} constitutes a universality result analogous to what has previously been proven for Feedforward Neural Networks (FNNs) in machine learning. In theory, a feedforward neural network can be constructed to approximate any continuous function to a desired level of accuracy given a sufficient number of resources \cite{HORNIK1989359}. We have shown that \DQAOA{} is capable of something similar, as a \DQAOA{} circuit with a sufficient number of layers is able to approximate any \(\mathbb{Z}_2\)-symmetric distribution over bitstrings to a desired level of accuracy. \BBI{Furthermore, when using other multi-angle type QAOA as concretely discussed in Appendix~\ref{appendix:fractional_coloring} for independent sets, we can capture distributions that do not preserve \(\mathbb{Z}_2\) symmetry as well}. What's more, they share similar limitations. 

One of the most well-studied obstacles in QAOA research is the concept of a barren plateau, which is a region of the loss landscape for QAOA where the gradient of the loss nearly vanishes \cite{Larocca_2025}. When present, barren plateaus make it exceptionally difficult to train a QAOA circuit effectively. Likewise, FNNs suffer from similar problems, where one sees significant diminishing returns as the size of the network increases and vanishing gradients become more prevalent\cite{Basodi_gradient}.

What is needed for \DQAOA{} is something analogous to what Convolutional Neural Networks (CNNs) did for FNNs. By the clever introduction of inductive biases, CNNs like ALEXNET were able to revolutionize image processing overnight \cite{Krizhevsky_imagenet}. We hope that this work serves as both motivation and as an early step towards encoding better inductive biases into QAOA circuits in order to better capture the space of distributions relevant to the problem at hand.

\paragraph{Trainability and scaling potential}
A central question is whether the proposed method, when implemented with a variational quantum algorithm, can overcome trainability and scaling challenges well enough to yield a practical quantum advantage. 
In particular, our analysis using the \DQAOA{} ansatz with both LogSumExp and min objective leads to barren plateaus when the number of layers is enough to form a unitary 2-design (the variance of the gradient decays exponentially with the problem size). 
\BBI{
Our results support three distinct conclusions about circuit depth. Theorem~\ref{thm:high-girth-fcc-separation} establishes a separation from Gaussian hyperplane rounding of any SDP embedding on its specified high-girth cubic family, using the reported $15$ layers standard QAOA evaluation. This QAOA layer count is independent of graph size. Figure~\ref{fig:large_scale_tfcg} provides numerical evidence at QAOA with $3$ layers on the tested finite ensemble of triangle-free cubic graphs. Triangle-freeness alone does not ensure the high-girth condition required by Theorem~\ref{thm:high-girth-fcc-separation}, so these experiments do not automatically instantiate that theorem.}
Second, as in the case of MaxCut, many pre-training, parameter transfer, or warm-start strategies \cite{montanez2025transfer, falla2024graph, tate2023bridging}  can be used to improve either the training results or training time of the \DQAOA{} ansatz. We show concretely in the manuscript that using the same MaxCut initialization for the Fair Cut Cover objective also gives a good warm start, as shown in Figure~\ref{fig:large_scale_tfcg} and Figure~\ref{fig:trainability_snark}. Third, we point the reader to \cite{sciorilli2025towards, bach2024mlqaoa} as examples of how to scale up the current approach using an efficient qubit encoding scheme using either quantum information or classical information. 
The authors of \cite{sciorilli2025towards} use Pauli-correlation encoding to encode the binary variables of the MaxCut problem into the Pauli basis, while  \cite{bach2024mlqaoa} uses classical information (rank-2 embedding) to coarsen the classical problem and solve the coarsest level using a quantum solver. Extending a distribution on such a coarse embedding to a distribution on the original input is an interesting avenue for further research in this direction. \BBI{Also, the concentration result concerns the unitary-design regime and does not by itself characterize the behavior under structured initialization. Recent work uses first-moment analysis to distinguish initialization-induced biases and provides numerical evidence that different initializations can lead to different optimization outcomes \cite{kulshrestha2026exponentially,kulshrestha2022beinit}. This would motivate additional investigation  which initialization-induced biases are beneficial for the Fair Cut Cover objective, rather than treating avoidance of concentration alone as a guarantee of successful optimization.}

\section{Method}
\subsection{Fair Cut Cover through distributional fairness framework}

Given a distribution over solutions $x \in S$ which minimizes the value of a function $f_x$ on any partition $\Gamma \coloneq \{U_i\}^{\gamma}_{i}$, let $Y \in \mathbb{R}^{\gamma \times \abs{S}}$ such that $Y_{i, x} = \frac{f_x(U_i)}{\abs{U_i}}$. Then, we have the fairness objective to maximize the minimum partition
\begin{align}
    \max_{p} \min_{i \in [\gamma]} \langle Y_i, p \rangle 
\end{align}
Taking $\gamma = E$ as the set of all edges, the partition $\Gamma \coloneq \{e \in E\}$, the solution $x$ as a cut, and $f_x(e)$ denoting if the cut $x$ covers edge $e$, naturally yields $p$ as the cut distribution. The function $\langle Y_e, p \rangle$ now can be interpreted as the cut probability $\mathrm{Pr}_{e=(u,v)}(p) \coloneq \mathrm{Pr}_{X \sim p}[X_u \neq X_v] = \mathbb{E}_{X \sim p}\left[\mathds{1}[x_u \neq x_v]\right]$. Therefore, the fairness objective now becomes finding the cut distribution that maximizes the smallest edge cut probability $\min_{e \in E} \mathrm{Pr}_{e}(p)$, yielding
\begin{align}
    \bar{\eta}(G) = \max_{p} \min_{e \in E} \mathrm{Pr}_{e}(p) = \max_{p} \min_{e \in E} \langle Y_e, p \rangle
\end{align}
Based on this fairness objective, we can natively derive the linear program that solves this problem as finding the maximum pair $ (p, t)$ such that $\mathrm{Pr}_{e}(p) \geq t, \forall e \in E$,
\begin{equation}
    \begin{aligned}
    \label{eq:primal_LP}
        \max t & \\ 
        \text{s.t: } &\sum_{x \in \{-1, 1\}^{\abs{V}}} Y_{ex}p_{x} \geq t, \quad \forall e \in E \\
        & \sum_{x \in \{-1, 1\}^{\abs{V}}} p_{x} = 1, \quad p \geq 0
    \end{aligned}
\end{equation}
Noted that, instead of considering the cut distribution, we can show that for any problem that respects the $\mathbb{Z}_{2}$ symmetry, $\mathbb{Z}_{2}$-symmetric cut distribution $p$ (\(\forall x \in \{\pm1\}^n,\) \(p(x) = p(- x)\)),  can replace all cut distribution.

\subsection{Numerical Details}
\label{sec:numerical_details}
\paragraph{Choice of instances}
The numerical results in Figure \ref{fig:maQAOA_performance} were carried out on $220$ random graphs generated from the Erdős- Rényi model \cite{erdds1959random} with an edge probability of $0.15$ with a filter. The graph size ranges from $10$ to $20$ nodes, each with $10$ instances. 
Here, the filter ensures random instances with the same number of nodes have a maximum clique of size $4$ or $5$, while no isomorphic or disconnected graph is selected. 
We divide the random instances into two groups with max clique of size $4$ and size $5$, which are used in both Figure~\ref{fig:variance_grad},  Figure~\ref{fig:maQAOA_performance}, and Figure~\ref{fig:test_graph_structure}. For the large-scale experiments with graph sizes ranging from $50$ to $500$ nodes each with $5$ instances, we generate random 3-regular graphs\cite{kim2003generating,steger1999generating}
and filter out graphs with triangles. We then end up with $50$ triangle-free 3-regular graphs, which are used in Figure~\ref{fig:large_scale_tfcg}
The Double-star snark with $30$ vertices and $45$ edges is a well-known snark which can be constructed through~\cite{brinkmann2013generation}. This Double-star snark is used in Figure~\ref{fig:trainability_snark}
 
Especially, the set of graphs with a max clique size of $5$ and an edge probability of $0.1$ is used for experiments in Trapped Ion Quantinuum's H2-2 devices in Figure \ref{fig:hardware_smallscale}. Especially, we choose a $60$-node complete graph for running on Trapped Ion Quantinuum's Helios devices in the inset of Figure \ref{fig:k60-distributions}. 

\paragraph{Best solutions}
For finding the best Fair Cut Cover solution for a given instance, we solve the linear program as stated in Equation (\ref{eq:primal_LP}). Although the LP contains an exponential number of variables, we manage to find the best solution corresponding to the considered graph size with $n=20$. For the final quantum state $\ket{\psi}$ obtained from the quantum solver, we denote its best Fair Cut Cover value yielded from the empirical distribution derived from measuring $\ket{\psi}$  multiple times (for the emulator and simulator, it is 10000 shots), as $\mathrm{FCC}(\ket{\psi})$. Then the approximation ratio is defined as $\mathrm{FCC}(\ket{\psi}) / \mathrm{FCC}_{\mathrm{LP}}$, where $\mathrm{FCC}_{\mathrm{LP}}$ is the objective value of the optimal Fair Cut Cover obtained from the LP. Note that for large graph size, solving Linear Programming exactly is not possible. 

\paragraph{Variational Ansatz}
For the circuit ansatz, we choose to use the standard QAOA structure for MaxCut on edge-transitive graphs and 3-regular graphs, and use the multi-angle QAOA structure for MaxCut on other ``archetypal'' graphs. 
We observe that this ansatz structure can capture the $\mathbb{Z}_{2}$-symmetric cut distribution well using only a small number of layers.

\paragraph{Quantum-circuit simulations}
The classical simulation of quantum circuits has been done using PennyLane \cite{bergholm2018pennylane}, where for $20-$qubit or smaller circuits are simulated using ``lightning.gpu'' statevector simulator. For systems bigger than $20$ qubits, we use ``QUAPOPT''~\cite{quapopt_repo} with a built-in statevector GPU-based simulator and the lightcone technique. 

\paragraph{Optimization of circuit parameters}
In our experiments on small instances, we use the Adam optimizer with the following parameters (step size = $0.01$, $\beta_{1}=0.9$, $\beta_{2} = 0.99$, and $\epsilon = 1e-8$), where the gradient of the circuit is calculated using a time- and memory-efficient method, adjoint differentiation \cite{jones2020efficient}. The maximum number of iterations for the optimizer is $1000$ steps, where the stopping criterion is after $30$ steps with accumulated improvement less than $10^{-4}$. Here, for each instance, $10$ random seeds (using the default integer random generator) corresponding to $10$ random small initial angles are chosen. We chose the range for small-angle initialization to be $0 \rightarrow 0.05$.

\BBI{For the large-instance experiments and the initialization comparisons, we use sequential least squares programming (SLSQP) to refine the initial circuit parameters. The fixed-angle and balanced initializations are constructed as described in \cite{wurtz2021fixed, vcepaite2025quantum}.  We use an SLSQP convergence tolerance of \texttt{1e-9}. For the comparison experiments with Column Generation in Figure~\ref{fig:large_scale_tfcg}, SLSQP  has an iteration limit of $10$, while for initialization comparison in Figure~\ref{fig:trainability_snark}, SLSQP  has an iteration limit of $16$.}

\section*{Acknowledgement}

B.B. wants to thank Filip B. Maciejewski for the discussion about the QAOA implementation and Boris Tsvelikhovskiy for the insightful discussions about DLA. 
Research presented in this article was supported by the NNSA’s Advanced
Simulation and Computing Beyond Moore’s Law Program at Los Alamos National Laboratory. We would also like to thank the New Mexico Consortium, under subcontract C2778, the Quantum Cloud Access Project (QCAP), for providing quantum computing resources and technical collaboration. LANL report LA-UR-26-21897. This work was supported in part by NSF award \#2427042.

\section{Disclaimers}

The views expressed in this article are those of the authors and do not reflect the official policy or position of the U.S. Naval Academy, Department of the Navy, the Department of War, the U.S. Government, or Los Alamos National Laboratory.

\appendix
\setcounter{equation}{0}
\renewcommand{\theequation}{S\arabic{equation}}
\section{More details on Fair Cut Cover}\label{appendix:equiv}
Here, we show concretely how to form Fair Cut Cover from Fractional Cut Cover, and we restate the Fractional Cut Cover problem represented by Equation (\ref{eq:fcc}).
\begin{align*}
    \eta(G) = \min_{w \geq 0} \langle w, \mathbf{1} \rangle \quad \text{s.t.}\quad Yw \geq \mathbf{1}_{E}.
\end{align*}
Normalizing $w$ produces a probability distribution $p_x = w_x/\langle w, \mathbf{1} \rangle$ on $\{\pm1\}^{\abs{V}}$. For any cut distribution $p$ and edge $e$, the \emph{edge cut probability} is $\mathrm{Pr}_{e}(p) = \langle Y_{e},p\rangle$. Let $t = 1/\langle w, \mathbf{1} \rangle$, as we then can define the Fair Cut Cover problem by taking the reciprocal of the Fractional Cut Cover as follows:
\begin{align*}
    \bar{\eta}(G) = \max_{p} t \quad \text{s.t.}\quad Yp \geq t\mathbf{1}_{E}, \quad \sum_{x}p_{x} = 1.
\end{align*}
This leads to the maximin \emph{Fair Cut Cover} objective: learn a cut distribution that maximizes the
smallest edge cut probability,
\begin{align}
    \bar{\eta}(G) = \max_{p}\min_{e \in E} \langle Y_{e},p\rangle = \frac{1}{\eta(G)}.
\end{align}
This transformation has been mentioned by \cite{neto2019fractional} as (P3).
\begin{lemma}
    For any even function \(f\) and cut distribution \(q\) there exists a \(\mathbb{Z}_2\)-symmetric cut distribution \(p\) such that
    \begin{align}
        \mathbb{E}_{x\sim p}[f(x)] = \mathbb{E}_{x\sim q}[f(x)]
    \end{align}
\end{lemma}
\begin{proof}
    Given a cut distribution $q$, we can construct a $\mathbb{Z}_{2}$-symmetric cut distribution $p$ where
    \begin{align}
        p(x) = p(- x) = \frac{q(x) + q(- x)}{2}, \quad \forall x \in \{\pm1\}^{n}
    \end{align}
    As $f$ is an \BBI{even} function, this gives $\mathbb{E}_{x\sim p}[f(x)] = \mathbb{E}_{x\sim q}[f(x)]$
\end{proof}
This primal LP formulation has an exponential number of variables (scales with the number of possible unique cuts) and a linear number of constraints (scales with the number of edges). This can be transformed into the dual LP formulation with edge weight $q \in \Delta(\abs{E})$ and variable $z$.
\begin{equation}
    \label{eq:dual_form}
    \begin{aligned}
        \min_{q} z & \\ 
        \text{s.t: } & \sum_{e \in E} Y_{ex}q_{e} \leq z, \quad \forall x \in \{\pm 1\}^{\abs{V}}\\
        & \sum_{e \in E} q_{e} = 1\\
        & q \geq 0
    \end{aligned}  
\end{equation}

Notice that $q \in \Delta(E)$ is a probability distribution over edges. For any fixed $q$, the quantity $\sum_{e \in E} Y_{ex}q_{e}$ is exactly the weight of cut $x$ in the graph with edge weights $q$. Therefore, checking whether a candidate point $(q,z)$ satisfies all cut inequalities amounts to computing the maximum cut value of the weighted graph $(G,q)$. In this sense, the separation problem for the dual feasible region is exactly the weighted MaxCut problem. Let $K$ denote the feasible space of the dual program; we have the following remark.
\begin{remark}
    \label{rm:sep_oracle_MC}
     Given a point $s=(q,z)$ \BBI{$\in \mathbb{R}^{|E|+1}$}, constructing the separation oracle \cite{grotschel2012geometric} for $K$ is polynomially equivalent to solving weighted MaxCut.
\end{remark}
\begin{proof}
    Given point $s = (q, z)$, the oracle should return either $s \in K$ or return \BBI{a separating hyperplane, i.e.,} a pair $(a, b)$ such that where $a^\top s > b$ and $a^\top x \leq b$ for all $x$ \BBI{$\in K$}.

    First, we start with the simplex constraints. For any point $s = (q, z)$, if $\sum_e q_e \neq 1$, the constraint is violated, and the oracle returns $a$ correspondingly. We first consider the case when $b = 1$ and $\sum_e q_e > 1$, let $a = (\mathbf{1}_e, 0)$ then, for $x = (q', z') \in K$ , $\langle \mathbf{1}_e, q'\rangle = \sum_e q'_e = 1$ but $\langle \mathbf{1}_e, q\rangle > 1$. 
    
    A similar argument holds when considering $b = -1$ and $\sum_e q_e < 1$, in this case we can choose $a = (-\mathbf{1}_e, 0)$. N
    
    Next, suppose $q_j<0$ for some edge $j$. Let $\mathbf e_j$ denote the corresponding standard basis vector, and choose $a=(-\mathbf e_j, 0)$ and $b= 0$. For every feasible point $x=(q',z')\in K$, non-negativity of $q'_j$ gives $a^\top x=-q'_j \le 0$. For the query point $s=(q,z)$, however, $a^\top s=-q_j>0$. Therefore, $a^\top x\le b<a^\top s$ for every $ x\in K$ so the pair $(a,b)$ defines a valid separating hyperplane.
    
    Second, let us consider the cut inequalities. Define $M(q) \coloneq \max_{x} \sum_e Y_{ex}q_e$, it is easy to see that if $z \geq M(q)$, then $z \geq \sum_e Y_{ex}q_e, \forall x$. On the other hand, if $z < M(q)$, then there exists a cut $x^* = \arg \max_x\sum_e Y_{ex}q_e$ that violates the constraint. Let $a = (Y_{\cdot,x^{*}}, -1)$; for $ x \in K$, $(Y_{\cdot,x^{*}})^\top q_x - z_x \leq 0$ and at the query point $s$, $a^\top s=(Y_{\cdot,x^{*}})^\top q_s - z_s > 0$ which create the separation $a^\top x \leq 0 < a^\top s$. Therefore, to create the separation oracle for $K$, we need the ability to find $M(q)$, which is solving weighted MaxCut for a given $q$. Consequently, solving weighted MaxCut yields a separation oracle for the feasible set
\end{proof}
Based on this Remark \ref{rm:sep_oracle_MC},  the strong separation oracle from the Gr\"otschel--Lov\'asz-Schrijver theorem \cite{grotschel2012geometric} on the ellipsoid algorithm can be used to show the equivalence between the hardness of weighted MaxCut and the hardness of the dual form of Fair Cut Cover, which is NP-hard. Moreover, because the graph is finite, both the primal and dual are finite-dimensional linear programs. The primal is feasible and bounded; hence strong duality follows from standard LP duality.
\begin{remark}
    \label{rm:strong_duality}
    The primal and dual form of Fair Cut Cover has strong duality.
\end{remark}
\begin{proof}
For graph \(G=(V,E)\), the number of cuts is finite, so the primal Fair Cut Cover formulation is a finite-dimensional linear program. It is feasible, since any cut distribution \(p\) together with \(t=0\) satisfies the constraints, and it is bounded above because each edge-cut probability lies in \([0,1]\), hence \(t \le 1\). Therefore, standard LP strong duality implies that the dual attains the same optimal value.
\end{proof}

By complementary slackness, if \((p^\star,t^\star)\) and \((q^\star,z^\star)\) are optimal primal and dual solutions, then any cut \(x\) in the support of \(p^\star\) must satisfy the corresponding tight dual constraint
\[
\sum_{e \in E} Y_{e x} q^\star_e = z^\star.
\]
Equivalently, the support of an optimal Fair Cut Cover distribution can be chosen from cuts that are optimal solutions of the weighted MaxCut instance defined by \(q^\star\). \BBI{Moreover, complementary slackness also gives us the following equality 
\begin{align*}
    q^{\star}_{e}(\mathrm{Pr}_e(p^{\star}) - t^{\star}) = 0
\end{align*}
Therefore, if $q^{\star}_{e} > 0$, then the edge is a bottleneck with $\mathrm{Pr}_e(p^{\star}) = t^{\star}$. On the other hand, if the edge cut probability $\mathrm{Pr}_e(p^{\star}) > t^{\star}$ is overcovered, then the dual edge weight $q^{\star}_e = 0$.}
\BBI{
Also, it is worth noticing that even though there might be exponentially many variables in our LPs, we only need a linear number of support with respect to the number of edges $\abs{E}$ through the use of slack variables $s \geq 0$, i.e.,  $Yp - s = t\mathbf{1}_{E}$. 
\begin{remark}
    \label{rm:linear_size_support}
    The Fair Cut Cover LP has exponentially many candidate cuts, but it admits an optimal solution supported on at most $\abs{E}$ cuts (having $\abs{E}$ constraints) as shown in~\cite{salem2024expected}. The computational challenge is identifying suitable support cuts and weights. Furthermore, any optimal cut distribution with support $>\abs{E}$ can be compressed into a cut distribution with at most $\abs{E}$ cuts. \BBI{By complementary slackness, the edge-incidence vectors of its support cuts lie in a common affine hyperplane, so whenever more than $\abs{E}$ cuts remain, an affine dependence lets us adjust their weights until at least one vanishes without changing any edge-cut probability.} 
\end{remark}
}

Furthermore, recent work places Fair Cut Cover in the approximation landscape of Grothendieck-cover problems. 
\begin{remark}[Problem complexity \cite{benedetto2024generalized,benedetto2025primal}]
    Fair-cut cover is $\mathsf{APX-Complete}$. The SDP with hyperplane rounding achieves the Goemans–Williamson approximation factor $\alpha_{\mathrm{GW}}$, which is optimal among polynomial-time classical algorithms in the worst case under the Unique Games Conjecture. This worst-case optimality permits improvements on particular graph families, including the instance-wise separations from hyperplane rounding established here.
\end{remark}
Second, using the framework of gauge duality between the objective of MaxCut and Fractional Cut Cover, \cite{vsamal2015cubical, neto2019fractional} can show that
\begin{align}
     \mathrm{MAXCUT}(G)\eta(G) \geq \abs{E}
\end{align}
Consequently, this gives Fair Cut Cover a non-trivial upper bound
\begin{align}
    \bar{\eta}(G) \leq \frac{\text{MAXCUT}(G)}{\abs{E}}  
\end{align}
Moreover, the above equality holds for edge-transitive graphs through the following proposition, which tightly connects MaxCut and Fair Cut Cover.
\begin{proposition}[Lemma 2.1 \cite{vsamal2015cubical}]
    \label{coro:edge_transitive}
    If the considered graph $G$ is edge-transitive, $\bar{\eta}(G) = \frac{\mathrm{MaxCut}(G)}{\abs{E}}$ and the probability of all edges being cut is equal. \BBI{An optimal distribution can be obtained by applying a uniformly random element of an edge-transitive subgroup $(\Gamma \leq \operatorname{Aut}(G))$ to this cut. Implementing this symmetry-averaging construction requires access to suitable automorphisms and a procedure for sampling them uniformly. Note that even knowing the graph is edge-transitive, finding the automorphism group of this is non-trivial~\cite{mckay2014practical}. Thus, knowing a maximum cut determines the optimal value, while constructing an optimal distribution entails an additional step of balancing the edge cut probabilities.}
\end{proposition}

\section{Semidefinite Programming for Fair Cut Cover}

\label{appendix:SDP}
Based on the fact that Fractional Cut Cover is the anti-blocker and gauge dual of MaxCut, \cite{vsamal2015cubical, neto2019fractional, benedetto2025primal} defines an SDP formulation for Fractional Cut Cover on a graph $G(V,E)$ as
\begin{equation}
    \label{eq:fractional_cut_cover_SDP}
    \begin{aligned}
        \mathrm{SDP}^{\circ}(G) &= \min_{N} (s) \\ 
    \text{s.t: }&  \frac{w_{ij}}{4}(N_{ii} + N_{jj} - 2N_{ij}) \geq \mathbf{1}_{E}, \forall ij \in E \\
                & N \succeq 0, \quad \text{diag}(N) = s\mathbf{1}_{V} 
    \end{aligned}
\end{equation}

Similarly, we can define a Fair Cut Cover SDP based on the Fractional Cut Cover SDP with a positive-semidefinite matrix $C \coloneq N/s$ where $\text{diag}(C) = \mathbf{1}_{V}$, and let $t = \frac{1}{s}$, it can be easily seen that $s \geq 1$ gives $t \geq 0$.

\begin{equation}
    \begin{aligned}
        \mathrm{SDP}(G) \coloneq& \max_{C} (t) \\ 
    \text{s.t: }&  \frac{w_{ij}}{2}(1 - C_{ij}) \geq t\mathbf{1}_{E}, \forall ij \in E \\
                & C \succeq 0, \qquad \text{diag}(C) = \mathbf{1}_{V} 
    \end{aligned}
\end{equation}

These SDPs can be solved using the established Goemans-Williamson scheme \cite{goemans1995improved}. It is proved in \cite{benedetto2024generalized} that under UGC, this GW approach yields the optimal polynomial-time worst-case approximation ratio. By embedding the PSD matrix $C$ onto the $k-$dimensional sphere where $C_{ij} = \langle v_{i}, v_{j}\rangle$ and $v_{i} \in \mathbb{R}^{k}$, we can either solve this SDP directly using methods such as interior point methods with $k = \abs{V}$ or approximate the SDP by performing rank-$k$ relaxation on the sphere such aBurer-Monteiro\cite{burer2002rank}. 
The SDP embedding itself is sufficient to construct a cut distribution via a process called hyperplane rounding \cite{goemans1995improved}. More details about hyperplane rounding are provided in Appendix \ref{appendix:sdp_rounding}.

As suggested by \cite{goemans1995improved}, the random hyperplane rounding procedure is used to sample from the implicit representation of cuts when solving the SDP; see Appendix \ref{appendix:sdp_rounding}. Therefore, the Fractional Cut Cover/ Fair Cut Cover value is defined based on the cut distribution $p$ yielded by GW hyperplane rounding on the optimal SDP result. Next, we look at the monotonicity of the SDP solver and prove it as follows
\begin{lemma}[Monotonicity of SDP solver]
    \label{lem:monotonicity_SDP}
    Let $G=(V,E,w)$ be a weighted graph and $H=(V_H,E_H,w)$ a subgraph with $V_H\subseteq V$ and $E_H\subseteq E$. Let $\mathrm{SDP}(G)$ and $\mathrm{SDP}(H)$ denote the optimal values yielded by SDP when applied to Fair Cut Cover. Then $\mathrm{SDP}(G) \le \mathrm{SDP}(H)$. In particular, if $H$ attains the minimum SDP value among all subgraphs of $G$, i.e.,
    \begin{align*}
        \mathrm{SDP}(H) = \min\{\mathrm{SDP}(H') : H' \subseteq G\} =: \Gamma,
    \end{align*}
    then $\mathrm{SDP}(G) = \Gamma$, and there exists an optimal solution $(t^*,C^*)$ for $G$ with $t^*=\Gamma$ whose vertex restriction $C^*[V_H, V_H]$ is also optimal for $H$.
\end{lemma}
\begin{proof}
    Given any feasible solution $(t, C)$ for $G$, restricting $C$ to the vertex set of $H$, we obtain the principal sub-matrix $C^{H} \coloneqq C[V_H, V_H]$ which is positive semi-definite (principal submatrix property), $\text{diag}(C^{H}) = \mathbf{1}$ and the constraints $\frac{w_{ij}}{2}(1-C^{H}_{ij}) \geq t \geq 0, \forall (ij) \in E_H$ are satisfied. Therefore, any feasible solution $(t, C)$ for the SDP of original $G$ constitutes a feasible solution $(t, C^H)$ for the SDP of the subgraph $H$ and its optimum $\mathrm{SDP}(H) \geq t$. By taking over the supremum over all such $t$ we have the inequality
    \begin{align*}
        \mathrm{SDP}(G) \leq \mathrm{SDP}(H) 
    \end{align*}
    Next, assume that $\mathrm{SDP}(H) = \min\{\mathrm{SDP}(H') : H' \subseteq G\} =: \Gamma$, since $G \subseteq G$ is a subgraph of itself, we have that $\Gamma \leq  \mathrm{SDP}(G)$, combining with the inequality above, we have that $\Gamma = \mathrm{SDP}(G)$. Finally, as the feasible region of the SDP on $G$ is nonempty, closed, and bounded (in particular, $C\succeq 0$ with $\mathrm{diag}(C)=\mathbf{1}$ implies $|C_{ij}|\le 1$, and $t$ is bounded above by $\max_{(i,j)\in E} w_{ij}$ and is bounded below by $0$), by the Weierstrass extreme value theorem, the optimal value $\mathrm{SDP}(G)=\Gamma$ is attained by some feasible pair $(t^*,C^*)$ on $G$ and the $(t^*, C^*[V_H, V_H])$ is the optimal solution on $H$. When the subgraph $H$ contains only edge deletions, without vertex deletions, this result follows trivially, since $C$ remains a feasible solution but with weaker constraints.
\end{proof}
We now mention some of the known characteristics of the given SDPs, which were mentioned in \cite{benedetto2025primal} for $\mathrm{SDP}^{\circ}(G)$ due to its anti-blocker and gauge duality to MaxCut%
, which gives the corresponding characteristic for $\mathrm{SDP}(G)$
\begin{align}
    \mathrm{SDP}^{\circ}(G) &\leq \eta(G) \leq \frac{1}{\alpha_{\mathrm{GW}}}\mathrm{SDP}^{\circ}(G) \\
    \alpha_{\mathrm{GW}}\mathrm{SDP}(G)  &\leq \bar{\eta}(G) \leq \mathrm{SDP}(G)
\end{align}
Furthermore, based on the gauge duality with MaxCut, denote $\mathrm{SDP}_{\mathrm{MC}}$ as the MaxCut SDP, and we can solve it optimally; \cite{vsamal2015cubical, neto2019fractional} show that
\begin{align}
    \mathrm{SDP}_{\mathrm{MC}}(G)\mathrm{SDP}^{\circ}(G) \geq \abs{E}
\end{align}
Consequently, we also have the upper bound for Fair Cut Cover SDP as follows:
\begin{align}
    \mathrm{SDP}(G) \leq \frac{\mathrm{SDP}_{\mathrm{MC}}(G)}{\abs{E}}
\end{align}

\section{SDP + hyperplane rounding}
\label{appendix:sdp_rounding}
As shown in ~\cite{benedetto2024generalized, benedetto2025primal}, we can use ``Grothendieck cover'' to prove that the Fractional Cut Cover problem is $\mathsf{APX-Complete}$ and show that SDP with hyperplane rounding is the optimal polynomial-time worst-case approximation algorithm under UGC. We direct the reader to the cited papers for more details on the complexity proof. Here, we look into the algorithms and derive some of their characteristics.
\begin{definition}[Gaussian hyperplane rounding \cite{goemans1995improved}]
    \label{def:g_hr}
     Let $G = (V, E)$ be an unweighted graph and let $p \in \Delta(\{\pm1\}^{\abs{V}})$ be any distribution over cuts. Assume we already have $\{v_i\}_{i\in V}\subset \mathbb{R}^d$ be an SDP embedding with $\norm{v_i}_2=1$ for all $i$. Draw $g \sim \mathcal{N}(0,\mathbf{1}_d)$ and define the random cut $X(g)\in\{\pm 1\}^\abs{V}$ by $X_i(g)=\mathrm{sign}(\langle g,v_i\rangle)$. 
     Let $\mathcal{D}_{\mathrm{HR}}$ be the space of cut distributions corresponding to hyperplane rounding, and denote the cut distribution from rounding $p_{\mathrm{HR}} \in \mathcal{D}_{\mathrm{HR}}$. The edge cut probability of each edge $(i, j)$ can be expressed as
    \begin{align*}
        \mathrm{Pr}_e(p_{\mathrm{hr}}) &= \mathbb{E}_{X\sim p_{\mathrm{hr}}}\left[\mathds{1}[X_i\neq X_j]\right] 
        \\ &= \mathrm{Pr}_{g\sim\mathcal{N}(0,\mathbf{1}_d)}\left[\langle g, v_{i}\rangle \langle g, v_{j}\rangle < 0\right] 
        \\ &= \frac{\arccos(\langle v_{i}, v_{j}\rangle)}{\pi}
    \end{align*}
\end{definition}
Throughout, \(\mathcal D_{\mathrm{HR}}\) denotes the family of distributions obtained by Gaussian hyperplane rounding of a SDP embedding. The embedding may be chosen arbitrarily and optimized separately for each graph, but remains fixed across samples. Our separation results concern this family. Random mixtures of embeddings, reweighting of sampled cuts, and additional classical postprocessing are outside this comparison.
\BBI{
\begin{proposition}[Optimal hyperplane rounding and subgraph monotonicity]
\label{prop:sdp_optimal_HR}
Let $G=(V,E)$ be an unweighted, connected graph. Let $\mathrm{SDP}(G)$ denote the optimal value of Equation~\eqref{eq:fair_cut_cover_SDP}, and let $\mathrm{SDP}_{\mathrm{HR}}(G)$ denote the optimal Fair Cut Cover value attainable by Gaussian hyperplane rounding of an SDP embedding.

Every optimal SDP embedding is optimal for hyperplane rounding,
and
\begin{equation}
\label{eq:sdp_hr_identity}
    \mathrm{SDP}_{\mathrm{HR}}(G)
    =
    \frac{1}{\pi}
    \arccos\!\left(1-2\mathrm{SDP}(G)\right).
\end{equation}
Furthermore, for every connected subgraph $H\subseteq G$,
\begin{equation}
    \mathrm{SDP}_{\mathrm{HR}}(G)
    \leq
    \mathrm{SDP}_{\mathrm{HR}}(H).
\end{equation}
Consequently, writing $n=|V|$ and letting $\omega(G)$ denote
the maximum clique size,
\begin{equation}
    \frac{1}{\pi}\arccos\!\left(-\frac{1}{n-1}\right)
    \leq
    \mathrm{SDP}_{\mathrm{HR}}(G)
    \leq
    \frac{1}{\pi}\arccos\!\left(-\frac{1}{\omega(G)-1}\right).
\end{equation}
\end{proposition}

\begin{proof}
For any PSD embedding $C$, Gaussian hyperplane rounding cuts
an edge $ij\in E$ with probability
$\arccos(C_{ij})/\pi$. We define
\begin{align*}
     h(t)\coloneq \frac{1}{\pi}\arccos(1-2t),
    \qquad t\in[0,1].
\end{align*}
Since $h$ is strictly increasing, we have
\begin{equation}
    \begin{aligned}
    \mathrm{SDP}_{\mathrm{HR}}(G)
    &=
    \max_{C}
    \min_{ij\in E}
    h\left(\frac{1-C_{ij}}{2}\right)\\
    &=
    \max_{C}
    h\left(
        \min_{ij\in E}\frac{1-C_{ij}}{2}
    \right)\\
    &=
    h\left(
        \max_{C}
        \min_{ij\in E}\frac{1-C_{ij}}{2}
    \right)\\
    &=h\left(\mathrm{SDP}(G)\right).
    \end{aligned}
\end{equation}
which proves Equation~\eqref{eq:sdp_hr_identity}.
Strict monotonicity of $h$ also shows that the two optimization problems have the same optimal Gram matrices.
By Lemma~\ref{lem:monotonicity_SDP},
$\mathrm{SDP}(G)\leq\mathrm{SDP}(H)$ whenever $H\subseteq G$.
Applying the increasing function $h$ gives
\begin{align*}
    \mathrm{SDP}_{\mathrm{HR}}(G)
    =
    h\left(\mathrm{SDP}(G)\right)
    \leq
    h\left(\mathrm{SDP}(H)\right)
    =
    \mathrm{SDP}_{\mathrm{HR}}(H)
\end{align*}
As shown in Remark \ref{rm:lovasz_theta_number}, the optimal solution of the Fair Cut Cover SDP on edge-transitive graphs can be calculated from the Lov{\'a}sz number. Therefore,
\begin{align*}
    \mathrm{SDP}_{\mathrm{HR}}(K_r)
    =
    \frac{1}{\pi}\arccos\!\left(-\frac{1}{r-1}\right)
\end{align*}
Since $K_{\omega(G)}\subseteq G\subseteq K_n$,
subgraph monotonicity yields
\begin{align*}
    \mathrm{SDP}_{\mathrm{HR}}(K_n)
    \leq
    \mathrm{SDP}_{\mathrm{HR}}(G)
    \leq
    \mathrm{SDP}_{\mathrm{HR}}(K_{\omega(G)}),
\end{align*}
which proves the stated bounds.
\end{proof}
}

\begin{lemma}[Theorem \ref{thm:dist_subset}]
    Let \(\mathcal{D}\) be the space of $\mathbb{Z}_{2}$-symmetric cut distributions of a graph with node size at least $4$ and \(\mathcal{D}_{\mathrm{HR}}\) be the space of cut distributions corresponding to hyperplane rounding. Then \(
    \mathcal{D}_{\mathrm{HR}} \subsetneq \mathcal{D}.\)
\end{lemma}
\begin{proof}

    First, we show that every distribution \(p \in \mathcal{D}_{\mathrm{HR}}\) is \(\mathbb{Z}_2\)-symmetric. Let \(X \in \mathbb{R}^{n\times k}\) be an embedding that can produce \(p,\) that is for all \(x\) we have \(p(x) = \mathrm{Pr}[\mathrm{sign}(X\varepsilon) = x]\) where \(\varepsilon \sim \mathrm{Normal}(0, I).\) Because of the symmetry of the normal distribution, \(-\varepsilon\) is equivalent as a random variable to \(\varepsilon,\) so
    \begin{align*}
        p(x) &= \mathrm{Pr}[\mathrm{sign}(X\varepsilon) = x]\\ 
        &= \mathrm{Pr}[\mathrm{sign}(X(-\varepsilon)) = x]\\ 
        &= \mathrm{Pr}[\mathrm{sign}(X\varepsilon) = - x] = p(- x).
\end{align*}

    To show that \(\mathcal{D}_{\mathrm{HR}}\) is a strict subset, we look at an extreme case where $\tilde{C}$ is an equicorrelation matrix with constant $\rho$ for off-diagonal entries. Note that this is equivalent to considering complete graphs $K_n$ with $n > 2$ nodes. It is known that $\tilde{C}$ is positive semi-definite if and only if $-\frac{1}{n-1}\leq \rho \leq 1$. Let's check if $\rho < -\frac{1}{n-1}$ yields a valid cut distribution. We have
    \begin{align}
        \cos(\pi P_{ij}) = \rho& < -\frac{1}{n-1} \\\
            P_{ij} &> \frac{1}{\pi}\arccos(-\frac{1}{n-1})
    \end{align}
    As from the analysis of the complete graph in Lemma \ref{prop:best_Kn_value}, we can also see that if we put an upper bound on $P_{ij} \leq \frac{n}{2(n-1)}$ for $n$ even and $P_{ij} \leq \frac{n+1}{2n}$ for $n$ odd, which also satisfies the odd-cycle inequality and yields that the grand sum of $P < MAXCUT(K_n)$ constituting a valid cut distribution. This proves the lemma and also explains the gap for the complete graph. As a result, we can also conclude that SDP + randomized hyperplane rounding fails to capture the optimal Fair Cut Cover distribution of $K_n$ for $n \geq 4$
\end{proof}
In graph theory, the Lov{\'a}sz theta function or Lov{\'a}sz number, $\vartheta(G)$ was introduced in \cite{lovasz1979shannon} as an upper bound on the Shannon capacity of the graph. Accurate numerical approximations of the Lov{\'a}sz theta function can be solved in polynomial time as the solution of the vector $k$-coloring SDP defined in \cite{karger1998approximate}. This vector $k$-coloring SDP is the same as our SDP for fair cut coverage when considering edge-transitive graphs. Therefore, when considering structured families of graphs such as the complete graph $K(n)$, the cycle graph $C(n)$, and the Kneser graph $KG(n, k)$, we can use the established Lov{\'a}sz number on the complement graph to get the optimal SDP solution.
\begin{remark}
    \label{rm:lovasz_theta_number}
     The solution to the SDP for the Fair Cut Cover for structured families of edge-transitive graphs $G$ is $t^* = -1/(\vartheta(\bar{G})-1)$ with the objective of $(1 - t^{*})/2$, where the Lov{\'a}sz number on the complement graph $\bar{G}$ are as follows
     \begin{enumerate}
         \item Complete graph: $\vartheta(\bar{K_n}) = n$
         \item Cycle graph:\\ $\vartheta(\bar{C_n}) =  2 \text{ if }n \text{ is even}, \frac{1+\cos(\pi/n)}{\cos(\pi/n)} \text{ if }n \text{ is odd}$
         \item Kneser graph: $\vartheta(\bar{KG}(n, k)) = \frac{n}{k}$
     \end{enumerate}
\end{remark}
Based on this optimal solution $t^{*}$ of the SDP, we can get the best expected value when performing GW hyperplane rounding through the Gaussian distribution. More specifically, the expected Fair Cut Cover value from the empirical sampled  distribution is
\begin{align}
    \mathrm{SDP}_{\mathrm{HR}}(G) = \frac{1}{\pi}\arccos{(t^{*})}
\end{align}

\begin{lemma}[Theorem \ref{thm:Kn_separation}]
        For any complete graph \(K_n,\) with $n\geq 2$
\[
\mathrm{Q}_1^{\mathrm{std}}(K_n) \geq \mathrm{SDP}_{\mathrm{HR}}(K_n) = \frac{1}{\pi}\arccos\!\Big(\frac{1}{1-n}\Big).
\]
where the strict inequality holds when $n \geq 4$,
\end{lemma}
\begin{proof}
    For complete graphs, a \(k=1\) standard QAOA circuit with parameters \(\gamma,\beta\) produces a correlator $\langle{Z_uZ_v}\rangle_{\gamma, \beta}$ which gives distrubtion that cuts an edge \(uv\) with probability $(1 - \langle{Z_uZ_v}\rangle_{\gamma, \beta})/2$
    \begin{align}        
    \label{eq:p=1_K_n}
    \langle{Z_uZ_v}\rangle_{\gamma, \beta} &= -\sin(4\beta)\sin(\gamma)\cos^{n-2}(\gamma) +  \frac{1}{2}\sin^{2}(2\beta)(1-\cos^{n-2}(2\gamma))
    \end{align}
    \cite{shaydulin2021exploiting, wang2018quantum}.
    Therefore, we have the expected Fair Cut Cover value for $k=1$ QAOA as $\mathrm{Q}^{\mathrm{std}}_{1}(K_n) = (1-\min_{\gamma, \beta}
    \langle{Z_iZ_j}\rangle)/2$. Here, we want to prove the fact that there exist $\gamma, \beta$ such that 
    \begin{equation}
        \label{eq:Q1_geq_sdp}
        \begin{aligned}
            \frac{1 - \langle{Z_iZ_j}\rangle_{\gamma, \beta}}{2} &> \frac{1}{\pi}\arccos\left(-\frac{1}{n-1}\right)\\
            & = \frac{1}{2} + \frac{1}{\pi}\arcsin\left(\frac{1}{n-1}\right)
        \end{aligned}
    \end{equation}
    which is equivalent to proving that there exists
    \begin{align*}
        -\frac{\langle{Z_iZ_j}\rangle_{\gamma, \beta}}{2} > \frac{1}{\pi}\arcsin\left(\frac{1}{n-1}\right)
    \end{align*}
    We first denote $f(\gamma, \beta) \coloneq \frac{-\langle{Z_iZ_j}\rangle}{2}$. Following the proof strategy of \cite{bae2024recursive} with the optimal $\beta_{\gamma}^{*}$, we have $f(\gamma, \beta_{\gamma}^{*}) \coloneq (1/8)\left(\sqrt{A^{2}(\gamma) + B(\gamma)} - A\right)$, and $d \coloneq n - 2$
    \begin{align}
        A(\gamma) &\coloneq 1 - (2\cos^{2}(\gamma) - 1)^{d} \\
        B(\gamma) &\coloneq 16\sin^{2}(\gamma)\cos^{2d}(\gamma)
    \end{align}
    Next, we try to prove there exist $(\gamma, \beta)$ such that $f(\gamma, \beta_{\gamma}^{*})> \frac{1}{\pi}\arcsin\left(\frac{1}{n-1}\right)$, the inequality can be written as
    \begin{equation}
     \begin{aligned}
        (1-\cos^{2}(\gamma))\cos^{2}(\gamma)^{d} &> \frac{4\arcsin^{2}\left(\frac{1}{d+1}\right)}{\pi^2}
        + \frac{1}{\pi} \arcsin\left(\frac{1}{d+1}\right)\left(1-(2\cos^{2}(\gamma)-1)^{d}\right)
    \end{aligned}   
    \end{equation}
    Here to simplify things, we first denote $x \coloneq \frac{1}{d+1}$ and prove that by picking $\cos^{2}(\gamma)= 1-x$, the following strict inequality holds for all $d \geq 2$ or equivalently when $0 < x \leq 1/3$,
    \begin{equation}
     \begin{aligned}
        x(1- x)^{d} &> \frac{4\arcsin^{2}(x)}{\pi^2}
        + \frac{1}{\pi} \arcsin(x)\left(1-(1-2x)^{d}\right)
    \end{aligned}   
    \end{equation}
    \BBI{
    We denote $h(x) \coloneq (1-x)^{1/x - 1}$, then LHS $= xh(x)$. Furthermore, we can see that this function is an increasing function as
    \begin{equation}
       \begin{aligned}
            h'(x) = h(x)\frac{-\ln(1-x) - x}{x^2} > 0, \quad 0 <  x \leq 1/3
        \end{aligned} 
    \end{equation}
    Furthermore, 
    \begin{equation}
        \begin{aligned}
            h'(x) &= \frac{h(x)}{x^2}\int^{x}_{0}\frac{t}{1-t}\,dt \\
            &\leq \frac{h(x)}{x^2}\frac{1}{1-x}\int^{x}_{0}t \,dt = \frac{h(x)}{2(1-x)}
        \end{aligned}
    \end{equation}
    As $h(x) \leq h(1/3)$, we have the bound for the derivative $h'(x) \leq \frac{h(1/3)}{2(1-1/3)}= 1/3$.
    
    For the RHS, we first look at the upper bound of $1-(1-2x)^{d}$, which is a decreasing function of $d$. Therefore, as $d \geq 2$, we have 
    $1-(1-2x)^{d} = 1 - (1 - 2/(d+1))^{d} \leq 8/9$ where the equality holds when $d = 2$. 
    Second, we bound the $\arcsin(x) \leq x + x^{3}/3$ when $0 <  x \leq 1/3$, which gives the total upperbound of
    \begin{align}
        \label{eq:bound_RHS}
        \mathrm{RHS} \leq \frac{4}{\pi^2}\left(x + \frac{x^3}{3}\right)^2  + \frac{8}{9\pi}\left(x + \frac{x^3}{3}\right)
    \end{align}
    We then denote function $g(x)$ as follows
    \begin{align}
        g(x) \coloneq \frac{4}{\pi^2}x\left(1 + \frac{x^2}{3}\right)^2  + \frac{8}{9\pi}\left(1 + \frac{x^2}{3}\right)
    \end{align}
    Then the bound for the RHS as in Equation \eqref{eq:bound_RHS} can be written as RHS $\leq xg(x)$. We note that this $g(x)$ function is also an increasing function as
    \begin{align}
        g'(x) = \frac{4}{\pi^2}\left(1+ 2x^2 + \frac{5x^4}{9}\right) + \frac{16x}{27\pi} \geq \frac{4}{\pi^2} > \frac{1}{3}
    \end{align}
   Therefore, the gap function $h(x) - g(x)$ is a strictly decreasing function on $0 < x \leq 1/3$ as $h'(x) - g'(x) < 0$. The minimum of the gap function therefore occurs at $x = 1/3$, where at that endpoint we have
   \begin{align}
       h\left(\frac{1}{3}\right) - g\left(\frac{1}{3}\right) > 0
   \end{align}
    Therefore, for every integer $d \geq 2$, the following inequality holds
    \begin{equation}
     \begin{aligned}
        \mathrm{LHS} &= xh(x)\\
        &> xg(x) \\
        &\geq \mathrm{RHS}
    \end{aligned}   
    \end{equation}
    }
    This proves that there exists $(\gamma, \beta)$ such that Equation (\ref{eq:Q1_geq_sdp}) holds. As we already showed, the inequality is true for $K_n$ where $n \geq 4$, we then have the strict inequality
    \begin{align}
        \mathrm{Q}^{\mathrm{std}}_{1}(K_n) > \mathrm{SDP}_{\mathrm{HR}}(K_n)
    \end{align}
    For $K_3$ (triangle case), $\mathrm{Q}^{\mathrm{std}}_{1}(K_3) = \mathrm{SDP}_{\mathrm{HR}}(K_3) = 2/3$, \BBI{whike for $K_2$ (single edge case), $\mathrm{Q}^{\mathrm{std}}_{1}(K_2) = \mathrm{SDP}_{\mathrm{HR}}(K_2) = 1.0$}. In Figure \ref{fig:K_n_numeric}, we show the gap for $n$ up to 100.
\end{proof}
\begin{figure}[!ht]
    \centering
\includegraphics[width=.6\linewidth]{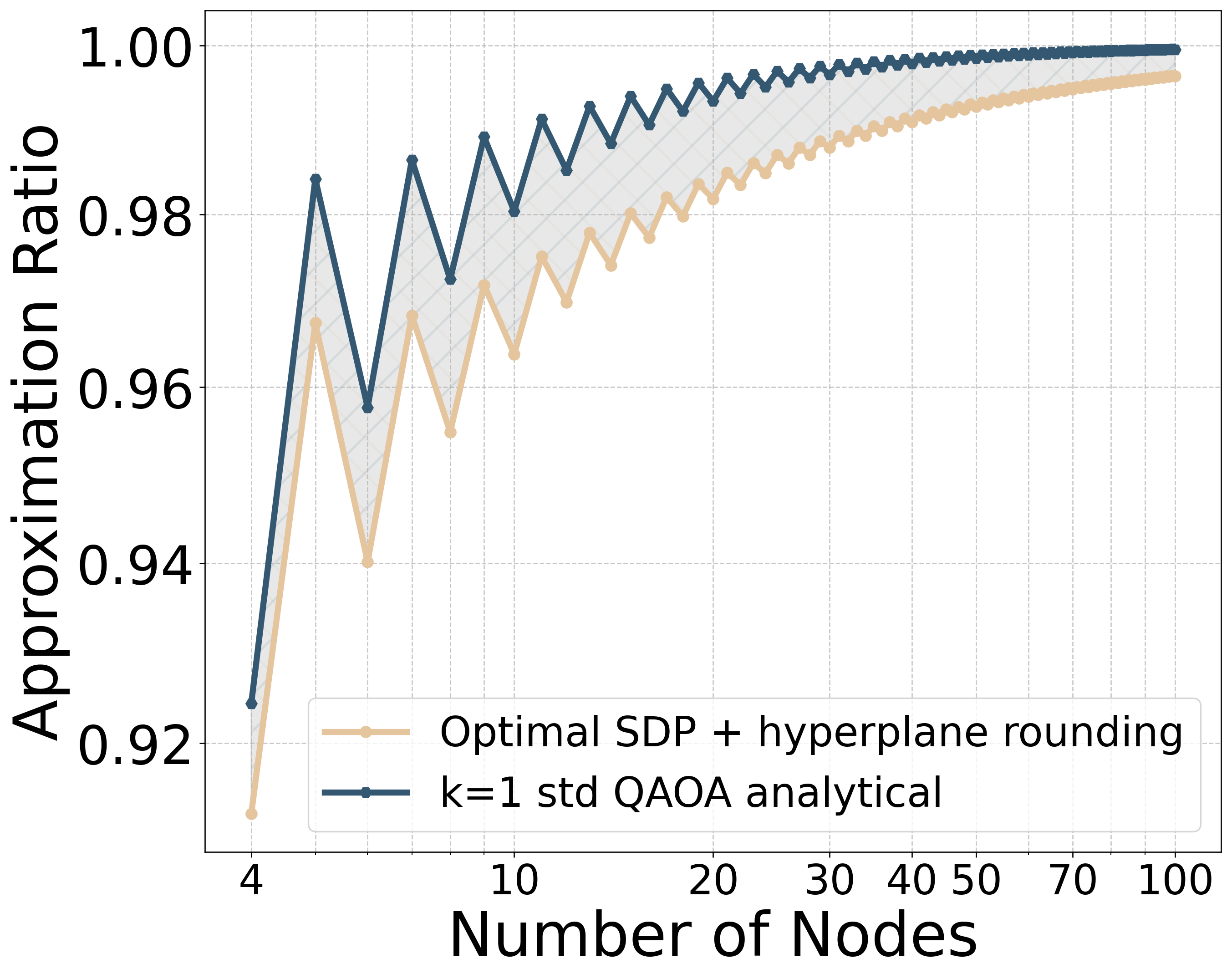}
    \caption{Performance of SDP+ hyperplane rounding and $k=1$ std QAOA for the complete graph families $K_{n}$ for $n \geq 4$}
    \label{fig:K_n_numeric}
\end{figure}
The reason for this gap is that SDP + randomized hyperplane rounding is not sufficient to capture all possible cut distributions. To prove this, we show an example where the optimal cut distribution for Fair Cut Cover on a complete graph can not be reproduced from SDP solutions and randomized hyperplane rounding. First, we show a simple way to construct the optimal cut distribution for Fair Cut Cover on a complete graph
\begin{lemma}
    \label{prop:best_Kn_value}
    The best Fair Cut Cover value on a complete graph $G(V, E)$ comes from a distribution created by using $k$-subset sampling where $k = \frac{\abs{V}}{2}$ if $\abs{V}$ is even and at $k = \frac{\abs{V}\pm1}{2}$ if $\abs{V}$ is odd.
\end{lemma}
\begin{proof}
    First, we start with a simple statement where for any cut distribution $p$, we have the following inequality between the minimum, average, and maximum.
\begin{gather}
    \begin{aligned}
    \min_{e\in E}\mathbb{E}_{X \sim p}[\chi_{e}(X)] &\leq \frac{\sum_{e \in E} \mathbb{E}_{X \sim p}[\chi_{e}(X)]}{\abs{E}} \\ &= \frac{\mathbb{E}_{X \sim p}[\sum_{e \in E} \chi_{e}(X)]}{\abs{E}} \\ &\leq \frac{MAXCUT}{\abs{E}}
    \end{aligned}    
\end{gather}

As we are considering $K_n$, the upper bound becomes 
\begin{equation}
    \begin{aligned}
        \min_{e\in E}\mathbb{E}_{X \sim p}[\chi_{e}(X)]  &= \min_{e\in E} \mathrm{Pr}_{e}(p) \\
        &\leq 
        \begin{cases}
            \frac{n}{2(n-1)}, & \text{if } n \text{ is even} \\
            \frac{n+1}{2n},  & \text{if } n \text{ is odd}
        \end{cases}
\end{aligned}
\end{equation}

where we use the already known MAXCUT for complete graphs. Now, we try to reach this upper bound using just a simple random solution by sampling $k$-subsets $S$ where $S \subset [n]$, which corresponds to a set of multiple cuts. We have the probability of edge $uv$ inside the cut given by number of $k$-subsets with only $u$ or $v$  over $\binom nk$:
\begin{align}
    \mathrm{Pr}_{e}(p) = \frac{2\binom{n-2}{k-1}}{\binom nk} = \frac{2k(n-k)}{n(n-1)}
\end{align}
We can see that the probability of all edges is equivalent, and $\min_{e \in E}\mathrm{Pr}_{e}(p) = \frac{2k(n-k)}{n(n-1)}$. Now we try to find $k \in \mathbb{N}^*$ such that $k$-subset sampling yields a cut distribution $p_{k}$ that maximizes $\min_{e \in E}\mathrm{Pr}_{e}(p_{k})$. This can be obtained by maximizing $2k(n-k)$ which gives $k = \frac{n}{2}$ if $n$ is even and at $k = \frac{n\pm1}{2}$ if $n$ is odd. These values of $k$ give the best fairness of the complete graph that can be achieved using $k$-subset sampling 
\begin{align}
    \min_{e \in E}\mathrm{Pr}_{e}(p)= \begin{cases}
        \frac{n}{2(n-1)}, & \text{if } n \text{ is even} \\
        \frac{n+1}{2n},  & \text{if } n \text{ is odd}
    \end{cases}
\end{align}
which matches the upper bound and yields the best fairness value.
\end{proof}
\BBI{
Besides the cases of complete graphs, we further show the proof of the existence of a separation gap between standard QAOA and SDP for a family of high-girth regular graphs through Theorem~\ref{thm:high-girth-fcc-separation}. We prove it through the following Lemmas
\begin{lemma}
\label{lem:fcc-hg-spectral-upper}
For every $d$-regular graph $G$ with $d\ge1$ and adjacency matrix $A_G$,
\begin{equation}
     \mathrm{SDP}_\mathrm{HR}(G)
     \le \frac{1}{\pi}
     \arccos\left(\frac{\lambda_{\min}(A_G)}{d}\right)
\end{equation}
\end{lemma}
\begin{proof}
Let $G$ be $d$-regular, with $n$ vertices and $m=nd/2$ edges, and let
$A=A_G$. For any PSD embedding $C$,
\begin{equation}
 \max_{ij\in E}C_{ij}
 \ge \frac{1}{m}\sum_{ij \in E}C_{ij}  = \frac{\operatorname{Tr}\left[AC\right]}{nd}
\label{eq:fcc-hg-average-correlation}
\end{equation}
Since $A - \lambda_{\min}(A)I \succeq 0$ and $C \succeq 0$,
\begin{align}
\operatorname{Tr}\left[(A-\lambda_{\min}(A)I)C\right]
 = \operatorname{Tr}\left[
 C^{1/2}(A-\lambda_{\min}(A)I)C^{1/2}\right] \ge 0
\end{align}
Therefore, $\operatorname{Tr}\left[AC\right] \geq \lambda_{\min}(A)\operatorname{Tr}\left[C\right]$. As $C$ is a Gram matrix of unit vectors,$C_{ii} = 1$ for every $i$, then $\operatorname{Tr}\left[C\right] = n$. The equation~\eqref{eq:fcc-hg-average-correlation} then gives
\begin{equation}
 \max_{ij\in E}C_{ij}\ge\frac{\lambda_{\min}(A)}d.
\label{eq:fcc-hg-correlation-certificate}
\end{equation}
Apply the decreasing function $\arccos(\,\cdot\,)/\pi$ and maximize
over $C$ proves that this is an upper bound over all fixed embeddings.
\end{proof}

Next, we point to the non-bipartite cubic Ramanujan graph family for our graph description as in Theorem~\ref{thm:high-girth-fcc-separation}. We then give the proof as follows,
\begin{proof}[Proof of Theorem~\ref{thm:high-girth-fcc-separation}]
The general spectral inequality follows from Lemma~\ref{lem:fcc-hg-spectral-upper}. Since $G$ is connected, nonbipartite, cubic, and Ramanujan, by the definition of a Ramanujan graph, we have that $\lambda_{\min}(A_G)\ge-2\sqrt2$. Thus, 
    \[
     \mathrm{SDP}_{\mathrm{HR}}(G)
     \le \frac1\pi
          \arccos\!\left(\frac{\lambda_{\min}(A_G)}3\right)
     \le \frac{1}{\pi}\arccos\left(-\frac{2\sqrt{2}}{3}\right)
     \]
Next, we consider standard QAOA with $p$ layers; note that each edge expectation depends only on its backward light cone. Moreover, whenever the girth of the graph is big enough, $g(G)>2p+1$,
these light cones are isomorphic edge-rooted cubic trees \cite{basso2021quantum,farhi2025lower}. The initial product state, shared layer parameters, and expectation of edge observables agree under these isomorphisms. Thus, for every parameter tuple $\theta$, every edge has the same cut probability $q_p(\theta)$ from the resulting QAOA state determined entirely by the finite tree light cone and independent of graph size.
\[
 \min_{uv\in E}
 \frac{1-\langle Z_uZ_v\rangle_\theta}{2}
 =\frac{\langle C_{\mathrm{MaxCut}}\rangle_\theta}{|E|} = q_p(\theta)
\]
Consequently, 
\[
    Q_p^{\mathrm{std}}(G_r) =\max_{\theta'}q_p(\theta') \ge q_p(\theta) \qquad\text{whenever }g(G_r)>2p+1.
\]
Furthermore, note that we can generate a non-bipartite cubic Ramanujan graph family by taking $q=2$ in Morgenstern's construction \cite{morgenstern1994existence} gives an explicit family of connected, non-bipartite cubic Ramanujan graphs $G_r$, indexed by even integers $r$, with
\begin{align*}
    \abs{V(G_r)}=2^{3r}-2^r,
    \qquad
    g(G_r)\ge
    \left\lceil
        \frac{2}{3}\log_2(2^{3r}-2^r)
    \right\rceil
    = 2r
\end{align*}
where $g(G_r)$ denotes girth of graph $G_r$. Now set $p=15$. The girth condition holds for every member of the family for every even $r\ge16$ as $g(G_r)\ge32$. The numerical evaluation reported in Table~1 of \cite{farhi2025lower} gives a fixed parameter tuple $\theta_{15}$ with
$q_{15}(\theta_{15})\approx0.8935$. Using this reported evaluation,
$q_{15}(\theta_{15})> \frac{1}{\pi}\arccos\left(-\frac{2\sqrt{2}}{3}\right) \approx 0.8918$. Hence, for every even $r \ge 16$,
\[
    Q_{15}^{\mathrm{std}}(G_r)
    \ge q_{15}(\theta_{15}) > \mathrm{SDP}_{\mathrm{HR}}(G_r).
\]

The depth and parameter tuple are the same for all
members of the family.
\end{proof}
}

\section{Sampling cut distribution to approximate edge cut probability}
\label{appendix:sampling}
Given an unweighted graph $G(V, E)$ and $p \in \Delta$ be any distribution over cuts. For each edge $e = (u, v) \in E$, we have cut probability $\mathrm{Pr}_{e}(p) = \mathbb{E}_{X\sim p}[\chi_{e}(X)]$ where $\chi_{e}(X) = \mathds{1}[X_{u} \neq X_{v}]$ is the edge cut indicator.
\begin{lemma}[Small-support approximation of edge cut probabilities through Hoeffding' inequality]
    \label{prop:sampling_hoeffding_proof}
    Draw i.i.d samples $X^{(1)}, \dots, X^{(T)} \sim p$ and define the empirical distribution $\hat{p}(x) = \frac{1}{T} \sum^{T}_{t = 1} \mathds{1}[{X^{(t)}} = x]$, which has support size $\abs{\mathrm{supp}(\hat{p})} \leq T$. Consequently, the cut probability of an edge with respect to the empirical distribution is defined as 
    \begin{align}
        \mathrm{Pr}_e(\hat{p}) = \mathbb{E}_{X\sim \hat{p}}[\chi_e(X)] = \frac{1}{T}\sum^{T}_{t = 1} \chi_e\left(X^{(t)}\right)
    \end{align}
    For any $\epsilon \in (0, 1)$ and $\delta \in (0, 1)$, with probability at least $1 - \delta$ that 
    \begin{align}
        \max_{e \in E} \abs{\mathrm{Pr}_e(\hat{p}) - \mathrm{Pr}_e(p)} \leq \epsilon
    \end{align}
    provided that $T = \mathcal{O}(\epsilon^{-2}(\log(2\abs{E}/\delta)))$. Note that this can be generalized to the fair-cut cover objective, where, given $e' = \arg\min_{e \in E} \mathrm{Pr}_{e}(p)$, then $\abs{\mathrm{Pr}_{e'}(\hat{p}) - \mathrm{Pr}_{e'}(p)} \leq \epsilon$
\end{lemma}
\begin{proof}
    Fix an edge $e \in E$ and define i.i.d bounded random variables $Z_{t} \coloneq \chi_e(X^{(t)}) \in [0, 1]$ where $\mathbb{E}[Z_{t}] = \mathrm{Pr}_e(p)$. Consequently, we have $\mathrm{Pr}_e(\hat{p}) = \frac{1}{T}\sum^{T}_{t = 1}Z_{t}$. By Hoeffding’s inequality for independent Bernoulli random variables $Z_{t}$, we have
    \begin{align*}
    \mathrm{Pr}(\abs{\mathrm{Pr}_{e}(\hat{p}) - \mathrm{Pr}_{e}(p)}\geq \epsilon) &=\mathrm{Pr}\left(\abs{\frac{1}{T}\sum^{T}_{t = 1} Z_{t} - \mathbb{E}[Z_{t}]}\geq \epsilon \right)\\ 
     &= \mathrm{Pr}\left(\abs{\sum^{T}_{t = 1} Z_{t} - \mathbb{E}\Big[\sum^{T}_{t = 1} Z_{t}\Big]}\geq T\epsilon \right)\\
     &\leq 2\exp(-2T\epsilon^2)
    \end{align*}
    Using union bound over all $\abs{E}$ edges, we have:
    \begin{align}
        \mathrm{Pr}\left(\exists e \in E: \abs{\mathrm{Pr}_{e}(\hat{p}) - \mathrm{Pr}_{e}(p)} \geq \epsilon\right) \leq 2\abs{E}\exp(-2T\epsilon^2)
    \end{align}
    Following the above inequality, for $T \geq (1/2\epsilon^2)\log(2\abs{E}/\delta)$, \BBI{we have  $\mathrm{Pr}\left(\exists e \in E: \abs{\mathrm{Pr}_{e}(\hat{p}) - \mathrm{Pr}_{e}(p)} \geq \epsilon\right)$ with probability at most $\delta$ or $\mathrm{Pr}\left(\forall e \in E: \abs{\mathrm{Pr}_{e}(\hat{p}) - \mathrm{Pr}_{e}(p)} \leq \epsilon\right)$ with probability at least $1 - \delta$. This proves the small support claim and the above theorem. }
\end{proof}

\begin{lemma}[Small-support approximation of edge cut probabilities through Multiplicative Chernoff bound]
\label{lem:relative_error}
Draw i.i.d samples $X^{(1)}, \dots, X^{(T)} \sim p$ and define the empirical distribution $\hat{p}(x) = \frac{1}{T} \sum^{T}_{t = 1} \mathds{1}[{X^{(t)}} = x]$ where $x$ is a cut, which has support size $\abs{\mathrm{supp}(\hat{p})} \leq T$. Consequently, the cut probability of an edge with respect to the empirical distribution is defined as 
\begin{align}
    \mathrm{Pr}_e(\hat{p}) = \mathbb{E}_{X\sim \hat{p}}[\chi_e(X)] = \frac{1}{T}\sum^{T}_{t = 1} \chi_e(X^{(t)})
\end{align}
For any $\varphi \in (0, 1)$ and $\delta \in (0, 1)$, with probability at least $1 - \delta$ that 
\begin{align}
    \max_{e \in E} \frac{\abs{\mathrm{Pr}_e(\hat{p}) - \mathrm{Pr}_e(p)}}{\mathrm{Pr}_e(p)} \leq \varphi
\end{align}
    this relative error bound is obtained, provided that $T = \mathcal{O}((\varphi^2\min_e \mathrm{Pr}_e(p))^{-1}\log(2\abs{E}/\delta))$.
\end{lemma}

\begin{proof}
    Fixing edge $e \in E$, using the multiplicative Chernoff bound for relative error on i.i.d bounded random variable $Z_{t} \coloneq \chi_e(X^{(t)}) \in [0, 1]$ where $\mathbb{E}[Z_{t}] = \mathrm{Pr}_e(p)$, we have 
    \begin{align*}
        \mathrm{Pr}(\abs{\mathrm{Pr}_e(\hat{p}) - \mathrm{Pr}_e(p)} \geq \varphi \mathrm{Pr}_e(p)) 
        &= \mathrm{Pr}\left(\abs{\sum^{T}_{t = 1} Z_{t} - \mathbb{E}\Big[\sum^{T}_{t = 1} Z_{t}\Big]}\geq T\varphi \mathrm{Pr}_e(p)\right) \\
        &= \mathrm{Pr}\left(\abs{\sum^{T}_{t = 1} Z_{t} - \mathbb{E}\Big[\sum^{T}_{t = 1} Z_{t}\Big]}\geq \varphi\mathbb{E}\left[\sum^{T}_{t = 1} Z_{t}\right]\right)
        \\
        &\leq 2\exp(-\varphi^2T\mathrm{Pr}_e(p)/3)
    \end{align*}
    Using union bound over all $\abs{E}$ edges, we have the following bound:
    \begin{align}
        \mathrm{Pr}(\exists e \in E: \abs{\mathrm{Pr}_e(\hat{p}) - \mathrm{Pr}_e(p)} \geq \varphi \mathrm{Pr}_e(p)) 
        \leq 2\abs{E}\exp(-\varphi^2T\min_{e}\mathrm{Pr}_e(p)/3)
    \end{align}
    Using $T \geq 3\varphi^{-2}(\min_{e}\mathrm{Pr}_e(p))^{-1}\log(2\abs{E}/\delta)$ or $T = \mathcal{O}((\varphi^2 \min_{e}\mathrm{Pr}_e(p))^{-1} \log(\abs{E}/\delta))$, \BBI{the probability $\mathrm{Pr}(\exists 
    e \in E: \abs{\mathrm{Pr}_e(\hat{p}) - \mathrm{Pr}_e(p)} \geq \varphi \mathrm{Pr}_e(p))$ is at most $\delta$ or the probability $\mathrm{Pr}(\forall 
    e \in E: \abs{\mathrm{Pr}_e(\hat{p}) - \mathrm{Pr}_e(p)} \leq \varphi \mathrm{Pr}_e(p))$ is at least $1- \delta$.}
\end{proof}
\BBI{       
Hoeffding’s inequality provides a uniform additive-error guarantee for any cut distribution, with a sufficient sample size independent of the edge-cut probabilities. Multiplicative Chernoff bounds instead provide relative-error guarantees. For a common relative-error tolerance ($\varphi$) and ($p_{\min}:=\min_{e\in E}\mathrm{Pr}_e(p)>0$), applying Hoeffding’s inequality directly gives a sufficient sample size proportional to ($p_{\min}^{-2}$), whereas the multiplicative Chernoff bound gives dependence proportional to ($p_{\min}^{-1}$) as in Lemma \ref{lem:relative_error}. Consequently, Chernoff provides better dependence when the smallest positive edge-cut probability is small, although neither of these bounds is uniformly tighter for all parameter values. 
}

\section{Proof of Corollary \ref{thm:sufficient_p_covering_distributions}}
\label{appendix:DQAOA_epxressibility}
\begin{definition}[Repeat of Definition~\ref{def:d_QAOA}]
    \label{def:d_QAOA_repeated}
    A \(\mathcal{D}\)-QAOA circuit for a graph \(G = (V, E)\) is defined as a ma-QAOA circuit on G with the following modification:
    \begin{enumerate}
        \item If \(G\) is disconnected or if \(G\) is a path/cycle with \(\lvert V \rvert\geq 4,\) add an ancilla qubit \(A\) and add \(X_{A}\) to the mixer unitary. Choose one vertex \(v\) with maximal degree for each connected component and add \(Z_vZ_{A}\) to the phase separator unitary.
        \item If \(G\) is disconnected and \(\lvert V \rvert \leq 4,\) do (1), then add an additional ancilla \(B\) and add \(Z_{A}Z_{B}\) to the phase separator unitary and add \(X_{B}\) to the mixer unitary.
        \item Let $H$ denote the graph obtained after applying steps (1) and (2). If $H$ is bipartite with bipartition $V(H)=L \sqcup R$, where both $\abs{L}$ and $\abs{R}$ are even, add an ancilla qubit $C$ and connect it to any vertex $v\in V(H)$. Add Pauli $X_C$ to the mixer generators and $Z_vZ_C$ to the phase-separator generators, each with an independent variational parameter in every layer.
    \end{enumerate}
    
\end{definition}

We show that this construction produces a circuit which implicitly represents a connected graph which for \(\lvert V \rvert \geq 4\) is not a cycle or a path or an even-even bipartite graph.

\begin{lemma}
    \label{lm:G_prime_d_QAOA}
    Given a graph \(G = (V,E),\) let \(G^\prime = (V^\prime, E^\prime)\) be the graph representing the \(\mathcal{D}\)-QAOA circuit generated by \(G.\) Then \(G^\prime\) is a connected graph, such that \(\lvert V^\prime\rvert\geq 4\) implies \(G^\prime\) is not a path or a cycle. Moreover, if $G'$ is bipartite, its two bipartition classes cannot both have even cardinality.
\end{lemma}

\begin{proof}
    Trivially, \(G^\prime\) will be connected, as each connected component of a disconnected input will be connected to \(a.\)

    If \(G^\prime\) is a path or a cycle with \(\lvert V \rvert \geq 4,\) then \(a\) has to have degree at most 2. If \(G^\prime\) was produced by (2), then there were at least 2 connected subgraphs in \(G\), so \(a\) will have degree at least 3 due to it's connection to \(b.\)

    If \(G^\prime\) was produced by (1) and \(G\) had at least 3 connected subgraphs, then \(a\) has degree at least 3. If \(G\) had 2 or 1 connected subgraphs, then at least 1 of them has a vertex of degree at least \(2,\) else \(\lvert V \rvert \leq 4\) and \(G^\prime\) would have been produced by (2).

    If \(G^\prime\) is bipartite, its two bipartition classes cannot both have even cardinality. This was prevented by (3) by adding a leaf to exactly one bipartition. It also preserves an existing vertex of degree at least three.
    
    In all cases where \(G^\prime \geq 4,\) \(G^\prime\) contains a vertex of degree at least \(3\) and so cannot be a path or a cycle.
\end{proof}

\label{appendix:proof_sufficient_p}
Given bitstring $x \in \{\pm1\}^{n}$, a $\mathbb{Z}_{2}$-symmetric bitstring distribution $p$ where $\forall x, p(x) = p(- x)$, yield the probability for $m \coloneq 2^{n-1}$ pairs $\{x_j, - x_j\}, j \in [m]$. We can define the symmetric basis vector
\begin{align}
    \ket{e_{j}} \coloneq \frac{\ket{x_j} + \ket{- x_j}}{\sqrt{2}}, \quad j = 1, \dots, m.
\end{align}
which forms an orthonormal basis for the $+1$ eigenspace $\mathcal{H}_{+} \coloneq \{\ket{\psi}:X^{\otimes n}\ket{\psi} = \ket{\psi}\}$. Given any state $\ket{\psi} \in \mathcal{H}_{+}$, we can express it using the following form 
\begin{align}
    \sum^{m}_{j = 1} \alpha_{j} \ket{e_{j}}, \quad \sum^{m}_{j=1} \abs{\alpha}^{2}_{j} = 1
\end{align}
Measuring this state in the standard computational basis yields the representative of a pair with following probability
\begin{align}
    \mathrm{Pr}(x_j) = \mathrm{Pr}(- x_j) = \frac{\abs{\alpha}^{2}_{j}}{2} 
\end{align}
resulting in a $\mathbb{Z}_{2}$-symmetric distribution. Conversely, for any $\mathbb{Z}_{2}$-symmetric bitstring distribution $p$, we can define a non-negative unit vector
\begin{align}
    \alpha'_{j} \coloneq \sqrt{p(x_j) + p(- x_j)}= \sqrt{2p(x_j)}
\end{align}
corresponding the state $\ket{\psi_{p}} \coloneq \sum^{m}_{j = 1} \alpha'_{j} \ket{e_{j}} \in \mathcal{H}_{+}$. Therefore, for any $\mathbb{Z}_{2}$-symmetric bitstring distribution $p$, there exists a non-negative unit vector $\alpha'$ lies on the positive orthant of the real unit sphere $S^{m-1}_{+} \coloneq \{\alpha' \in S^{m-1}: \forall j, \alpha'_{j} \geq 0\}$. Therefore, to prove that QAOA-ansatz can capture the $\mathbb{Z}_{2}$-symmetric bitstring distribution, we only need to prove that every state $\ket{\psi_{p}}$ lies in the reachable orbit of the initial state $\ket{+}^{\otimes n} \in \mathcal{H}_{+}$

\begin{lemma}[Corollary~\ref{thm:sufficient_p_covering_distributions} restate]
For every \(G=(V,E)\) and every $\epsilon \geq 0$, there exists a depth $k^\prime$ such that for every \(k \geq k^\prime\) and every $\mathbb Z_2$-symmetric cut distribution $p$ on $\{\pm1\}^{\abs{V}}$, there exists a $k$-layer \(\mathcal{D}\)-QAOA state $\tilde{\rho}$ such that for any \(x \in \{\pm1\}^{\abs{V}}\),
\begin{align}
    \tilde p(x) &\coloneq \Tr{\ketbra{x}{x}\tilde{\rho}}, & \norm{\tilde p-p}_{\mathrm{TV}} &\leq \epsilon.
\end{align}
\end{lemma}

\begin{proof}
Let $F \coloneq X^{\otimes n}$ be the global bit-flip operator, and let $\mathcal{H}_{\pm}$ be its $\pm 1$ eigenspaces where $\dim_{\mathbb{C}}(\mathcal{H}_{+}) = \dim_{\mathbb{C}}(\mathcal{H}_{-}) = 2^{n-1}$. Since ma-QAOA is initialized in the state $\ket{+}^{\otimes n} \in \mathcal{H}_{+}$ and all of its free generators commute with $F$, the evolution remains entirely inside $\mathcal{H}_+$. 
Therefore, the free Lie algebra (DLA) is restricted to the sector $\mathfrak{g}_{\mathrm{free}}(G)\big|_{\mathcal H_+}$. Here, we follow \cite{kazi2025analyzing} analysis on both the free DLA for different graph structures (Theorem 1 in \cite{kazi2025analyzing}) and its implications (Section IV in \cite{kazi2025analyzing}). We now denote $m \coloneq 2^{n-1}$ and consider these graphs case by cases according to the free-DLA classification.

\medskip
\noindent\textbf{Case 1: Non-cycle odd-odd bipartite graphs.}
The free dynamical Lie algebra of ma-QAOA satisfies
\begin{align}
    \mathfrak{g}_{\mathrm{free}}(G)
    \cong
    \mathfrak{sp}(m/2) \oplus \mathfrak{sp}(m/2),
\end{align}
In this case, the restricted Lie algebra on $\mathcal{H}_{+}$ is $\mathfrak{g}_{\mathrm{free}}(G)\big|_{\mathcal H_+} = \mathfrak{sp}(m/2)$ (pick one symplectic block). Equivalently, the reachable group on $\mathcal{H}_{+}$ is the compact symplectic group $Sp(m/2) \subset U(m)$, acting on the complex space $\mathcal{H}_{+} \cong \mathbb{C}^m$ \cite{gallier2020differential}.
The compact symplectic group acts transitively on the complex unit sphere of $\mathbb{C}^m$. Hence, for every target state $\ket{\psi_p} \in \mathcal{H}_{+}$, there exists a reachable unitary $U_p \in Sp(m/2)$ such that
\begin{align*}
    U_{p}\ket{+}^{\otimes n} = \ket{\psi_{p}}
\end{align*}

\medskip
\noindent\textbf{Case 2, 3: Even-odd bipartite graphs and Archetypal graphs}
In both of these cases, the restricted Lie algebra on $\mathcal{H}_{+}$ is $\mathfrak{g}_{\mathrm{free}}(G)\big|_{\mathcal H_+} =\mathfrak{su}(m)$. Therefore, the connected reachable group on $\mathcal H_+$ is $SU(m)$, which acts transitively on the complex unit sphere in $\mathcal{H}_{+}$ \cite{gallier2020differential}. Thus every target state $\ket{\psi_p} \in \mathcal{H}_{+}$ is reachable from $\ket{+}^{\otimes n}$.

\medskip
\noindent\textbf{Case 4: Path/Cycle graphs and Disconnected graphs}
For path and cycle graphs with graph size $n$, the restricted Lie algebra on $\mathcal{H}_{+}$ is $\mathfrak{so}(2n)$ \cite{kazi2025analyzing}. Let us consider the case when $n \leq 3$. For $n=1$, the case is trivial, as the only $\mathbb Z_2$-symmetric distribution is uniform and is already obtained by measuring the initial state $\ket{+}$. For the case of $n=2$, we have a single edge with the controls $(X_1, X_2)$ and $Z_1Z_2$ which give the restricted dynamical Lie algebra of $\mathfrak{su}(2)$. Similarly, we can calculate the case of $n=3$ for both path graph and cycle graph, which both yield the restricted algebra of $\mathfrak{su}(4)$, which gives access to every $\mathbb Z_2$-symmetric distribution.
\\
When the graph is Path/Cycle with $n \geq 4$ or is disconnected, we modify the ma-QAOA ansatz based on the definition of $\mathcal{D}$-QAOA in Definition \ref{def:d_QAOA}. As proved in Lemma \ref{lm:G_prime_d_QAOA}, $G^{\prime}$ representing $\mathcal{D}$-QAOA ansatz will not be Path/Cycle graphs, so we can use arguments from Cases 1 to 3 to prove that ma-QAOA based on $G^{\prime}$ can capture all $\mathbb{Z}_{2}$-symmetric bitstring distributions

In all cases, every $\mathbb{Z}_{2}$-symmetric bitstring distribution $p$, representing by $\ket{\psi_{p}} \in \mathcal{H}_{+}$ lies in the reachable orbit of the initial state $\ket{+}^{\otimes n}$. Moreover, by the order of generation argument in control theory \cite{lowenthal1971uniform} invoked in the proof of Lemma~1 of \cite{kazi2025analyzing}, there exists a finite integer $k$ such that every reachable unitary on the Lie Group can be implemented by a $k$-layer free ma-QAOA circuit. In particular, this same depth $k$ works uniformly for all target $\mathbb{Z}_{2}$-symmetric bitstring distributions $p$.

Moreover, based on 
By the previous paragraph, for any $\mathbb Z_2$-symmetric cut distribution $p$ on $\{\pm1\}^n$, there exists a $k$-layer ma-QAOA state $\ket{\psi_p}$ such that its measurement distribution satisfies
\begin{align}
    \tilde{p}(x) =
    \Tr{\ketbra{x}{x}\tilde{\rho}}
    =
    \abs{\braket{x}{\psi_p}}^{2}
    =
    p(x)
\end{align}
for all $x\in\{\pm1\}^n$. Therefore,
\begin{align}
    \|\tilde p-p\|_{\mathrm{TV}}=0.
\end{align}
In fact, the above argument yields the stronger conclusion $\epsilon=0$.
\end{proof}

\section{Properties of QAOA ansatz for Fair Cut Cover}
\label{appendix:ma_QAOA_prop}
Now, let us consider the density matrices coming from multi-angle QAOA \cite{herrman2022multi}, where for the cost Hamiltonian of the considered graph $G(V, E, w)$, each edge has its own parameter
\begin{align*}
    U^{\mathrm{multi}}_C(\gamma_p) &= e^{-i\sum_{(uv) \in E}\gamma_{(uv),p}w_{(uv)}Z_{u}Z_{v}}, \\
    U^{\mathrm{multi}}_M(\beta_p) &= e^{-i\sum_{u \in V}\beta_{u, p}X_{u}}
\end{align*}
The multi-angle QAOA yields the density matrix $\rho$ for the original graph $G$, and denotes $\rho_{H} = \Tr_{V\setminus V_H}{\rho}$ as its partial trace on graph $H$. As we are considering multi-angle QAOA with $k$ layers for original graph $G$, we have the set of all parameters  $\Theta^{G}_{\mathrm{multi}} = \{(\gamma, \beta)^{\abs{E} \times \abs{V} \times k}\}$, which corresponds to the set of all possible states $\mathcal{R}_{k}(G) = \{\rho(\gamma, \beta): (\gamma, \beta) \in \Theta^{G}_{\mathrm{multi}}\}$ with the corresponding partial trace set of states $\mathcal{R}^{G \rightarrow H}_{k}(H) = \Tr_{V\setminus V_H}{[\mathcal{R}_{k}(G)]}$. Let define $F_{G}(\rho) \coloneq \min_{e = (u,v ) \in E}(1 - \Tr[Z_uZ_v\rho_H])/2$ as the fair-cut cover value of $\rho$ on graph $G$.
\begin{lemma}
    \label{lm:rho_mono}
     Consider any subgraph $H(V_H, E_H)$ of original graph $G(V, E)$ and partial trace $\rho_{H} = \Tr_{V \setminus V_H}{[\rho]} \in \mathcal{R}^{G \rightarrow H}_{k}(H)$. We have the following inequality,
    \begin{align*}
        F_G(\rho) \leq F_{H}(\rho_{H})
    \end{align*}
\end{lemma}
\begin{proof}
    When considering any edge $(uv) \in E_H$, we have $\Tr[Z_{u}Z_{v}\rho_H] = \Tr[Z_{u}Z_{v}\rho]$ (extended by identity) as observables $Z_{u}Z_{v}$ only act non-trivally on vertices $V_H$. Therefore, we have
    \begin{align}
        F_{H}(\rho_{H}) &= \min_{(uv) \in E_H} \frac{1 - \Tr[Z_{u}Z_{v}\rho_H]}{2} \\
        &= \min_{(uv) \in E_H}\frac{1-\Tr[Z_{u}Z_{v}\rho]}{2} 
        \\ &= F_{H}(\rho)
    \end{align}
    Since $E_H \subseteq E$, $\min_{(uv) \in E} (1 - \Tr[Z_{u}Z_{v}\rho])/2 \leq \min_{(uv) \in E_H}(1-\Tr[Z_{u}Z_{v}\rho])/2$ which proves the Lemma.
\end{proof}

\begin{proposition}[Corollary \ref{cor:sufficient_p_covering_distributions} restate]
    For any graph \(G\), there exists a \(k^\prime\) such that for every \(k \geq k^\prime\) and every subgraph \(H \subseteq G\) we have
    \begin{align}
        \mathrm{Q}_{k}(G) \leq \mathrm{Q}_{k}(H).
    \end{align}
\end{proposition}
\begin{proof}
    For any graph $G$ and subgraph $H \subseteq G$, based on Corollary \ref{thm:sufficient_p_covering_distributions}, there exist $k_G$ layers and $k_H$ layers such that $\mathcal{D}$-QAOA can capture all $\mathbb{Z}_2$-symmetric cut distributions on the graph $G$ and $H$ correspondingly. Let $k^{\prime} = \max\{k_{G}, k_{H}\}$, for any $k \geq k^{\prime}$, then $\forall \rho_{H} \in \mathcal{R}^{G \rightarrow H}_{k}(H)$, there exists $\hat{\rho} \in \mathcal{R}_{k}(H)$ such that
    \begin{align}
        F_{H}(\rho_{H}) = F_{H}(\hat{\rho})
    \end{align}
    Consequently, this gives the monotonicity for $k$-layer $\mathcal{D}$-QAOA
    \begin{align}
         \sup_{\rho \in \mathcal{R}_{k}(G)} F_{G}(\rho) &\leq \sup_{\rho_{H} \in \mathcal{R}^{G \rightarrow H}_{k}} F_{H}(\rho_{H}) \\ &= \sup_{\hat{\rho} \in \mathcal{R}_{k}(H)} F_{H}(\hat{\rho})
    \end{align}
    Which yields $Q_{k}(G) \leq Q_{k}(H)$
\end{proof}

Multi-angle QAOA ansatz solver is pretty interesting, as a different choice of $k$ layers leads to different amounts of expressibility of the circuit. This difference in expressibility can lead to differences in performance and characteristics. For example, when considering random $k$-layer ma-QAOA on $G$, given any subgraph $H \subseteq G$, monotonicity is not assured, or $Q_{k}(G) \leq Q_{k}(H)$ is not guaranteed, because not all reduced density matrix $\rho_{H} \coloneq \Tr_{V\setminus V_H}{\rho}$, where $\rho$ is density matrix yielded by $k$-layer ma-QAOA on original graph $G$, can be obtained from $k$-layer ma-QAOA on original subgraph $H$. Hence, monotonicity can only be obtained when both ansatzs have sufficient number of layers.

\begin{proposition}
\label{prop:std_QAOA_triangle_free_regular}
Let $G=(V,E)$ be a triangle-free $\Delta$-regular graph and let $H=(V_H,E_H)$ be any induced subgraph of $G$ ($V_H \subseteq V$ and $E_H \subseteq E$). Then
\[
\mathrm{Q}^{\mathrm{std}}_{1}(G) \leq \mathrm{Q}^{\mathrm{std}}_{1}(H).
\]
\end{proposition}

\begin{proof}
Since $G$ is triangle-free and $\Delta$-regular, every edge has the same triplet $(\Delta-1,\Delta-1,0)$, hence for $k=1$ standard QAOA, for all edge $(u, v) \in E$, 
\begin{align}
\frac{1-\langle Z_{u} Z_{v}\rangle}{2}
=\frac{1}{2}+\frac{1}{2}\sin(4\beta)\sin(\gamma)\cos^{\Delta-1}(\gamma),
\end{align}
Therefore, the minimum over edges is redundant and $\mathrm{QAOA}^{\mathrm{std}}_{1}(G)$ is defined by 
\begin{equation}
    \begin{aligned}
    &\max_{\beta,\gamma}\left(\frac12+\frac12\sin(4\beta)\sin(\gamma)\cos^{\Delta-1}(\gamma)\right) \\
    & \qquad =\frac12+\frac12\max_{\gamma\in[0,2\pi]} \abs{\sin(\gamma)\cos^{\Delta-1}(\gamma)}
    \end{aligned}
\end{equation}
Using  $\abs{\sin(\gamma)\cos^{\Delta-1}(\gamma)} =\abs{\sin(\gamma)}\abs{\cos(\gamma)}^{\Delta-1}$ and folding $(\abs{\sin(\gamma)},\abs{\cos(\gamma)})$ to a point $(\sin(\theta),\cos(\theta))$ with $\theta \in [0, \pi/2]$, we obtain
\begin{align}
&\max_{\gamma\in[0,2\pi]} \abs{\sin(\gamma)\cos^{\Delta-1}(\gamma)}
\nonumber 
\\ &  \qquad=
\max_{\theta\in[0,\pi/2]} \sin(\theta)\cos^{\Delta-1}(\theta)
\end{align}

Inside the domain $\theta \in[0,\pi/2]$, $\cos(\theta) \in[0,1]$. For any induced subgraph $H\subseteq G$, we have $\deg_H(v)\le \Delta$, hence, $d_{v}^{H} \coloneq \deg_{H}(v)-1\le \Delta-1$. Since $\cos(\theta)^{d}$ is non-increasing in $d$ for $\cos(\theta) \in [0,1]$, for every edge $(u, v)\in E_H$,
\begin{align}
    \cos(\theta)^{d_{u}^{H}}+\cos(\theta)^{d_{v}^{H}} \geq 2\cos(\theta)^{\Delta-1}
\end{align}
Multiplying by $\sin(\theta)\ge 0$ yields
\begin{align}
\sin(\theta)\left(\cos(\theta)^{d_{u}^{H}}+\cos(\theta)^{d_{v}^{H}}\right) \geq 2\sin(\theta)\cos(\theta)^{\Delta-1}
\end{align}
As any subgraph of a triangle-free graph is triangle-free, taking maxima over $\theta\in[0,\pi/2]$ and adding the shift from the QAOA formula implies
$\mathrm{Q}^{\mathrm{std}}_{1}(H) \ge \mathrm{Q}^{\mathrm{std}}_{1}(G)$.
\end{proof}
\begin{remark}
    There is no concrete argument for standard QAOA (except for the case in Proposition \ref{prop:std_QAOA_triangle_free_regular}) because the difference in dependence of the parameter set results in the set of states for subgraph $H$. It is hard to check the relationship between $\mathcal{R}^{\mathrm{std}}_{p} (H)$ and $\mathcal{R}^{\mathrm{std}, G \rightarrow H}_{p} (H)$. Therefore, the reduced density matrix from original graph $G$,  $\rho^{\mathrm{std}}_{H} \notin \mathcal{R}^{\mathrm{std}}_{p} (H)$ and the comparison can not be made.
\end{remark}

\BBI{\section{
Heuristic Classical Algorithms for benchmark}
\label{appendix:heuristic_column_generation}

We describe the classical column-generation heuristic \cite{desaulniers2006column} used in our benchmarks alongside SDP with hyperplane rounding. The method alternates between two tasks: optimizing the probabilities assigned to the currently available cuts, and searching for additional cuts that can improve the Fair Cut Cover objective.

The algorithms start with maintaining a pool of distinct cuts $\mathcal{C} = (x^{(1)}, \dots, x^{(M)})$ with a corresponding probability vectors $p$ such that $\sum_{x\in \mathcal{C}} p_{x} = 1$. This population of cuts represents a cut distribution. Note that, we can solve the small scale LP pretty efficiently with small $M$. Let denote $\bar{\eta}(G; \mathcal{C})$ as the LP for the Fair Cut Cover with the cut pool $\mathcal{C}$,
\begin{equation}
    \label{eq:master_LP}
    \bar{\eta}(G; \mathcal{C}) \coloneq \max_{p \geq \mathbf{0}}\; t \quad \text{s.t.}\quad Y_{\mathcal{C}}p \geq t\mathbf{1}_{E}, \quad \sum_{x \in \mathcal{C}}p_{x} = 1.
\end{equation}
This LP finds the best mixture of the available cuts. It contains $M+1$ variables and $\abs{E}$ coverage constraints, avoiding the explicit enumeration of all possible cuts. However, its achievable value is limited by the current pool. An interesting problem is how can we non-trivially expand the unique cut set $\mathcal{C}$ to improve the performance of our heuristic. 

Following the notation in column generation literature, we call our LP as defined in Equation \eqref{eq:master_LP} as the \textit{restricted master problem}. Column generation \cite{desaulniers2006column} uses dual information to guide the search for additional cuts. We define the dual of Equation \eqref{eq:master_LP} as
\begin{equation}
\label{eq:dual_pricing}
\begin{aligned}
\min_{q\geq\mathbf 0}\; z \quad
\text{s.t.}\quad
&\sum_{e\in E}q_eY_{ex}\leq z
&&\forall x\in \mathcal C,\
&\sum_{e\in E}q_e=1.
\end{aligned}
\end{equation}
Here $q_e$ is the dual edge weight as defined similarly in Equation \eqref{eq:dual_form}. These edge weights can be obtained from the LP solver together with the optimal cut probabilities. Let $p^\star$ and $(q^\star,z^\star)$ be optimal primal and dual solutions. Strong duality gives
\begin{align*}
    z^\star=t^\star=\bar{\eta}(G;\mathcal C)
\end{align*}
while complementary slackness gives
\begin{align*}
    q_e^\star\left(\mathrm{Pr}_e(p^\star)-t^\star\right) = 0,
\qquad\forall e\in E
\end{align*}
Thus, the dual solution identifies a weighted collection of edges that limit the performance of the current pool. We use these weights to define the pricing subproblem
\begin{equation}
\Phi(q^\star)
\coloneq
\max_{x\in\mathcal X}
\sum_{e\in E}q_e^\star Y_{ex},
\end{equation}
where $\mathcal{X}$ denotes the set of all cuts of $G$. This is a weighted MaxCut problem with nonnegative edge weights $q^\star$. Intuitively, it searches for a cut that separates many of the edges emphasized by the dual solution. The criterion for adding a candidate follows directly from the dual constraints of Equation~\eqref{eq:dual_pricing}: every stored cut satisfies

\begin{align*}
    \sum_{e\in E}q_e^\star Y_{ex}\leq t^\star.
\end{align*}
Therefore, a candidate cut $x_{\mathrm{new}}$ satisfying
\begin{equation}
\label{eq:candidate_cut}
    \sum_{e\in E}q_e^\star Y_{e x_{\mathrm{new}}} \geq t^\star+\varepsilon_{\mathrm{price}}
\end{equation}
violates a constraint required by the full dual problem. Equivalently, its probability variable has positive reduced cost, exceeding the numerical tolerance $\varepsilon_{\mathrm{price}}$. We add this cut to $\mathcal{C}$ and re-solve the restricted master, allowing the probabilities of all stored cuts to change. Furthermore, we retain all the existing cuts to ensure that the restricted optimum cannot decrease. However, note that adding a column with positive reduced cost need not immediately increase the objective strictly, because another optimal dual solution may still certify the previous value.

In our heuristic, we search for pricing candidates using a simulated annealing algorithm. Finding any cut satisfying Equation \eqref{eq:candidate_cut} is sufficient for addition; solving the pricing subproblem globally is unnecessary for this step. Conversely, failure of the heuristic search to find such a cut does not certify global optimality. Exact pricing would provide the bound
\[
t^\star\leq\bar{\eta}(G)\leq\Phi(q^\star),
\]
and hence an optimality certificate whenever \(\Phi(q^\star)\leq t^\star+\varepsilon_{\mathrm{price}}\) and tolerance $\varepsilon_{\mathrm{price}}$ reach $0$. With heuristic pricing, the returned distribution instead provides a feasible lower bound on the optimal FCC value.

\paragraph{Simulated Annealing as subproblem solver}
At each column-generation round, we use simulated annealing to search for a cut with a large value of the dual-weighted MaxCut objective,
\begin{align*}
    C_q(x)=\sum_{e\in E}q_eY_{ex},
\end{align*}
where the dual edge weights $q_e$ remain fixed throughout the pricing call. We also fix one anchor vertex in each connected component to zero, thereby eliminating equivalent representations obtained by complementing all bits in a component (getting rid of $\mathbb{Z}_2$ symmetry). Each pricing call starts by generating a random assignment in which every nonanchored bit is independently chosen with equal probability,
\begin{align*}
    \min\left\{1,\exp\left(\frac{\Delta C_q}{T}\right)\right\}.
\end{align*}
Here $\Delta C_q$ is the difference in the pricing subproblem value between the new assignment and the current one, and $T$ is the temperature of the simulated annealing. Thus, moves that do not decrease the objective are always accepted, while decreasing moves may be accepted according to the current temperature. Each pricing call performs a prescribed number of steps, $S$, with the temperature decreasing geometrically from $T_{\mathrm{start}} $ to $T_{\mathrm{end}}$.
One step denotes a new proposed assignment. If the proposed bitstring is accepted, it becomes the current cut; otherwise, the current cut remains unchanged. Each call uses a single annealing trajectory.
Throughout the pricing call, the cut pool remains fixed throughout each annealing call. Separately, we record the highest-scoring candidate among the initial and proposed cuts that are absent from the pool. This comparison includes the initial assignment and every proposed assignment, whether or not the proposal is accepted into the annealing trajectory. After the annealing call finishes, we recompute this best candidate’s score and test the pricing condition to determine whether it should be added to the pool. The candidate is added to the pool only if
\begin{align}
    C_q(x_{\mathrm{new}}) - z > \varepsilon_{\mathrm{price}},
\quad
z=\max_{x\in\mathcal C}C_q(x),
\end{align}
where the pricing tolerance $\varepsilon_{\mathrm{price}}$ accounts for numerical precision. At most one cut is added per round. After an addition, we resolve the restricted master LP to optimize the probabilities of all stored cuts, retaining previous cuts even when their probabilities become zero.

Column generation terminates when the maximum number of pricing rounds or the cut-pool limit is reached, or when a prescribed number of consecutive pricing calls fail to add a cut. Adding a cut resets this failure counter even if the restricted-master objective does not immediately increase. Failure of the annealing search to identify an improving column does not certify global optimality.
\paragraph{Hyperparameter configurations in numerical results.}
For each graph, we performed one column-generation run with random seed $42$. The initial pool contained a number of distinct cuts equal to the number of graph edges and collectively covered every edge. \BBI{This cut pool is generated by making sure that the cut index $j$ cover edge with index $j$ and uniformly and independently selecting the assignment value between $\{\pm 1\}$ for other vertices beside the considering edge.} The maximum pool size was set to $3$ times the number of edges. All stored cuts were retained, including those assigned zero probability, and the algorithm terminated when the pool limit was reached.

Each simulated-annealing pricing call used $100,000$ steps. The temperature decreased geometrically from $1.0$ divided by the number of edges to $0.01$ divided by the number of edges. No internal restarts were performed. We allowed at most $6$ times the number of edges pricing calls and terminated early after $5$ consecutive calls added no cut. A successful addition reset this counter, including when the restricted-master objective remained unchanged.

The pricing tolerance was \texttt{1e-8}. The restricted master LP was solved using SciPy's \texttt{linprog} with the \texttt{HiGHS} solver, with both primal and dual feasibility tolerances set to \texttt{1e-9}. For each graph, we reported the minimum edge-cut probability of the final LP distribution, evaluated directly from its cuts and probabilities.
}

\section{LogSumExp smooth objective function in QAOA}
\label{appendix:LSE}
When considering the Fair Cut Cover objective on the given graph $G(V, E)$, $\bar{\eta}(G) = \max_{p}\min_{e \in E} \mathrm{Pr}_{e}(p)$ is a non-smooth continuous function. For convenience, we first reformulate the objective function as in the dual form using the simplex variable $q \in \Delta(\abs{E})$ (normalized edge weight),
\begin{align}
    \bar{\eta}(G) = \max_{p}\min_{q \in \Delta(\abs{E})} \sum_{e \in E} \mathrm{Pr}_{e}(p)q_{e}
\end{align}
To smooth our objective function, which prevent collapsing the normalized edge weight into one edge, we add the entropy regularization to the objective. We define the following term
\begin{align}
    f_{\tau}(p) = \min_{q \in \Delta(\abs{E})}\left[\sum_{e \in E} \mathrm{Pr}_{e}(p)q_{e} + \tau\sum_{e \in E} q_{e}\log(q_{e})\right]
\end{align}
The regularization term ensures stability under sampling, yielding a unique and smooth inner solution. Interestingly, it is known that the optimal value of the entropy-regularized inner problem is exactly the LogSumExp (LSE) min approximation function \cite{boyd2004convex}.
\begin{align}
    f_{\tau}(p) = -\tau\log{\left(\sum_{e \in E} \exp{-\frac{\mathrm{Pr}_{e}(p)}{\tau}}\right)}
\end{align}
This justifies the use of LSE to smooth our objective, especially in the QAOA variational scheme
\begin{align}
        &\max_{\theta}\min_{(u,v) \in E} \frac{1-\bra{\psi(\theta)}Z_{u}Z_{v}\ket{\psi(\theta)}}{2} \approx -\min_{\theta} \tau\log{\Bigg(\sum_{(u, v) \in E} \exp{\frac{\bra{\psi(\theta)}Z_{u}Z_{v}\ket{\psi(\theta)} - 1}{2\tau}}\Bigg)} 
\end{align}

\begin{remark}[Gradient of the approximation function]
\label{rm:grad_LSE}
Using the chain rule, the gradient of the smooth minimum approximation $f_{\tau}$ can be expressed as the Softmax function \cite{gao2017properties}:
\begin{align}
    \frac{\partial f_{\tau}}{\partial \theta_{k}} &= -\frac{1}2\sum_{(u, v) \in E} p_{(u, v)}(\theta) \frac{\partial \langle Z_{u}Z_{v}\rangle_{\theta}}{\partial \theta_{k}} \\
    p_{(u, v)}(\theta) &= \frac{\exp{\langle Z_{u}Z_{v}\rangle_{\theta}/(2\tau)}}{\sum_{(u', v') \in E} \exp{\langle Z_{u'}Z_{v'}\rangle_{\theta}/(2\tau)}}
\end{align}
This smoothed minimum-approximation gradient can also be computed on a quantum device using methods such as the parameter-shift rule \cite{wierichs2022general}.
\end{remark}

\section{Variance of the gradient of our objective}
\label{appendix:var_grad_obj}
Given our current objective of the parameterized quantum circuit as follows
\begin{align*}
     \max_{\theta} f(\theta) = \max_{\theta}\min_{e\in E} \frac{1 - \bra{\psi(\theta)}Z_{u}Z_{v}\ket{\psi(\theta)}}{2}
\end{align*}
Let denotes $x_{e}(\theta) \coloneq  \bra{\psi(\theta)}Z_{u}Z_{v}\ket{\psi(\theta)}$ and assume the considering graph is archetypal \cite{kazi2025analyzing}.
\paragraph{LogSumExp smooth objective:}
Considering the smooth objective given $\tau > 0$, where the objective becomes 
\begin{align*}
    f_{\tau}(\theta) = 
    -\tau\log{\Bigg(\sum_{e \in E} \exp{-\frac{1-x_{e}(\theta)}{2\tau}}\Bigg)} 
\end{align*}
Following Remark \ref{rm:grad_LSE}, we have the partial derivative of the objective with respect to the $\vartheta$-th parameter $\theta_{\vartheta}$, in $\theta$
\begin{align*}
    \partial_{\vartheta} f_{\tau}(\theta) =  -\frac{1}{2} \sum_{e \in E} p_{e}(\theta) \partial_{\vartheta} x_{e}(\theta)
\end{align*}
where $p_{e}(\theta) = \frac{\exp{x_{e}(\theta)/(2\tau)}}{\sum_{e' \in E} \exp{x_{e'}(\theta)/(2\tau)}}$. From this, we first try to give an upper bound for the variance of $f_{\tau}(\theta)$
\begin{align}
    \mathrm{Var}\left[\partial_{\vartheta} f_{\tau}(\theta)\right] = \mathbb{E}\left[\left(\partial_{\vartheta} f_{\tau}(\theta)\right)^2\right] - \mathbb{E}\left[\partial_{\vartheta} f_{\tau}(\theta)\right]^2
\end{align}
As $p_{e}$ is a probability vector, using Jensen's inequality, we have the following
\begin{equation}
    \begin{aligned}
        \left(\sum_{e \in E} p_{e}(\theta) \partial_{\vartheta} x_{e}(\theta)\right)^{2} &\leq \sum_{e \in E} p_{e}(\theta) \left(\partial_{\vartheta} x_{e}(\theta)\right)^2 \\
        &\leq \sum_{e \in E} \left(\partial_{\vartheta} x_{e}(\theta)\right)^2
    \end{aligned}
\end{equation}
Or, $4\left(\partial_{\vartheta} f_{\tau}(\theta)\right)^{2} \leq \sum_{e \in E} \left(\partial_{\vartheta} x_{e}(\theta)\right)^2$. Taking the expectation value of both sides gives the following.
\begin{align}
    4\mathbb{E}\left[\left(\partial_{\vartheta} f_{\tau}(\theta)\right)^2\right] \leq \sum_{e \in E} \mathbb{E}\left[ \left(\partial_{\vartheta} x_{e}(\theta)\right)^2\right]
\end{align}
Based on \cite{larocca2022diagnosing, kazi2025analyzing}, let $d = 2^{n}$ where $n$ is the number of nodes, and assume that ma-QAOA has enough layers such that the distribution of unitaries is an $\varepsilon$-approximate unitary 2-design, we know that 
\begin{align}
    \mathbb{E}\left[\sum_{e \in E} \partial_{\vartheta} x_{e}(\theta)\right] = 0, \quad
    \mathrm{Var}\left[\sum_{e \in E} \partial_{\vartheta} x_{e}(\theta)\right] = \frac{4d^{2}\abs{E}}{(d^{2}-4)(d+2)}
\end{align}
which gives 
\begin{align}
  \mathbb{E}\left[\left(\partial_{\vartheta} f_{\tau}(\theta)\right)^2\right] &\leq \frac{1}{4}\sum_{e \in E} \mathbb{E}\left[ \left(\partial_{\vartheta} x_{e}(\theta)\right)^2\right] \\
  &= \frac{d^{2}\abs{E}}{(d^{2}-4)(d+2)}
\end{align}
Now, we turn our attention to the mean of the gradient of $f_{\tau}(\theta)$. From the proof of Theorem 2 in \cite{larocca2022diagnosing}, to traceless observable $Z_{u}Z_{v}$ (also commuting with $X^{\otimes n}$) on $\mathcal{H}_{+}$ symmetry block gives  $\mathbb{E}\left[\partial_{\vartheta} x_{e}(\theta)\right] = 0 $. Therefore, the mean of the gradient can be expressed as 
\begin{equation}
  \begin{aligned}
    \mathbb{E}\left[\partial_{\vartheta} f_{\tau}(\theta)\right] &=  -\frac{1}{2} \sum_{e\in E} \mathrm{Cov}\left(p_{e}(\theta), \partial_{\vartheta} x_{e}(\theta)\right)
\end{aligned}  
\end{equation}
This covariance term is tricky, and there is not much we can say given a unitary 2-design. Except when considering the full Haar measure, this covariance yields $0$. Therefore, the variance is bounded by 
\begin{align}
    \mathrm{Var}\left[\partial_{\vartheta} f_{\tau}(\theta)\right] \leq  \frac{d^{2}\abs{E}}{(d^{2}-4)(d+2)} 
\end{align}
Or $\mathrm{Var}\left[\partial_{\vartheta} f_{\tau}(\theta)\right] = \mathcal{O}(n^{2}/2^{n})$ similar to the analysis of linear objective in \cite{kazi2025analyzing}. 

\paragraph{Min objective} 
Given that our objective is the minimum over all edges, which is 
\begin{align}
    f(\theta) = \min_{e\in E} \frac{1 - \bra{\psi(\theta)}Z_{u}Z_{v}\ket{\psi(\theta)}}{2}
\end{align} 
Function $f(\theta) = \min_{e \in E} (1 - x_{e}(\theta))/2$ is a non-smooth function where these non-smooth points are also non-differentiable (i.e, there exists  $theta^{\prime}$ such that $f(\theta^{\prime}) = (1 - x_{e_{i}}(\theta^{\prime}))/2 = (1 - x_{e_{j}}(\theta^{\prime}))/2$). Since $x_{e}(\theta)$ is a smooth function, $f(\theta)$ is the maximum of finitely many smooth functions up to an affine transformation, which makes it locally Lipschitz (almost differentiable anywhere).
Specifically, when point $\theta^{\prime\prime}$ gives a unique edge, its gradient becomes $-\frac{1}{2}\nabla x_{e}(\theta^{\prime\prime})$ \cite{clarke1990optimization}. Based on this, we can bound the variance of those subgradients where these points give a unique edge $e_{*}$. Considering $\vartheta$-th parameter $\theta_{\vartheta}$ in $\theta$,
\begin{align*}
    4\mathrm{Var}\left[\partial^{C}_{\vartheta}f(\theta)\right] &= \left(\mathbb{E}\left[\left(\partial_{\vartheta} x_{e_{*}}(\theta)\right)^{2}\right] - \mathbb{E}\left[\partial_{\vartheta} x_{e_{*}}(\theta)\right]^{2}\right) \\
    &\leq  \mathbb{E}\left[\left(\partial_{\vartheta} x_{e_{*}}(\theta)\right)^{2}\right]\\
    &\leq  \sum_{e\in E}\mathbb{E}\left[\left(\partial_{\vartheta} x_{e}(\theta)\right)^{2}\right]\\
    &= \frac{4d^{2}\abs{E}}{(d^{2}-4)(d+2)}
\end{align*}
\BBI{
\section{A concrete example of applying QAOA on Fractional Coloring}
\label{appendix:fractional_coloring}
Fractional coloring is a well-established problem from the fractional graph theory literature \cite{scheinerman2011fractional}. It generalizes ordinary graph coloring by assigning a set of colors to each node, such that adjacent vertices share no colors. Fractional coloring can be viewed as the linear programming relaxation of traditional graph coloring with the following formulation
\begin{equation}
    \label{eq:fractional_coloring}
    \begin{aligned}
        \chi_{f}(G) \coloneq \min_{x} & \sum_{I \in \mathcal{I}(G)} x_{I} \\ 
        \text{s.t: } & \sum_{I \in \mathcal{I}(G,v)} x_{I} \geq 1 \quad \text{for each vertex }v\\
    \end{aligned}  
\end{equation}
Here $\mathcal{I}$ is the set of all independent set; each independent set (IS) has a non-negative value of $x_{I}$. Interestingly, if we follow the same transformation of Fractional Cut Cover to Fair Cut Cover as in Appendix~\ref{appendix:equiv} through normalizing the weight of IS, $x_{I}$, we end up with a distributional version of Fractional Coloring as follows
\begin{equation}
    \label{eq:fractional_coloring_prob}
    \begin{aligned}
        \bar{\chi}_{f}(G) \coloneq \max_{p}\min_{v\in V} \langle Y_{v, \cdot}, p \rangle = \frac{1}{\chi_{f}(G)}
    \end{aligned}  
\end{equation}
where $p$ is the distribution over all independent set and $Y_{v, I} = \mathds{1}[v \in  I]$ is an indicator matrix of node $v$ belongs to independent set $I$. This formulation has a direct scheduling interpretation. Let each vertex represent an activity, and connect two vertices when their activities cannot be performed simultaneously. An independent set then specifies a feasible set of simultaneous activities. Suppose that, in each equal-duration time slot, we independently sample an independent set from the chosen distribution and activate its vertices. The vertex-inclusion probabilities give each activity's probability of activation and its long-run fraction of active slots. Thus, the objective maximizes the activation fraction of the least-served activity. The desired output is a distribution over feasible schedules: repeatedly selecting a single large independent set can leave other activities unserved.

Here, similar to how we use QAOA on MaxCut to represent cut distribution, we can use a constraint-preserving mixer such as the one from the Quantum Alternating Operator Ansatz \cite{hadfield2019quantum} to represent distribution over independent sets by changing the objective value. We can define the concrete mixer and cost Hamiltonian and the objective as follows:
\begin{equation}
    H_{M} = \sum_{v \in V} X_{v} \prod_{u \in \mathcal{N}(v)} \frac{I + Z_{u}}{2}, \quad H_{C} = \sum_{v \in V} \frac{1 - Z_{v}}{2}, \quad Q(G) \coloneq \max_{\theta}\min_{v\in V}\frac{1-\bra{\psi(\theta)}Z_{v}\ket{\psi(\theta)}}{2}
\end{equation}
where $\ket{\theta} = \mathrm{QAOA}(H_C, H_M)\ket{0}^{n}$ is the resulting state of QAOA for MIS. Furthermore, recent analysis in DLA shows that when using a multi-angle scheme, any distribution over independent sets is reachable within a finite number of layers \cite{tsvelikhovskiy2026expressive} when considering connected graphs. This shows a concrete example of how we can use the QAOA ansatz and our strategy to solve the fractional coloring problem.
}

\BBI{\section{Hardware data, symmetry reduction, and noise models}
\label{appendix:hardware-methods}

\subsection{Experimental details}
\label{sec:exp_details}
\label{hw:data}
The experiments were performed on Quantinuum's $56$-qubit H2-2 and $98$-qubit Helios-1 trapped-ion processors through Quantinuum Nexus. H2-2 uses a racetrack architecture with $^{171}\mathrm{Yb}^{+}$ qubits, whereas Helios uses $^{137}\mathrm{Ba}^{+}$ qubits~\cite{moses2023race, ransford202698}.
The native parameterized gates relevant to these experiments are $R_z(\lambda)$, $R_{xy}(\theta,\phi)$, and $R_{ZZ}(\theta)$. The QAOA mixer rotations satisfy $R_x(\theta)=R_{xy}(\theta,0)$, and the cost interactions are implemented using $R_{ZZ}$ gates. We used the default Nexus compilation workflow based on
\texttt{tket}, which maps operations to the device gate set and
optimizes the circuit to reduce gate counts.

Table~\ref{tab:hardware-data} summarizes the hardware datasets. The H2-2 experiments comprise ten graph instances at each size $n=10,\ldots,15$, all with maximum clique size five, using classically optimized two-layer circuits. The Helios-1 experiment uses standard one-layer QAOA on $K_{60}$ with one shared cost angle and one shared mixer angle. The executed shared cost and mixer angles for the Helios-1 experiment were $\gamma = 6.20678$ and $\beta = 4.59203$ radians, respectively, using the convention in Eq.~\eqref{hw:qaoa-convention}. These parameters were obtained by analytical optimization. The ideal expectation value of FCC at these parameters is $0.5079$.

\begin{table}[!ht]
\centering
\begin{tabular}{llll}

\toprule
Device & Graphs & Instances & Shots per instance \\
\midrule
H2-2 & $n=10,\ldots,15$, clique size $5$ & $10$ per size & $1500$ \\
Helios-1 & $K_{60}$ & $1$ & $5000$ \\
\bottomrule
\end{tabular}
\caption{Hardware datasets used in the revised analysis. Shot counts
refer to recorded hardware measurements per instance.}
\label{tab:hardware-data}
\end{table}

For $K_{60}$, the logical cost layer contains one commuting interaction
for each of the $\binom{60}{2}=1770$ edges, followed by a mixer on all
$60$ qubits. Table~\ref{tab:experiment_resources} reports the
single-qubit and two-qubit gate counts after compilation. The
small-instance entries summarize ten graphs at each size, while
the $K_{60}$ entry corresponds to the single complete-graph instance.

\begin{table}[!ht]
\centering

\begin{tabular}{lll}
\toprule
Graph size & 1-q & 2-q \\
\midrule
$10$ & $30$ & $34.9\pm2.7$ \\
$11$ & $33$ & $38.6\pm4.2$ \\
$12$ & $36$ & $42.4\pm3.2$ \\
$13$ & $39$ & $45.4\pm3.5$ \\
$14$ & $42$ & $47.4\pm2.4$ \\
$15$ & $45$ & $51.2\pm5.1$ \\
$60$ & $120$ & $1770$ \\
\bottomrule
\end{tabular}
\caption{Compiled circuit resources: single-qubit (1-q) and two-qubit
(2-q) gate counts for ten graph instances at each size
$n=10,\ldots,15$ and the single complete-graph instance $K_{60}$.}
\label{tab:experiment_resources}
\end{table}

All empirical quantities are obtained from saved or reference
measurement samples in the computational basis. The H2-2E data provide
a device-emulator comparison for the small instances. Noiseless simulation,
emulator, and hardware data are evaluated at common shot budgets, with
the H2-2 comparison restricted to at most $1500$ shots per instance.
Exact ideal expectation values are shown separately from empirical
sample estimates. The $K_{60}$ comparisons use $5000$ samples per source
for the distribution panels and the displayed sample counts for the
FCC curves.

\subsection{Permutation symmetry of complete-graph QAOA}
\label{hw:symmetry}

For the complete graph $K_n$, with $n\geq 2$ and $m=\binom{n}{2}$ edges,
we use the one-layer, shared-angle QAOA convention
\begin{equation}
    \ket{\psi(\gamma,\beta)}
  =e^{-i\beta H_{M}}e^{-i\gamma H_{C}}\ket{+}^{\otimes n},
  \qquad
  H_C=\sum_{i<j}\frac{I-Z_iZ_j}{2},\qquad H_M=\sum_{i=1}^{n}X_i.
  \label{hw:qaoa-convention}
\end{equation}
Here the cost unitary acts before the mixer. With the definitions
$R_{ZZ}(\theta)=e^{-i\theta Z\otimes Z/2}$ and
$R_X(\theta)=e^{-i\theta X/2}$, these unitaries correspond to
$R_{ZZ}(-\gamma)$ on every edge, up to an overall phase, followed by
$R_X(2\beta)$ on every qubit. The ideal and noise-model references in Figure~\ref{fig:k60-distributions} are generated from the noiseless circuit at this same parameter pair.

Let $U_\pi$ permute the qubits according to any permutation $\pi$.
Completeness of the graph implies $[U_\pi, H_C]=[U_\pi, H_M]=0$, and
$U_\pi\lvert+\rangle^{\otimes n}=\lvert+\rangle^{\otimes n}$.
Consequently, the state in Eq.~\eqref{hw:qaoa-convention} is invariant
under every qubit permutation. Any two bitstrings of the same Hamming
weight are related by such a permutation, so their complex amplitudes
are equal. This proves that the entire evolution remains in the
$(n+1)$-dimensional subspace spanned by the normalized Dicke states
\begin{equation}
    \ket{D_w^{(n)}}
  =\binom{n}{w}^{-1/2}\sum_{b:\,\abs{b}=w}\ket{b},
  \qquad w=0,\dots,n.
  \label{hw:dicke-basis}
\end{equation}
This reduction is a special case of symmetry-based QAOA simulation
\cite{shaydulin2021classical}. In particular, $K_{60}$ requires only
$n + 1 = 61$ amplitudes rather than $2^{60}$ computational-basis amplitudes.
The same permutation argument applies at any QAOA depth for the complete graph family when all
layers use shared angles.

\subsection{Exact sampling of $p=1$ QAOA on complete graph}
\label{hw:reduced-dynamics}

A bitstring of weight $w$ cuts exactly $w(n-w)$ edges. Therefore, based on non-trivial $n + 1$ amplitudes, we can create the diagonal reduced cost matrix, with
$C_{ww}=w(n-w)$. The mixer flips one bit at a time. There are
$\binom{n}{w}(n-w)$ directed single-bit flips from weight $w$ to weight
$w+1$; dividing by the Dicke-state normalization factors gives
\begin{equation}
  M_{w+1,w}
  =\frac{\binom{n}{w}(n-w)}
         {\sqrt{\binom{n}{w}\binom{n}{w+1}}}
  =\sqrt{(w+1)(n-w)}=M_{w,w+1}.
  \label{hw:mixer-elements}
\end{equation}
All other matrix elements vanish, including the diagonal. Thus the
reduced mixer is a real symmetric tridiagonal matrix.

Writing $\lvert\psi\rangle=\sum_w a_w\lvert D_w^{(n)}\rangle$, the
initial amplitudes and final amplitudes are
\begin{equation}
  a_w^{(0)}=\sqrt{\frac{\binom{n}{w}}{2^n}},\qquad
  \boldsymbol a
  =e^{-i\beta M}
    \operatorname{diag}_{w=0}^{n}\!\left(e^{-i\gamma w(n-w)}\right)
    \boldsymbol a^{(0)}.
  \label{hw:reduced-state}
\end{equation}
Evaluating this finite-dimensional matrix exponential gives the ideal distribution to numerical precision, without a full state vector. The phase factors must be retained until after the mixer is applied: the mixing operation combines amplitudes coherently.

Define $q_w=|a_w|^2$. Equation~\eqref{hw:dicke-basis} immediately gives
\begin{equation}
  \Pr(W=w)=q_w,\qquad
  \Pr(X=b)=\frac{q_{\abs{b}}}{\binom{n}{\abs{b}}}.
  \label{hw:exact-sampling}
\end{equation}
An exact draw from the ideal measurement distribution is therefore obtained by first drawing $W$ from $q$, then choosing uniformly a subset of exactly $W$ distinct vertices and assigning ones to those positions. This reproduces the measurement probabilities, rather than the quantum
state itself. 

The global bit flip $X^{\otimes n}$ also commutes with both generators and fixes the initial state. Hence $a_w=a_{n-w}$ and $q_w=q_{n-w}$, providing an additional consistency check. This extra symmetry could reduce the representation further; we retain the full $(n+1)$-dimensional representation for convenience.

\subsection{Final noise channels}
\label{hw:noise-models}

The channels in Fig.~\ref{fig:k60-distributions} act after the ideal
QAOA circuit and before computational-basis measurement. We propagate
their measurement distributions exactly and then draw the reference
samples. This construction does not require a noisy $2^{60}$-dimensional
statevector or density matrix, and it does not model errors inserted
throughout the compiled circuit. Note that we also do not model SPAM (State preparation and measurement) error.

\paragraph{Global depolarization.}
For $0\leq\lambda\leq1$, define
\begin{equation}
  \mathcal D_\lambda(\rho)=(1-\lambda)\rho
       +\lambda\frac{\mathbf{I}}{2^n}.
  \label{hw:global-depolarization}
\end{equation}
This uses the maximally mixed states to replace the usual
depolarizing parameter~\cite{nielsen2010quantum}. Let $b_w=2^{-n}\binom{n}{w}$ be the Hamming-weight distribution of uniform bitstrings. Linearity gives
\begin{equation}
  q_w^{\mathcal{D}_{\lambda}}=(1-\lambda)q_w+\lambda b_w,\qquad
  \mathrm{FCC}_{\mathcal{D}}=\frac12+(1-\lambda)\left(\mathrm{FCC}_0-\frac12\right).
  \label{hw:depolarized-distribution}
\end{equation}
where $\mathrm{FCC}_0$ denotes the Fair Cut Cover value from ideal QAOA and $\mathrm{FCC}_{\mathcal{D}}$ denotes the Fair Cut Cover value of QAOA after going through global depolarization channel. The effect of the Global depolarization therefore makes sampling equivalent to drawing an ideal QAOA bitstring with probability $1-\lambda$ and a uniform bitstring with probability $\lambda$. The mixed component is $\mathbf{I}/2^n$ on the full Hilbert space;
replacing it with a uniform density matrix on the $61$ Dicke states would give the wrong distribution. Global depolarization preserves the symmetry under $w\mapsto n-w$ and moves the FCC toward $1/2$.

\paragraph{Independent amplitude damping.}
We apply the same zero-temperature amplitude-damping channel to every
qubit, with damping probability $0\leq \kappa \leq1$ and Kraus
operators~\cite{nielsen2010quantum}
\begin{equation}
  K_0=
  \begin{pmatrix}1&0\\0&\sqrt{1-\kappa}\end{pmatrix},
  \qquad
  K_1=
  \begin{pmatrix}0&\sqrt{\kappa}\\0&0\end{pmatrix}.
  \label{hw:amplitude-kraus}
\end{equation}
For final computational-basis measurements, each measured one survives
with probability $1-\kappa$, and a zero remains zero. Thus, conditional
on an input weight $w$, the output weight $W$ is distributed as
$\operatorname{Binomial}(w,1-\kappa)$, giving
\begin{equation}
  q_v^{\mathcal{A}_{\kappa}}
  =\sum_{w=v}^{n}q_w\binom{w}{v}(1-\kappa)^v\kappa^{w-v},
  \qquad v=0,\ldots,n.
  \label{hw:amplitude-distribution}
\end{equation}
This measurement rule is exact for a final damping channel and does not account for damping at intermediate circuit locations. 

The ideal QAOA distribution has unbiased single-qubit marginals. Let $\mathrm{Pr}_{e}(\rho)$ denotes the edge cut probability of from quantum state $\rho$. Then a cut remains cut with
probability $1- \kappa$. An edge with bitstring assignment of $11$, becomes a cut with probability $2\kappa(1-\kappa)$. Bitstring $00$ is not considered as amplitude damping does not affects $\ket{00}$ state. Consequently,
\begin{equation*}
    \begin{aligned}
        \mathrm{Pr}_{e}(\mathcal{A}_{\kappa}(\rho)) &= \mathrm{Pr}_{e}(\rho)(1-\kappa) + \frac{1 - \mathrm{Pr}_{e}(\rho)}{2}2\kappa(1-\kappa) \\
         &= \mathrm{Pr}_{e}(\rho)(1-\kappa)^2 +  \kappa(1-\kappa)
    \end{aligned}
\end{equation*}
As $(1-\kappa)^2 \geq 0$ and the additive term is identical for every edge, the Fair Cut Cover value after the amplitude damping becomes
\begin{equation}
  \mathrm{FCC}_{\mathcal A}= \min_{e}\left[\mathrm{Pr}_{e}(\rho)(1-\kappa)^2 +  \kappa(1-\kappa)\right] = (1-\kappa)^2 \mathrm{FCC}_0+\kappa(1-\kappa).
  \label{hw:amplitude-fcc}
\end{equation}
where $\mathrm{FCC}_0$ denotes the Fair Cut Cover value from ideal QAOA. Also, the expectation value of the Hamming weight of the measured bitstring after amplitude damping is $\mathbb{E}[\mathcal{A}_{\kappa}(W)]=(1-\kappa)n/2$, where $\mathbb{E}[W]=n/2$ is the noiseless expectation value of the Hamming weight, so damping breaks the complement symmetry and shifts the weight distribution toward zero.

\paragraph{Depolarization followed by damping.}
The combined model applies $\mathcal D_\lambda$ first and then
$\mathcal A_\kappa^{\otimes n}$. Its weight distribution is
\begin{equation}
  q_v^{\mathcal{DA}}
  =\sum_{w=v}^{n}
   \left[(1-\lambda)q_w+\lambda b_w\right]
   \binom{w}{v}(1-\kappa)^v\kappa^{w-v},
  \label{hw:combined-distribution}
\end{equation}
and its FCC is
\begin{equation}
  \mathrm{FCC}_{\mathcal{DA}}
  =(1-\kappa)^2
     \left[\frac12+(1-\lambda)\left(\mathrm{FCC}_0-\frac12\right)\right]
     +\kappa(1-\kappa).
  \label{hw:combined-fcc}
\end{equation}
The channel order matters because amplitude damping is not unital.
The curves use $\lambda\in\{0.25,0.75\}$ for depolarization alone,
$\kappa=0.25$ for damping alone, and $(\lambda,\kappa)=(0.5,0.25)$ for the
combined channel.

All these models preserve invariance under permutations of vertices,
so all population edge-cut probabilities remain equal. After obtaining
$q^{\mathcal D}$, $q^{\mathcal A}$, or $q^{\mathcal{DA}}$, one can sample a
weight and then choose a uniformly random subset of that size,
exactly as in Eq.~\eqref{hw:exact-sampling}. This relies on
exchangeability of the measurement distribution, not on a
$61$-dimensional pure-state description of the noisy system.}

\bibliographystyle{unsrt}
\bibliography{nat-com/sn-bibliography}
\end{document}